\documentclass[11pt,a4paper]{article}
\usepackage[T1]{fontenc}
\usepackage{lmodern}
\usepackage[margin=25mm,top=23mm,bottom=26mm]{geometry}
\usepackage{amsmath,amssymb,amsthm,mathtools,bm,mathrsfs}
\usepackage{booktabs,longtable,array,tabularx,multirow,caption,float,graphicx}
\usepackage{xcolor,enumitem,microtype}
\usepackage[numbers,square,sort&compress]{natbib}
\usepackage{tikz}
\usetikzlibrary{arrows.meta,positioning,shapes.geometric,calc}
\usepackage[colorlinks=true,linkcolor=blue!43!black,citecolor=blue!43!black,urlcolor=blue!43!black,bookmarksnumbered=true]{hyperref}
\hypersetup{pdftitle={Quantum Extremal Numbers from Marginal Compatibility},pdfauthor={Wanchen Zhang, Zicheng Han, Xiande Zhang},pdflang={en}}
\providecommand{\doi}[1]{\href{https://doi.org/#1}{doi:\nolinkurl{#1}}}
\newcommand{\ket}[1]{\lvert #1\rangle}
\newcommand{\bra}[1]{\langle #1\rvert}
\newcommand{\Tr}{\operatorname{Tr}}
\newcommand{\supp}{\operatorname{supp}}
\newcommand{\wt}{\operatorname{wt}}
\newcommand{\Sh}{\operatorname{Sh}}

\newcommand{\Qex}{Q_{\mathrm{ex}}}
\newcommand{\Qrob}[1]{Q_{\mathrm{ex},#1}^{D}}

\newtheorem{theorem}{Theorem}[section]
\newtheorem{lemma}[theorem]{Lemma}
\newtheorem{proposition}[theorem]{Proposition}
\newtheorem{corollary}[theorem]{Corollary}
\theoremstyle{definition}

\theoremstyle{remark}

\setlist{itemsep=2pt,topsep=3pt}
\allowdisplaybreaks[2]
\title{\bfseries Quantum Extremal Numbers from Marginal Compatibility}
\author{Wanchen Zhang$^{1}$\thanks{\texttt{wanchenz@mail.ustc.edu.cn}}\qquad
Zicheng Han$^{2}$\thanks{\texttt{hzcqj2020@mail.ustc.edu.cn}}\qquad
Xiande Zhang$^{2,1}$\thanks{Corresponding author. \texttt{drzhangx@ustc.edu.cn}}}
\date{\small $^1$Hefei National Laboratory, University of Science and Technology of China, Hefei 230088, China\\
$^2$School of Mathematical Sciences, University of Science and Technology of China, Hefei 230026, China}
\hypersetup{bookmarksdepth=3}
\numberwithin{equation}{section}
\begin{document}
\maketitle
\begin{abstract}
How many half-system marginals of a single pure state can be maximally mixed at the same time? We also ask how the maximum changes when each counted marginal is allowed a prescribed normalized purity defect. To bound the count, we compare two bounds on the same sum of purity defects. Shadow inequalities constrain the defects within a selected subsystem or across the full half-system family. Assuming that one more marginal meets the tolerance gives an opposing bound through subset counting and, when needed, compatibility among marginals of the same pure state. With explicit attaining states, we obtain the exact maxima $112$, $200$, $396$, and $720$ at nine through twelve qubits, each on an explicit interval containing zero and positive tolerances. At zero tolerance, simultaneous maximal mixing of extensions of a fixed subset imposes quantum error-detection conditions on a common Schmidt support. The resulting hypergraph bound and a family of quaternary code states determine the limiting even-size density $\prod_{j\ge1}(1-4^{-j})\simeq0.68854$. For odd sizes the limit inferior is at least this constant and the limit superior is at most $0.82411$. We compare the density and incidence bounds at finite sizes and derive consequences for simultaneous extremizers, isometric tensor flattenings, and pure secret-sharing encodings.
\end{abstract}

\begingroup\small\tableofcontents\endgroup
\clearpage
\section{Counting maximally mixed marginals}\label{sec:introduction}

A bipartition of a pure state has the largest possible entanglement entropy when the smaller subsystem is maximally mixed. An absolutely maximally entangled (AME) state has this property across every bipartition~\cite{Scott2004,Helwig2012}. For qubits the requirements on different cuts can conflict. Already at four qubits one can make four of the six two-qubit marginals maximally mixed, but no pure state makes all six maximally mixed~\cite{Higuchi2000,Zhang2025}. How many cuts can one state satisfy, and how much can their purities vary before that number increases?

The quantum extremal number $\Qex(n,k)$ counts the largest family of maximally mixed $k$-qubit marginals of a single pure state~\cite{Zhang2025}. A $k$-uniform state attains the full count $\binom nk$~\cite{Shi2025}. We take $k=\lfloor n/2\rfloor$ and count labeled subsets, so a balanced cut of an even system contributes both complementary labels. We also count marginals with normalized purity defect at most $\varepsilon$ and denote the maximum by $\Qrob{\varepsilon}(n,k)$. Among sizes without an AME state, the counts at four, seven and eight qubits were known, while the available bounds at nine and ten qubits were $112\le\Qex(9,4)\le120$ and $200\le\Qex(10,5)\le240$~\cite{Higuchi2000,Huber2017,Zhang2025}. Table~\ref{tab:main} gives our finite-size results.

\begin{table}[htbp]
\centering\small
\caption{Half-system extrema for pure qubit states. On each interval shown, the maximum equals its zero-tolerance value $q$. For twelve qubits, $\eta_{12}=2^{-2^{12943}}$. The upper endpoint is known to be sharp at four qubits.}\label{tab:main}
\setlength{\tabcolsep}{5pt}
\begin{tabular}{@{}rrrrlcl@{}}\toprule
Qubits $n$ & Marginal size $k$ & Labels $\binom nk$ & $q$ & Proved interval & Method & Proof\\\midrule
4 & 2 & 6 & 4 & $[0,1/3)$ & I & Sec.~\ref{sec:four-example}\\
7 & 3 & 35 & 32 & $[0,62/2751]$ & II & Sec.~\ref{sec:seven}\\
8 & 4 & 70 & 56 & $[0,1/5)$ & I & Sec.~\ref{ch:4m}\\
9 & 4 & 126 & 112 & $[0,1/127]$ & III & Sec.~\ref{sec:nine}\\
10 & 5 & 252 & 200 & $[0,1/16)$ & II & Sec.~\ref{sec:ten}\\
11 & 5 & 462 & 396 & $[0,2^{-12}]$ & III & Sec.~\ref{sec:eleven}\\
12 & 6 & 924 & 720 & $[0,\eta_{12}]$ & II & Sec.~\ref{sec:twelve}\\
\bottomrule\end{tabular}
\end{table}

The distribution of entanglement across cuts matters here. The mean linear entropy averages the purities of all $k$-qubit marginals~\cite{Facchi2008}. A count instead asks how many individual purities reach a specified value. At twelve qubits, for example, a state maximizing the mean six-body linear entropy has at most 690 maximally mixed six-body marginals, whereas another state attains 720 (Sec.~\ref{sec:twelve-mean}). 

For the finite counts in Table~\ref{tab:main}, the upper-bound proof asks what would happen if one more marginal were good. The shadow of a selected marginal, or a sum of shadows over the whole family, constrains a sum of defects. The assumed count restricts how much defect the same family can carry. One counts the good and bad labels and uses compatibility wherever different marginals must be considered together. An incompatibility between the two bounds excludes the proposed count; an explicit state attains the remaining value. Section~\ref{ch:method} gives this argument and distinguishes its three forms according to the scope of the shadow and the constraints needed in the counting step.

At large $n$ the natural quantity is the density $g_n=\Qex(n,k)/\binom nk$. The local counting bounds in Sections~\ref{ch:4m} and~\ref{sec:general} apply throughout explicit tolerance intervals. At zero tolerance, simultaneous maximal mixing of extensions of a fixed set of sites requires mutually orthogonal error spaces on a common Schmidt support. This compatibility condition gives an upper bound on the density of extensions, while quaternary code states attain the resulting limit at even sizes. The finite version of the density bound already improves on the general local counting bound from $n=20,27,42$ onward within the families $n\equiv0\pmod4$, odd $n$, and $n\equiv2\pmod4$, respectively (Sec.~\ref{sec:bound-crossover}). For every fixed tolerance $0<\varepsilon<1$, the even-size density of good marginals also stays strictly below one in the limit superior (Sec.~\ref{sec:asym-robust}).

Finally, a state attaining the half-system extremum need not attain the extremum for every smaller subsystem. Section~\ref{sec:structure} tests whether these requirements are compatible. Section~\ref{sec:applications} translates the count into numbers of isometric tensor flattenings and completely depolarized coalitions in pure secret-sharing encodings~\cite{Goyeneche2015,Pastawski2015,Cleve1999,Helwig2012,LiuShao2025}.

\section{Purity defects and counting constraints}\label{ch:method}\label{app:preliminaries}
\subsection{From a purity defect to an extremal count}
Label the qubits of a pure state $\rho=\ket\psi\bra\psi$ by $[n]=\{1,\ldots,n\}$. For $A\subseteq[n]$, retain the qubits labeled by $A$ and trace out the rest. The resulting $|A|$-body marginal $\rho_A$ and its defect $D_A$ satisfy
\begin{equation}\label{eq:def}
D_A=2^{|A|}\Tr\rho_A^2-1=2^{|A|}\|\rho_A-I_A/2^{|A|}\|_2^2,\qquad \rho_A=\Tr_{A^c}\rho.
\end{equation}
Here $\|X\|_2^2=\Tr(X^\dagger X)$, so $D_A=0$ exactly when $\rho_A$ is maximally mixed. At tolerance $\varepsilon\ge0$, call $\rho_A$ a good marginal if $D_A\le\varepsilon$ and a bad marginal if $D_A>\varepsilon$. At zero tolerance the bad marginals are precisely those with positive defect. In counting arguments the label $A$ inherits the same designation. For fixed $k$, write $b_\varepsilon$ for the number of bad $k$-body marginals, or $b$ when the tolerance is clear. The maximum number of good $k$-body marginals for $1\le k\le\lfloor n/2\rfloor$ is
\begin{equation}\label{eq:qex}
\Qrob{\varepsilon}(n,k)=\max_{\ket\psi}\#\{A\in\tbinom{[n]}k:D_A\le\varepsilon\},\qquad
\Qex(n,k)=\Qrob{0}(n,k).
\end{equation}
A state attaining the full zero-tolerance count $\binom nk$ is called $k$-uniform~\cite{Shi2025}. On a balanced cut the two complementary marginals are counted separately. For $Q=\Qex(n,k)<\binom nk$, order the $k$-body defects of $\ket\psi$ as $D_{(1)}(\psi)\le\cdots\le D_{(\binom nk)}(\psi)$. The optimal stability threshold is
\begin{equation}\label{eq:critical-defect}
\varepsilon_*(n,k)=\min_{\ket\psi}D_{(Q+1)}(\psi).
\end{equation}
Continuity of the ordered defects and compactness of the pure-state space ensure that the minimum is attained. It is positive, since a zero minimum would give a state with at least $Q+1$ zero defects. The stability plateau is therefore $0\le\varepsilon<\varepsilon_*(n,k)$, with $\Qrob{\varepsilon_*}(n,k)\ge Q+1$ at the endpoint. The local sum $S_T=\sum_{A\in\binom Tk}D_A$ runs over the $k$-body marginals of $\rho_T$, and the global sum $S_k=\sum_{|A|=k}D_A$ runs over all $k$-body marginals of $\rho$.

\subsection{Two bounds on the same defect sum}
To prove an upper bound $q$, suppose that a pure state has $g_\varepsilon(\psi)\ge q+1$ good $k$-sets. The two sides of Fig.~\ref{fig:bound-comparison} constrain the same quantity. On the quantum side, shadow inequalities and complementary purities constrain a defect sum $S_{\mathcal F}=\sum_{A\in\mathcal F}D_A$. On the counting side, each good set contributes at most $\varepsilon$. Incidence counts determine which labels can be bad; where necessary, relations among the marginals bound how much defect the remaining labels can carry. For instance, choose a collection $\mathcal C$ of families of $k$-sets and establish
\begin{equation}\label{eq:main-proof-principle}
 \begin{aligned}
 S_{\mathcal F}&\ge L &&\text{for every pure state and every }\mathcal F\in\mathcal C,\\
 g_\varepsilon(\psi)\ge q+1&\ \Longrightarrow\ S_{\mathcal F}\le U(q+1,\varepsilon)
                  &&\text{for some }\mathcal F\in\mathcal C\text{ in that state}.
 \end{aligned}
\end{equation}
The first inequality applies in particular to the family selected in the second. If $U<L$, the assumed state cannot exist. If $U=L$, the two equality conditions must hold in the same state and can still give a contradiction. A state with $q$ zero-defect marginals then attains the count throughout the tolerance interval on which the upper bound holds.
\begin{figure}[htbp]
\centering
\begingroup
\definecolor{ink}{HTML}{1F2933}
\definecolor{slate}{HTML}{405C72}
\definecolor{teal}{HTML}{146B70}
\newcommand{\FigSum}{Two bounds on the same defect sum}
\newcommand{\FigShadow}{Local or global shadow}
\newcommand{\FigLowerA}{Every admissible family}
\newcommand{\FigLowerB}{in every pure state}
\newcommand{\FigHypothesis}{Assume $g_\varepsilon\ge M$}
\newcommand{\FigUpperA}{Incidences and, if needed,}
\newcommand{\FigUpperB}{marginal compatibility}
\sffamily
\setlength{\fboxsep}{2pt}
\resizebox{.62\linewidth}{!}{\raisebox{0pt}[\dimexpr\height+2pt\relax][\dimexpr\depth+2pt\relax]{\hspace{2pt}%
\begin{tikzpicture}[
  x=1pt,y=1pt,font=\fontsize{8.3}{9.8}\selectfont,
  heading/.style={text=ink,align=center,inner sep=0pt},
  box/.style={align=center,text=ink,inner sep=2pt,minimum height=49pt,
    rounded corners=2pt},
  connector/.style={draw=ink!66,line width=.55pt},
  arrow/.style={connector,-{Latex[length=3.5pt,width=4pt]}}
]
\path[use as bounding box] (0,0) rectangle (234,91);
\node[heading] at (117,85) {\textbf{\FigSum} $S_{\mathcal F}$};
\node[box,draw=slate!78,fill=slate!7,text width=105pt] (lower) at (55.5,52.8)
  {\textbf{\FigShadow}\\[2pt]
   \FigLowerA\\\FigLowerB\\[2pt]$S_{\mathcal F}\ge L$};
\node[box,draw=teal!85,fill=teal!7,text width=110pt] (upper) at (176.5,52.8)
  {\textbf{\FigHypothesis}\\[2pt]
   \FigUpperA\\\FigUpperB\\[2pt]$S_{\mathcal F}\le U(M,\varepsilon)$};
\draw[connector] (lower.south) -- (55.5,22) -- (176.5,22) -- (upper.south);
\node[heading] (test) at (117,5.5) {$U(M,\varepsilon)<L\ \Longrightarrow\ g_\varepsilon<M$};
\draw[arrow] (117,22) -- (test.north);
\end{tikzpicture}%
\hspace{2pt}}}%
\endgroup
\caption{A typical comparison of two bounds on the same defect sum. The lower bound applies to every admissible family. If at least $M=q+1$ marginals were good, counting would produce a family whose sum is at most $U$. Comparing the two bounds on that family excludes $M$ when $U<L$. At seven and ten qubits a local shadow contributes to the conditional upper bound.}
\label{fig:bound-comparison}
\end{figure}

The choice of $\mathcal F$ depends on what the shadow constrains. Its members may be the $k$-sets within one subsystem or the entire half-system layer; the eleven-qubit argument also uses weighted sums. Within this comparison, Sections~\ref{sec:local-combinatorial}--\ref{sec:global-shadow} distinguish the following three cases.

\paragraph{I. Local shadow inequalities with combinatorial counting.}
A shadow inequality for a marginal $\rho_T$, with $T\subsetneq[n]$, forbids a configuration of good labels inside $T$. Once this obstruction is known, one needs only the incidences between $T$ and its subsets. At eight qubits the local shadow gives $S_T\ge1$ for every five-set $T$. For $\varepsilon<1/5$, if all five of its four-sets were good, then $S_T\le5\varepsilon<1$. Thus the 56 five-sets must each contain a bad four-set; each four-set covers four of them, so at least 14 four-sets are bad. Complementary spectra enter the derivation of the local inequality; the subsequent count applies to any label family satisfying the obstruction (Sec.~\ref{sec:local-combinatorial}).

\paragraph{II. Local shadow inequalities with global compatibility counting.}
The proof again selects a proper marginal, but now the assignment of defects on other, overlapping or complementary sets must be consistent with the same pure state. At seven qubits the six-site shadow enters the conditional upper bound. If 33 triples were good, the shadow and a cover by good triples would give $S_3\le1+42\varepsilon$. Lewis's independent bound $S_3\ge255/131$ supplies the other side; its equality conditions exclude the endpoint. At ten qubits, if a seven-set $T$ contains at most three bad five-sets and all of its four-sets are good, the five-site shadows give $S_T\le3+48\varepsilon$. Deletion and complementary purities give $S_T\ge6$. For $\varepsilon<1/16$ this excludes the local configuration. The count also uses propagation from bad four-sets to their five-site extensions and closure of the bad five-sets under complementation. At twelve qubits the zero-tolerance local shadow forces the bad six-sets to cover every five-set. Complementary six-sets have equal defect. If exactly one four-set has positive defect, overlap constraints force 42 bad complementary pairs and a weighted count of the uncovered five-sets forces at least 60 more. Four-uniform states satisfy further constraints on the defects of their five- and six-body marginals (Sec.~\ref{sec:twelve}).

\paragraph{III. Global shadow inequalities.}
The shadow bound applies to a sum taken across the entire layer of $k=\lfloor n/2\rfloor$ marginals. At nine qubits, summing the shadows of all 36 seven-site marginals yields a lower bound on $S_4$ over all 126 four-sets. A hypothetical extra good marginal gives an incompatible upper bound on that same sum. At eleven qubits, if 397 five-sets were good at the stated tolerance, the global lower bound on $S_5$ over all 462 five-sets would force a four-site Pauli support to have at least three five-site extensions of large defect. When the largest such extension count is three or four, averaging over triples of extension sites yields the contradiction. From five extensions onward, disjointness forces triangles in graphs with at most seven edges in total (Sec.~\ref{sec:eleven}).

\subsection{Pauli supports and complementary spectra}
Let $\mathcal P_n=\{I,X,Y,Z\}^{\otimes n}$ denote the unsigned Pauli strings. The support $\supp(P)$ consists of sites where $P$ acts nontrivially, and $\wt(P)=|\supp(P)|$. The Pauli expansion is
\begin{equation}
\rho=2^{-n}\sum_{P\in\mathcal P_n}r_P P,
\qquad r_P=\Tr(\rho P)\in\mathbb R,\quad r_I=1.
\end{equation}
Pauli orthogonality $\Tr(PQ)=2^n\delta_{P,Q}$, together with the fact that partial trace removes strings acting outside $A$, yields
\begin{equation}\label{eq:Pauli-defect}
D_A=\sum_{P\ne I:\supp(P)\subseteq A}r_P^2.
\end{equation}
Write $\pi_X=\Tr\rho_X^2$, with $\pi_\varnothing=1$, and let $w_V=\sum_{P:\supp(P)=V}r_P^2$ be the squared Pauli mass with support exactly $V$, including $w_\varnothing=1$. Grouping by support and applying subset inversion gives
\begin{equation}\label{eq:pre-support-mass}
 D_A=\sum_{\varnothing\ne V\subseteq A}w_V,\qquad
 w_V=\sum_{X\subseteq V}(-1)^{|V|-|X|}2^{|X|}\pi_X\ge0 .
\end{equation}
Since each $w_V$ is nonnegative, $R\subseteq S$ implies
\begin{equation}\label{eq:monotonicity}
0\le D_R\le D_S.
\end{equation}
Across $A|A^c$, a pure state admits a Schmidt decomposition $\ket\psi=\sum_\alpha\sqrt{\lambda_\alpha}\ket{u_\alpha}_A\ket{v_\alpha}_{A^c}$ with nonnegative coefficients $\lambda_\alpha$ summing to one. The reduced states $\rho_A$ and $\rho_{A^c}$ therefore have the same \emph{nonzero eigenvalues}, even when their dimensions differ. In particular, their purities are equal. Restoring the dimension factors in $D_A$ gives
\begin{equation}\label{eq:complement}
\boxed{D_A+1=2^{2|A|-n}(D_{A^c}+1).}
\end{equation}
For $n=2k$, $D_A=D_{A^c}$ on balanced cuts. For $n=2k+1$ and $|A|=k$, $D_{A^c}=1+2D_A$. These complementary-spectrum identities apply to the global pure states in Eq.~\eqref{eq:qex}.
For disjoint label sets $A,B$, monotonicity with $A\subseteq B^c$ and complementarity on $B|B^c$ give
\begin{equation}\label{eq:pre-disjoint}
 D_A+1\le 2^{n-2|B|}(1+D_B).
\end{equation}

\subsection{Local shadow inequalities and deletion identities}
Let $\omega$ be any density matrix on a set $T$ of $s$ qubits. In the applications, $\omega=\rho_T$ is a marginal of a global pure state. For $X\subseteq T$ let $\omega_X=\Tr_{T\setminus X}\omega$ and $\pi_X=\Tr\omega_X^2$, with $\pi_\varnothing=1$. The weight-$j$ Pauli mass of $\omega$ is
\[
 a_j(T)=\sum_{\substack{P\in\mathcal P_s\\\wt(P)=j}}[\Tr(\omega P)]^2,
 \qquad a_0(T)=1,\qquad D_T=\sum_{j=1}^s a_j(T).
\]
On two copies of $\omega$, let $V_i$ swap the copies of qubit $i$. The expectation of $\prod_{i\in X}V_i$ is $\pi_X$. For $J\subseteq T$, expanding the product below gives the two-copy shadow inequality~\cite{Rains1999,Huber2018,EltschkaHuber2018,Eltschka2018}
\begin{equation}\label{eq:pre-swap-shadow}
 \sum_{X\subseteq T}(-1)^{|X\cap J|}\pi_X
 =\Tr\!\left[\omega^{\otimes2}\prod_{i\in J}(I-V_i)\prod_{i\in T\setminus J}(I+V_i)\right]\ge0 .
\end{equation}
The factors commute and are positive semidefinite. For $J=T$, write $h_T$ for the alternating sum. The spin flip $\Theta_s(\omega)=Y^{\otimes s}\omega^{\mathsf T}Y^{\otimes s}$, with transposition in the computational basis, gives the equivalent scalar forms
\begin{align}
 h_T&:=\sum_{X\subseteq T}(-1)^{|X|}\pi_X
     =\Tr[\omega\Theta_s(\omega)]\ge0,\label{eq:pre-shadow-purity}\\
 2^s h_T&=\sum_{j=0}^s(-1)^j a_j(T)\ge0.\label{eq:shadow0}
\end{align}
Indeed, the squared expectation of a weight-$j$ Pauli string with support $V$ enters the alternating purity sum with coefficient $\sum_{X\supseteq V}(-1)^{|X|}2^{-|X|}=(-1)^j2^{-s}$, where $X\subseteq T$. The spin flip acts on the same string by $(-1)^j$. If $E_T=\sum_{\ell=1}^{\lfloor s/2\rfloor}a_{2\ell}(T)$, separating the even and odd weights gives
\begin{equation}\label{eq:pre-shadow-even}
 D_T\le 1+2 E_T .
\end{equation}
To sum the defects after tracing out $q$ qubits of $T$, count how often each Pauli support survives. A support of size $j$ survives $\binom{s-j}{q}$ choices, with the binomial coefficient zero for $q>s-j$. Thus for $0\le q\le s$,
\begin{equation}\label{eq:pre-deletion}
 M_q(T):=\sum_{B\in\binom T{s-q}}D_B
          =\sum_{j=1}^{s}\binom{s-j}{q}a_j(T).
\end{equation}
In particular $M_0(T)=D_T$. For $s\ge2$, a support missing $d$ sites has coefficient $\binom d2-d+1=(d-1)(d-2)/2\ge0$ in the following combination, since $d$ is a nonnegative integer. Hence
\begin{equation}\label{eq:pre-delete-two}
 M_2(T)-M_1(T)+D_T\ge0 .
\end{equation}

\subsection{Global weight enumerators and mean linear entropy}
For the full $n$-qubit system, define
\[
W_j:=\sum_{\substack{P\in\mathcal P_n\\\wt(P)=j}}r_P^2,\qquad W_0=1.
\]
Each support of size $j$ lies in $\binom{n-j}{k-j}$ different $k$-sets, so
\begin{equation}\label{eq:Sk}
S_k:=\sum_{|A|=k}D_A=\sum_{j=1}^k\binom{n-j}{k-j}W_j.
\end{equation}
Use the convention $L(\rho_A)=1-\Tr\rho_A^2$ for linear entropy. The mean over all $k$-body marginals is $\bar L_k=\binom nk^{-1}\sum_{|A|=k}L(\rho_A)$, and their mean purity is $\bar\pi_k=1-\bar L_k$. Thus
\begin{equation}\label{eq:meanpurity}
S_k=\binom nk(2^k\bar\pi_k-1)=\binom nk\big[2^k(1-\bar L_k)-1\big].
\end{equation}
Minimizing $S_k$ therefore maximizes $\bar L_k$, or equivalently minimizes $\bar\pi_k$~\cite{Facchi2008}. The extremal count instead maximizes the number of zero terms in $S_k$. For five-body marginals at eleven qubits the graph state of Section~\ref{sec:eleven} attains both optima. For six-body marginals at twelve qubits the optimizing sets are disjoint (Section~\ref{sec:twelve-mean}).

For later use, define the dimension-scaled global purity sums and the global Rains-shadow enumerator.
\begin{align}
F_\ell&:=\sum_{|A|=\ell}2^\ell\Tr(\rho_A^2)
=\sum_{j=0}^{\ell}\binom{n-j}{\ell-j}W_j,
\quad 0\le\ell\le n,\label{eq:global-F}\\
\Sh(z)&:=2^{-n}\sum_{j=0}^{n}W_j(z-1)^j(1+3z)^{n-j}
=\sum_{r=0}^{n}H_rz^r,\qquad H_r\ge0.\label{eq:global-shadow}
\end{align}
Here $z$ is a formal indeterminate. Expanding $\prod_{i=1}^n[(I-V_i)+z(I+V_i)]$ and taking its expectation in $\rho^{\otimes2}$ gives $\Sh(z)$. Thus $H_r$ sums the rows of Eq.~\eqref{eq:pre-swap-shadow} with $r$ symmetric and $n-r$ antisymmetric sites. For globally pure states, complementary spectra imply
\begin{equation}\label{eq:pre-Fcomp}
 F_{n-\ell}=2^{n-2\ell}F_\ell.
\end{equation}
The nonnegativity $H_r\ge0$ is the global Rains-shadow inequality~\cite{Rains1999}. In particular, $H_0=h_{[n]}$ by Eq.~\eqref{eq:pre-shadow-purity}. The Pauli weight-enumerator formulation is the qubit version of the quantum weight-enumerator framework associated with the quantum analogue of the MacWilliams identities~\cite{ShorLaflamme1997}.

\subsection{Stabilizer marginals and graph-state cut ranks}\label{sec:graph-cut}
For a pure stabilizer state with stabilizer group $S$, let $S_A$ be the subgroup supported in $A$ and write $s_A=\dim_{\mathbb F_2}S_A$, using Pauli labels modulo phases. In Eq.~\eqref{eq:Pauli-defect} the squared expectation of an unsigned string $P$ is one if $P$ or $-P$ belongs to $S$, and zero otherwise~\cite{Fattal2004}. The reduced state $\rho_A=2^{-|A|}\sum_{g\in S_A}g_A$, where $g_A$ is the restriction of the signed operator $g$ to $A$, is proportional to the projector onto the common $+1$ eigenspace of $S_A$. This space has dimension $2^{|A|-s_A}$. Its normalized projector therefore has defect and von Neumann entropy (in bits) given by
\begin{equation}\label{eq:pre-stab-defect}
 D_A=|S_A|-1=2^{s_A}-1,\qquad S(\rho_A)=|A|-s_A .
\end{equation}

For a graph $G$, set $\ket G=\prod_{\{i,j\}\in E(G)}\mathrm{CZ}_{ij}\ket+^{\otimes n}$, where $\ket+=(\ket0+\ket1)/\sqrt2$. With binary adjacency matrix $\Gamma$, a product of vertex generators is specified by $v\in\mathbb F_2^n$, with $X$- and $Z$-labels $(v,\Gamma v)$~\cite{Hein2004}. It is supported in $A$ precisely when $v_{A^c}=0$ and $\Gamma_{A^c,A}v_A=0$. The kernel dimension gives $s_A=|A|-\operatorname{rank}_{\mathbb F_2}\Gamma_{A,A^c}$ and hence
\begin{equation}\label{eq:witness-rank}
 D_A=2^{|A|-\operatorname{rank}_{\mathbb F_2}\Gamma_{A,A^c}}-1.
\end{equation}
Thus a marginal $\rho_A$ of the graph state has zero defect exactly when its cut matrix $\Gamma_{A,A^c}$ has full row rank.
The possible nonzero stabilizer defects are $1,3,7,\ldots$. In particular a graph state certifying a lower bound at zero tolerance has the same count for $0\le\varepsilon<1$. The critical tolerance for the optimum within the stabilizer class is determined in Section~\ref{sec:stab-threshold}.

\section{Local shadow inequalities with combinatorial counting}\label{sec:local-combinatorial}
Fix a subsystem $T$. If all its half-system subsets were good, their total defect would be at most their number times $\varepsilon$. A scalar shadow inequality for $\rho_T$, combined with complementarity, gives a larger lower bound on precisely this sum. Thus each $T$ must contain a bad subset. Counting the pairs $(A,T)$ with $A$ bad and $A\subset T$ now bounds the total number of good marginals. At four qubits $T$ is a triple; at eight it is a five-set. The same incidence argument works at every size $4m$. Figure~\ref{fig:graph-states} gives attaining states for these and two later cases.

\begin{figure}[htbp]
\centering
\begingroup
\definecolor{ink}{HTML}{1E272D}
\definecolor{graphink}{HTML}{344650}
\newcommand{\FigQubits}[1]{#1 qubit}
\newcommand{\GraphPanel}[6]{%
  \begin{scope}[shift={(#1,0)}]
    \node[anchor=west,inner sep=0pt] at (0,84)
      {\textbf{(#2)}\enspace\FigQubits{#3}};
    \node[anchor=east,inner sep=0pt] at (109,84) {$#4$};
    \begin{scope}[shift={(54.5,39)}]
      \foreach \qubit/\px/\py in {#5}
        \coordinate (q-\qubit) at (\px,\py);
      \foreach \u/\v in {#6}
        \draw[gate] (q-\u) -- (q-\v);
      \foreach \qubit/\px/\py in {#5}
        \node[site] at (q-\qubit) {\qubit};
    \end{scope}
  \end{scope}
}
\setlength{\fboxsep}{2pt}
\resizebox{\linewidth}{!}{\raisebox{0pt}[\dimexpr\height+2pt\relax][\dimexpr\depth+2pt\relax]{\hspace{2pt}%
\begin{tikzpicture}[
  x=1pt,y=1pt, text=ink,
  font=\fontsize{8.4}{10}\selectfont,
  gate/.style={draw=graphink, line width=.62pt, line cap=round},
  site/.style={circle, minimum size=10.5pt, inner sep=0pt,
    draw=ink!88, line width=.59pt, fill=white, text=ink,
    font=\fontsize{7.25}{8.5}\selectfont}
]
\path[use as bounding box] (0,0) rectangle (482,92);

\GraphPanel{6}{a}{4}{4/6}
  {1/-37/0,2/-13/0,3/13/0,4/37/0}
  {1/2,3/4}
\GraphPanel{126}{b}{7}{32/35}
  {1/0/29,2/25.1/14.5,3/25.1/-14.5,4/0/-29,5/-25.1/-14.5,6/-25.1/14.5,7/0/0}
  {1/2,2/3,3/4,4/5,5/6,6/1,1/7,2/7,3/7,4/7,5/7,6/7}
\GraphPanel{246}{c}{8}{56/70}
  {1/-34/12,2/16/12,3/34/30,4/-16/30,5/-34/-30,6/16/-30,7/34/-12,8/-16/-12}
  {1/2,1/3,1/4,2/3,2/4,3/4,5/6,5/7,5/8,6/7,6/8,7/8,1/5,2/6,3/7,4/8}
\GraphPanel{366}{d}{9}{112/126}
  {1/-34/12,2/16/12,3/34/30,4/-16/30,5/-34/-30,6/16/-30,7/34/-12,8/-16/-12,9/0/0}
  {1/2,1/3,1/4,2/3,2/4,3/4,5/6,5/7,5/8,6/7,6/8,7/8,1/5,2/6,3/7,4/8,1/9,2/9,7/9,8/9}
\end{tikzpicture}%
\hspace{2pt}}}%
\endgroup
\caption{Graph states attaining the extremal counts for four, seven, eight, and nine qubits. Vertices represent qubits prepared in $\ket{+}$, and each edge specifies a controlled-$Z$ gate. Each fraction gives the share of maximally mixed $k$-body marginals, $k=\lfloor n/2\rfloor$. Panel (b) follows the wheel layout in Ref.~\cite{Huber2017}. Panel (c) draws $T_4$ with two complete $K_4$ layers joined by a matching~\cite{Zhang2025}. The 9-qubit PEME graph in (d) is obtained from $T_4$ in (c) by joining a new vertex $9$ to $1,2,7,8$ (row encoding in Eq.~\eqref{eq:peme-nine-drawn}).}
\label{fig:graph-states}
\end{figure}

\subsection{Four qubits and a sharp tolerance}\label{sec:four-example}\label{app:four}
Fix a pure four-qubit state $\rho$. For each three-body marginal $\rho_T$, let $S_T$ be the sum of the defects of its three two-body marginals.

\begin{lemma}\label{lem:four-shadow}
For every three-body marginal $\rho_T$ of a pure four-qubit state $\rho$,
\begin{equation}\label{eq:four-local-example}
 S_T\ge1+2D_{T^c}\ge1.
\end{equation}
\end{lemma}
\begin{proof}
Set $n=4$ and $|T|=3$ in Eq.~\eqref{eq:complement} to get $D_T=3+4D_{T^c}$. Here $E_T=a_2(T)$, so Eq.~\eqref{eq:pre-shadow-even} gives $a_2(T)\ge1+2D_{T^c}$. With $s=3$ and $q=1$, Eq.~\eqref{eq:pre-deletion} gives $S_T=2a_1(T)+a_2(T)\ge a_2(T)$.
\end{proof}

\begin{theorem}\label{thm:four-pairs}
If at most one two-body marginal of $\rho$ is bad, some three-body marginal $\rho_T$ satisfies $1\le S_T\le3\varepsilon$. Hence at least two two-body marginals are bad for $0\le\varepsilon<1/3$.
\end{theorem}
\begin{proof}
Each pair of labels lies in only two of the four triples. We can therefore choose $T$ so that all three two-body marginals of $\rho_T$ are good, giving $S_T\le3\varepsilon$. Lemma~\ref{lem:four-shadow} gives $S_T\ge1$ for the same $\rho_T$.
\end{proof}

\begin{corollary}\label{cor:four-plateau}
The maximum number of good two-body marginals of a pure four-qubit state is
\begin{equation}\label{eq:four-full-plateau}
\Qrob{\varepsilon}(4,2)=
\begin{cases}4,&0\le\varepsilon<1/3,\\6,&\varepsilon\ge1/3.\end{cases}
\end{equation}
Thus $\Qex(4,2)=4$, and the first tolerance at which the count increases is $\varepsilon_*(4,2)=1/3$.
\end{corollary}
\begin{proof}
Theorem~\ref{thm:four-pairs} leaves at most four good two-body marginals below $1/3$. In the product of Bell pairs, each pair of qubits drawn from different Bell pairs has marginal $(I/2)\otimes(I/2)$. Thus the state $\ket{\Phi^+}_{12}\otimes\ket{\Phi^+}_{34}$ attains the bound on the four pairs $13,14,23,24$. At the endpoint use the Higuchi--Sudbery state~\cite{Higuchi2000,Gour2010},
\begin{equation}\label{eq:four-hs-example}
\ket{M_4}=\frac1{\sqrt6}\bigl[\ket{0011}+\ket{1100}
+\omega(\ket{0101}+\ket{1010})+\omega^2(\ket{0110}+\ket{1001})\bigr],
\qquad\omega=e^{2\pi i/3}.
\end{equation}
Every two-body marginal of this state has diagonal $(1,2,2,1)/6$. Its only nonzero off-diagonal entries are those between $01$ and $10$, both $(\omega+\omega^2)/6=-1/6$. The spectrum is $(1/2,1/6,1/6,1/6)$, so the purity is $1/3$ and every two-body marginal has $D_B=1/3$. All six are good for $\varepsilon\ge1/3$.
\end{proof}

\subsection{Eight qubits and the counting bound at multiples of four}\label{ch:4m}\label{app:eight}
The four-qubit argument extends because the local subsystem has odd size and its complement has two fewer sites. For $n=8$ this leads to a forbidden configuration: five vertices cannot carry all five possible good four-sets.
Let $\rho$ be a pure state of $n=4m$ qubits, where $m\ge1$, and let $N=\binom{4m}{2m}$. For a $(2m+1)$-body marginal $\rho_T$ of $\rho$, the local sum $S_T$ runs over its $2m+1$ marginals of size $2m$.

\begin{lemma}\label{lem:4m-deletion}
For every $(2m+1)$-body marginal $\rho_T$ of $\rho$,
\begin{equation}\label{eq:4mlocal}
 S_T\ge1+2D_{T^c}.
\end{equation}
\end{lemma}
\begin{proof}
Here $s=|T|=2m+1$. Equation~\eqref{eq:complement} with $n=4m$ gives $D_T=3+4D_{T^c}$, so Eq.~\eqref{eq:pre-shadow-even} gives $E_T\ge1+2D_{T^c}$. In Eq.~\eqref{eq:pre-deletion} with $q=1$, the coefficient $s-j$ is at least one for each positive even $j$, since $s$ is odd. Hence $S_T\ge E_T\ge1+2D_{T^c}$.
\end{proof}

\begin{lemma}\label{lem:4m-uncovered}
If the number $b$ of bad $2m$-body marginals of $\rho$ satisfies $2m\,b<\binom{4m}{2m+1}$, then there is a $(2m+1)$-body marginal $\rho_T$ for which $S_T\le(2m+1)\varepsilon$.
\end{lemma}
\begin{proof}
Each bad $2m$-set has $2m$ one-site extensions. These extensions cover at most $2m\,b$ of the $(2m+1)$-sets. For any uncovered $T$, all $2m+1$ of the $2m$-body marginals of $\rho_T$ are good, giving the claimed upper bound.
\end{proof}

\begin{theorem}\label{thm:4mlocal}\label{thm:4m}
For $0\le\varepsilon<1/(2m+1)$,
\[
 \Qrob{\varepsilon}(4m,2m)\le N-\left\lceil\frac{N}{2m+1}\right\rceil .
\]
\end{theorem}
\begin{proof}
An uncovered $T$ would have $1\le S_T\le(2m+1)\varepsilon<1$ by Lemmas~\ref{lem:4m-deletion} and~\ref{lem:4m-uncovered}. Thus $2m\,b\ge\binom{4m}{2m+1}=2mN/(2m+1)$. Rounding $b$ up gives the finite-tolerance extension of the zero-error bound of Ref.~\cite{Zhang2025}.
\end{proof}

At $m=2$ the 56 five-site labels must each contain a bad four-set, and a bad four-set lies in four five-sets. Thus at least 14 of the 70 four-sets are bad. The graph $K_4\mathbin{\square}K_2$ in Fig.~\ref{fig:graph-states}(c) consists of two complete four-vertex layers joined by four matching edges~\cite{Zhang2025}. To count its maximally mixed marginals, label each layer by $1,2,3,4$ in matching order. A four-set selects index sets $X,Y$ in the two layers. Permuting the four matching edges and exchanging the layers shows that its cut rank depends only on $(i,j)=(\min\{|X|,|Y|\},|X\cap Y|)$. For the six possibilities
\[
 (i,j)=(0,0),(1,0),(1,1),(2,0),(2,1),(2,2)
\]
the binary cut ranks are $4,3,4,4,4,2$, respectively. These follow by ordering the cut columns as $X^c,Y^c$: the two diagonal blocks are all-ones matrices, and the off-diagonal blocks are the corresponding submatrices of the matching matrix $I_4$. There are eight four-sets of type $(1,0)$ (choose the smaller layer and its single index) and six of type $(2,2)$ (choose the two matching edges). The remaining 56 have full cut rank, proving $\Qrob{\varepsilon}(8,4)=56$ throughout $0\le\varepsilon<1/5$.

At equality, the 14 bad four-site labels each lie in four five-sets. All 56 five-sets must be covered, so each contains exactly one bad label. Complementary four-body marginals have equal defect; taking complements shows that each triple belongs to exactly one bad four-site label. These 14 labels therefore form a Steiner quadruple system $S(3,4,8)$~\cite{Hanani1960,Zhang2025}.

\section{Local shadow inequalities with global compatibility counting}\label{sec:local-global}
The comparison in Fig.~\ref{fig:bound-comparison} again uses a shadow inequality on a selected subsystem. Here the count also depends on whether the local configurations forced by the proposed number of good marginals can coexist in one pure state. At seven qubits a shadow on six sites gives a conditional upper bound on $S_3$, to be tested against the sharp global mean-purity bound. At ten qubits the restrictions within a seven-set must hold for a single family of complementary pairs across all ten sites. At twelve qubits the cover of five-sets by bad six-sets must be compatible with four- and five-body marginals, including the global shadow constraints in the separate uniformity cases.

\subsection{Seven qubits and a six-site cover}\label{sec:seven}\label{sec:seven-purity}\label{app:seven}
At zero tolerance, any state with 33 maximally mixed triples would be two-uniform: one defective pair would already make all five triples containing it defective by Eq.~\eqref{eq:monotonicity}. Huber, G\"uhne, and Siewert proved the upper bound of 32 for two-uniform states and gave the wheel graph attaining it~\cite{Huber2017}. To control a positive tolerance, we use a six-qubit shadow to bound the total three-body defect from above under the hypothesis of 33 good triples.

Lewis determined the maximum mean linear entropy of the three-body marginals of a pure seven-qubit state, $\max\bar L_3=796/917$~\cite[Theorem~11]{Lewis2026}. Equation~\eqref{eq:meanpurity} then gives the sharp global bound
\begin{equation}\label{eq:seven-lewis-purity}
 S_3\ge \frac{255}{131}.
\end{equation}

Fix a pure seven-qubit state $\rho$. Tracing out any qubit $x$ leaves its six-body marginal $\rho_R=\Tr_x\rho$ on the remaining qubits $R=[7]\setminus\{x\}$. Apply the scalar shadow identity~\eqref{eq:pre-shadow-purity} to this marginal with $T=R$. Its alternating purity sum $h_R$ is nonnegative. To relate it to $S_3$, group the $j$-body marginals of $\rho$ according to whether their labels contain $x$. Let
\[
 \Delta_j=\sum_{U\in\binom Rj}D_U-\sum_{\substack{|U|=j\\x\in U}}D_U
 \qquad(j=1,2,3).
\]

\begin{lemma}\label{lem:seven-six-shadow}
For every choice of $x$, the alternating purity sum of $\rho_R$ satisfies
\begin{equation}\label{eq:seven-union-shadow}
\begin{aligned}
 8h_R&=1-4\Delta_1+2\Delta_2-\Delta_3\\
     &=1+2\sum_{\substack{|U|=3\\x\in U}}D_U+2\Delta_2-4\Delta_1-S_3 .
\end{aligned}
\end{equation}
\end{lemma}
\begin{proof}
In the expansion of $8h_R$, replace each $\Tr\rho_X^2$ with $|X|\ge4$ by $\Tr\rho_{X^c}^2$ using pure-state complementarity~\eqref{eq:complement}, where $X^c=[7]\setminus X$. For $j=1,2,3$, a $j$-body defect $D_U$ with $U\subseteq R$ now has coefficient $8(-1)^j2^{-j}$, whereas a $D_U$ with $x\in U$ comes from $X=U^c$ and has the opposite coefficient. Thus the defect terms give $-4\Delta_1+2\Delta_2-\Delta_3$. The constant term is $8\sum_{j=0}^6(-1)^j\binom6j2^{-\min(j,7-j)}=1$, proving the first line. The second uses $S_3-\Delta_3=2\sum_{|U|=3,\,x\in U}D_U$.
\end{proof}

\begin{lemma}\label{lem:seven-good-cover}
If at most two three-body marginals are bad, one can choose $x$ so that every three-body marginal $\rho_U$ with $x\in U$ is good. Within $R=[7]\setminus\{x\}$ there are six good triples covering every pair in $R$ and containing each site of $R$ three times.
Writing $\mathcal F$ for these six good triples, the chosen $x$ and $R$ satisfy
\[
  2\Delta_2-4\Delta_1\le 2\sum_{U\in\mathcal F}D_U\le 12\varepsilon .
\]
\end{lemma}
\begin{proof}
Choose $x$ outside the union of the bad triples, which has size at most six. Partition $R$ into pairs $A,B,C$. Take $A\cup\{b'\}$ for $b'\in B$, $B\cup\{c\}$ for $c\in C$, and $C\cup\{a\}$ for $a\in A$. Every pair in $R$ occurs at least once and each site occurs three times. After a uniformly random relabeling, each fixed triple is selected with probability $6/\binom63$. The expected number of selected bad triples is at most $12/20<1$, so some relabeling avoids them all.

Expand the defects in the exact-support weights of Eq.~\eqref{eq:pre-support-mass}, writing $w_i$ and $w_{ij}$ for singleton and pair supports. Collecting coefficients gives
\[
 2\Delta_2-4\Delta_1
 =4\sum_{i\in R}w_i+2\sum_{\{i,j\}\subset R}w_{ij}
   -8w_x-2\sum_{i\in R}w_{ix}.
\]
In $2\sum_{U\in\mathcal F}D_U$, each $w_i$ with $i\in R$ has coefficient $2\cdot3=6$, and each $w_{ij}$ with $i,j\in R$ has coefficient at least $2$. These dominate the positive coefficients $4$ and $2$ above. The two terms involving $x$ are nonpositive. Thus $2\Delta_2-4\Delta_1\le2\sum_{U\in\mathcal F}D_U$, and each of the six marginals in this sum is good, bounding the sum by $2\cdot6\varepsilon$.
\end{proof}

\begin{theorem}\label{thm:seven-two-bad}
If at most two three-body marginals of a pure seven-qubit state are bad, then $S_3\le 1+42\varepsilon$. For $0\le\varepsilon\le62/2751$, the maximum number of good three-body marginals is $\Qrob{\varepsilon}(7,3)=32$.
\end{theorem}
\begin{proof}
Assume $b_\varepsilon\le2$. Lemma~\ref{lem:seven-good-cover} selects a site $x$ outside the union of the bad triples. All $\binom62=15$ three-body marginals $\rho_U$ with $x\in U$ are therefore good. The same lemma selects six good triples in $R=[7]\setminus\{x\}$ and bounds $2\Delta_2-4\Delta_1$ by $12\varepsilon$. Substituting both estimates into the identity of Lemma~\ref{lem:seven-six-shadow} for this same marginal $\rho_R$ gives
\[
\begin{aligned}
 S_3+8h_R
 &=1+2\!\sum_{\substack{|U|=3\\x\in U}}D_U+(2\Delta_2-4\Delta_1)\\
 &\le 1+2\cdot15\varepsilon+12\varepsilon=1+42\varepsilon .
\end{aligned}
\]
Since $h_R\ge0$, this gives the conditional upper bound on $S_3$.

Comparing with the global lower bound in Eq.~\eqref{eq:seven-lewis-purity}, we find that $b_\varepsilon\le2$ requires
\[
 \frac{255}{131}\le S_3\le1+42\varepsilon .
\]
For $\varepsilon<62/2751$ these bounds contradict each other. At $\varepsilon=62/2751$ both bounds must be equalities. By Lewis's equality characterization, for some qubit $s$ the 15 three-body marginals containing $s$ have a common defect, as do the 20 avoiding $s$~\cite[Theorems~11 and~41(3)]{Lewis2026}. Under $b_\varepsilon\le2$ neither group can be bad, so $S_3\le35\varepsilon<1+42\varepsilon$, contradicting equality in the upper bound.

Thus at most $35-3=32$ three-body marginals are good throughout this interval. For the wheel in Fig.~\ref{fig:graph-states}(b), the only triples with cut rank two are the hub together with two opposite rim vertices, $\{1,4,7\}$, $\{2,5,7\}$, and $\{3,6,7\}$. Indeed, if the hub is in the triple, the two rim rows are distinct and nonzero. They sum to the all-ones hub row precisely for opposite rim vertices: adjacent rim vertices each see one of the four remaining vertices, and rim vertices at distance two have a common neighbor. For a triple of rim vertices, all three rows have a one in the hub column. Any two selected rim vertices have at least two distinct neighbors other than each other that distinguish them on the six-cycle (compare distances one, two, and three). At most one of these neighbors is the third selected vertex. Thus the rows are distinct, and the hub column rules out a zero sum of all three. The cut rank is three. Equation~\eqref{eq:witness-rank} therefore gives 32 zero defects, attaining the bound. Since the minimum of the 33rd smallest defect in Eq.~\eqref{eq:critical-defect} is attained, $\varepsilon_*(7,3)>62/2751$.
\end{proof}

\subsection{Ten qubits and the covering of complementary pairs}\label{sec:ten}\label{app:ten}
Take two five-cycles and join every vertex of one cycle to every vertex of the other. In this graph the 200 cuts that split one cycle as $2+3$ and the other as $3+2$ have full binary rank, as we verify below. More than 200 good five-body marginals would leave at most 25 bad complementary pairs. We examine their intersections with seven-site sets. Whenever such a set contains at most three bad five-body marginals and all its four-body marginals are good, the shadow upper bound $3+48\varepsilon$ for the local defect sum contradicts its lower bound of six at $\varepsilon<1/16$. The other possibility forces two intersecting four-site sets all of whose five-site extensions are bad. Together with the remaining covering conditions, these require at least 26 bad pairs.

\subsubsection{Defect bounds for the five-body marginals of a seven-body marginal}\label{app:ten-local}
Fix a pure ten-qubit state $\rho$ and let $\rho_T=\Tr_{T^c}\rho$ be any of its seven-body marginals. For any tolerance $\varepsilon\ge0$, write
\[
S_T=\sum_{A\in\binom T5}D_A,\qquad
b(T)=\#\{A\in\binom T5:D_A>\varepsilon\}.
\]

\begin{lemma}\label{lem:budget}
For any seven-body marginal $\rho_T$ of $\rho$,
\begin{equation}\label{eq:budget}
S_T\ge6+12D_{T^c}\ge6.
\end{equation}
\end{lemma}
\begin{proof}
For $|T|=7$, the two-site deletion inequality~\eqref{eq:pre-delete-two} gives $S_T\ge\sum_{i\in T}D_{T\setminus\{i\}}-D_T$.
Let $\rho_C$ be the three-body marginal of $\rho$ on $C=T^c$. Equations~\eqref{eq:complement} and~\eqref{eq:monotonicity} give $D_T=15+16D_C$ and
$D_{T\setminus\{i\}}=3+4D_{C\cup\{i\}}\ge3+4D_C$. Hence
$S_T\ge7(3+4D_C)-(15+16D_C)=6+12D_C$.
\end{proof}

For any five-body marginal $\rho_A$ of $\rho$, apply Eq.~\eqref{eq:pre-shadow-even} with $s=5$. In the one-site deletion sum~\eqref{eq:pre-deletion}, the coefficients of $a_2(A)$ and $a_4(A)$ are three and one. Hence
\begin{equation}\label{eq:five-shadow}
D_A\le1+2\bigl[a_2(A)+a_4(A)\bigr]
\le1+2\sum_{R\in\binom A4}D_R.
\end{equation}

\begin{lemma}\label{lem:five-shadow}
If $b(T)\le3$ and every four-body marginal of $\rho_T$ is good, then
\begin{equation}\label{eq:ten-local-upper}
 S_T\le b(T)(1+10\varepsilon)+(21-b(T))\varepsilon\le3+48\varepsilon.
\end{equation}
\end{lemma}
\begin{proof}
The five four-body marginals of each bad five-body marginal $\rho_A$ of $\rho_T$ are good. Equation~\eqref{eq:five-shadow} thus gives $D_A\le1+10\varepsilon$. Each of the other $21-b(T)$ defects is at most $\varepsilon$. Summing gives the claimed bound. Comparing with Eq.~\eqref{eq:budget} gives $6\le S_T\le3+48\varepsilon$, so this configuration requires $\varepsilon\ge1/16$.
\end{proof}

Call a four-site set $C$ a \emph{full core} if the five-body marginals $\rho_{C\cup\{i\}}$ of $\rho$ are bad for all six $i\in[10]\setminus C$. If the four-body marginal $\rho_C$ is bad, then $C$ is a full core by monotonicity.

\begin{lemma}\label{lem:rigidity}
For $0\le\varepsilon<1/16$, every seven-body marginal $\rho_T$ of $\rho$ has at least three bad five-body marginals. If $b(T)=3$, their site sets are the three extensions in $T$ of a single four-site set $C$. The marginal $\rho_C$ is bad, and $C$ is a full core.
\end{lemma}
\begin{proof}
Suppose $b(T)\le3$. The preceding comparison then forces a bad four-body marginal $\rho_C$ of $\rho_T$ in this tolerance range. All six five-body extensions of $\rho_C$ are bad by monotonicity, and exactly three lie in $T$. Thus $b(T)=3$, with all three labels extending $C$.
\end{proof}

\subsubsection{Existence of two intersecting full cores}
Let $V=[10]$. For a complement-closed family $\mathcal B\subseteq\binom V5$ write
$b_5=|\mathcal B|$, $t=b_5/2$, and $b(T)=|\mathcal B\cap\binom T5|$. For the state $\rho$, take $\mathcal B$ to consist of the site sets $A$ of its bad five-body marginals $\rho_A$. Pure-state complementarity makes this family complement-closed. For $0\le\varepsilon<1/16$, Lemma~\ref{lem:rigidity} gives the following two restrictions on this family.
\begin{enumerate}[label=(C\arabic*)]
\item Every seven-set contains at least three bad five-sets.
\item If $b(T)=3$, then $T$ contains a full core. Its three extensions in $T$ are precisely the three bad five-sets.
\end{enumerate}

\begin{lemma}\label{lem:intersecting-cores}
Under (C1)--(C2), if $t\le25$, two distinct full cores intersect.
\end{lemma}
\begin{proof}
We first show that some seven-set contains exactly three bad five-sets. Suppose instead that each contains at least four. Let $\lambda_R$ count the bad five-sets containing $R$. For each triple $R$, complementarity gives $\lambda_R=b(V\setminus R)\ge4$. For distinct $i,j$, summing over the eight choices of a third site gives $3\lambda_{ij}\ge32$, so $\lambda_{ij}\ge11$. Each complementary pair contributes one five-set containing $i$, hence
\[
99\le\sum_{j\ne i}\lambda_{ij}=4t\le100 .
\]
Thus $t=25$ and every $i$ has exactly one partner $j$ with $\lambda_{ij}=12$ and the other eight counts equal 11. For any third site $k$, complementarity equates the numbers of bad five-sets containing $i,j,k$ and avoiding them. Inclusion--exclusion then gives $2\lambda_{ijk}=\lambda_{ij}+\lambda_{ik}+\lambda_{jk}-t=12+11+11-25=9$, impossible. By (C2), a seven-set with exactly three bad five-sets supplies a full core $C$.

Consider the $\binom43\binom64=60$ seven-sets $T$ with $|T\cap C|=3$. None contains either half of the six pairs forced by $C$. If a remaining pair meets $C$ in sizes $(1,3)$, its two halves occur in $3+6=9$ of these seven-sets. For intersection sizes $(2,2)$ they occur in $6+6=12$. The at most 19 remaining bad pairs therefore give at most $19\cdot12$ incidences. By (C1), at least $4\cdot60-19\cdot12=12$ of these seven-sets contain exactly three bad five-sets. Each must contain a full core other than $C$. If all full cores were pairwise disjoint, only one further four-set could be a core, and it would occur in just four of the 60 seven-sets, one for each triple of $C$. This contradicts the count of twelve.
\end{proof}

\subsubsection{Counting with two intersecting full cores}\label{sec:ten-core-count}

\begin{lemma}\label{lem:ten-intersecting-count}
A complement-closed family satisfying (C1)--(C2) and containing two distinct intersecting full cores has at least 26 bad complementary pairs.
\end{lemma}
\begin{proof}
Write $[A]=\{A,V\setminus A\}$ for the complementary pair of a five-set $A$. Let the cores be $C,D$, put $s=|C\cap D|$, and partition the sites as
\[
 X=C\cap D,\quad Y=C\setminus D,\quad Z=D\setminus C,\quad H=V\setminus(C\cup D).
\]
Their sizes are $s,4-s,4-s,s+2$. The six extensions of each core force $f=12$ distinct bad pairs for $s=1,2$, and $f=11$ for $s=3$, when $C\cup D$ is their one common extension. Represent a seven-set $T$ by its omitted triple $Q=V\setminus T$. A complementary pair contributes to $b(T)$ precisely when one side contains $Q$.

If $s=1$, take the 27 triples with one site each in $Y,Z,H$. None is covered by a forced pair, and each needs at least three bad pairs by (C1). The side of a pair avoiding $X$ has $b,c,d$ sites in $Y,Z,H$, with $b+c+d=5$. The pair covers
\[
 bcd+(3-b)(3-c)(3-d)=3(bc+bd+cd)-18\le6
\]
of the triples, since $bc+bd+cd\le8$ for three nonnegative integers summing to five. Thus $t\ge12+\lceil81/6\rceil=26$. For the remaining cases assume $t\le25$. Consequently no two full cores intersect in just one site.

For $s=2,3$, consider two families of omitted triples. Family $\mathcal Q_0$ takes one site from $X$ and two from $H$. Family $\mathcal Q_1$ takes one each from $X,Y\cup Z,H$. No forced pair covers a triple of $\mathcal Q_0$. Exactly one forced pair covers each $Q\in\mathcal Q_1$. For $Q=\{x,y,h\}$ with $y\in Y$, its side containing $Q$ is $C\cup\{h\}$, and the other side is $Z\cup(H\setminus\{h\})$. The case $y\in Z$ exchanges $C,D$. Let $N$ count those $Q\in\mathcal Q_1$ with $b(V\setminus Q)=3$. Counting $\mathcal Q_0$ twice and $\mathcal Q_1$ once, the $t-f$ remaining bad pairs must supply at least
\[
 6|\mathcal Q_0|+3|\mathcal Q_1|-N
\]
incidences. Indeed every triple of $\mathcal Q_0$ still needs three pairs, and those of $\mathcal Q_1$ need three more except for the $N$ triples that need two.

For each of these $N$ triples, (C2) supplies a four-site core $E$ contained in the forced five-set $Z\cup(H\setminus\{h\})$ inside $T$ (writing $y\in Y$). If $s=2$, then $|Z|=2$. Since $E$ cannot meet $D$ just once, $E$ consists of $Z$ and two sites of $H$. Assign this triple to the additional bad pair $[E\cup\{y\}]$. Its side avoiding $X$ has two sites in one of $Y,Z$, a lone site $y$ in the other, and two sites in $H$. This recovers $y$ and $E$, leaving at most two choices each for $x\in X$ and $h\in H\setminus E$. Thus a pair receives at most four triples, including the case $y\in Z$ with the parts exchanged. If $s=3$, then $|Z|=1$, so $E$ must be $H\setminus\{h\}$. Assign the triple to $[E\cup\{x\}]$, another bad pair. The side containing one site of $X$ recovers $x$ and $E$, hence $h$, leaving just two choices for the site in $Y\cup Z$. Each pair thus receives at most two triples.

In both cases neither side of the assigned pair contains $C$ or $D$, so it is an additional pair. Count each assigned triple as one extra contribution of its pair. A side with $a,b,c,d$ sites in $X,Y,Z,H$ covers triples of total weight
\[
 2a\binom d2+a(b+c)d=ad(4-a).
\]
Choose the side with $a\le s-a$. For $s=2$, if $a=1$ the two sides contribute 12 in total. If $a=0$ they contribute $4(4-d)\le12$, since $d\ge1$. Only a pair with $a=0,d=2$ can receive up to four extra contributions, giving at most 12 in total. For $s=3$, if $a=0$ we have $d\ge3$, so the two sides contribute $3(5-d)\le6$. If $a=1$ their contribution is $3d+4(5-d)=20-d$, with $d=2,3,4$. Only pairs with $d=4$ can receive up to two extra contributions, so every corrected total is at most 18.

Finally, $(|\mathcal Q_0|,|\mathcal Q_1|)=(12,32)$ for $s=2$ and $(30,30)$ for $s=3$. The corrected counts require
\[
\begin{aligned}
12(t-12)&\ge6\cdot12+3\cdot32=168 &&(s=2),\\
18(t-11)&\ge6\cdot30+3\cdot30=270 &&(s=3).
\end{aligned}
\]
Both imply $t\ge26$, as required.
\end{proof}

\subsubsection{The robust quantum extremal number at ten qubits}
\begin{theorem}\label{thm:quantum}
For every $0\le\varepsilon<1/16$,
\begin{equation}\label{eq:main}
\boxed{\Qrob{\varepsilon}(10,5)=200.}
\end{equation}
\end{theorem}
\begin{proof}
By complementarity, if the bad five-body marginals form $t$ complementary pairs, then $\binom{10}{5}-2t=252-2t$ marginals are good. More than 200 good marginals would imply $t\le25$. In the stated tolerance range, Lemma~\ref{lem:rigidity} gives (C1)--(C2), so Lemma~\ref{lem:intersecting-cores} forces two intersecting full cores. Lemma~\ref{lem:ten-intersecting-count} then gives $t\ge26$, a contradiction.

For the matching lower bound~\citep{Zhang2025,DanielsenCodeDatabase}, take two five-cycles and add all 25 edges between them. We verify directly that the resulting graph state has exactly 200 maximally mixed five-body marginals. Suppose first that $A$ takes two sites of the first cycle and three of the second. In the first cycle let $P$ be the $2\times3$ binary cut matrix from the selected pair to the remaining triple. For an adjacent pair its rows are two distinct unit vectors. For a nonadjacent pair they are two distinct weight-two vectors. In either case the two rows, together with $(1,1,1)$, are linearly independent. The same observation applies to the complementary pair in the second cycle, with cut matrix $Q$ of size $2\times3$.

Writing $J_{a,b}$ for the all-ones $a\times b$ matrix and ordering sites by cycle, the full cut matrix is
\[
 \Gamma_{A,A^c}=
 \begin{pmatrix}P&J_{2,2}\\J_{3,3}&Q^{\mathsf T}\end{pmatrix}.
\]
If $(u,v)\in\mathbb F_2^3\times\mathbb F_2^2$ is in its kernel, the last three equations read $Q^{\mathsf T}v=(\sum_i u_i)(1,1,1)^{\mathsf T}$. Independence forces $v=0$ and $\sum_i u_i=0$. Now $Pu=0$, so the same independence for $P$ gives $u=0$. Thus all $2\binom52\binom53=200$ cuts with a $2+3$ or $3+2$ split have full rank and give zero defect by Eq.~\eqref{eq:witness-rank}. This attains the upper bound throughout the interval.
\end{proof}

\subsection{Twelve qubits and a covering by six-body marginals}\label{sec:twelve}\label{app:twelve}

Fix a pure twelve-qubit state $\rho$. For each five-body marginal $\rho_F$, consider the seven-body marginal $\rho_T$ on its complement $T=F^c$. Apply Eq.~\eqref{eq:4mlocal} at $m=3$ and pair each six-body marginal of $\rho_T$ with its complement of equal defect. This gives
\begin{equation}\label{eq:12cover}
 \sum_{\substack{B\supset F\\ |B|=6}}D_B=\sum_{A\in\binom T6}D_A\ge1+2D_F .
\end{equation}
Thus one of the seven extensions of $\rho_F$ has defect at least $1/7$. In the three cases below take $\varepsilon=0$. Call the labels of six-body marginals blocks. There are $b=b_0$ bad blocks, of which $d(F)$ contain $F$. Equation~\eqref{eq:12cover} gives $d(F)\ge1$ for every $F$.

The count splits according to the four- and five-body marginals. A four-body marginal has positive defect, or the state is four-uniform with a positive five-body defect, or the state is five-uniform. The common cover condition is supplemented by overlap inequalities from global shadows and, in the five-uniform branch, by the global identity $S_6=396$. These inputs constrain which local covers a single state can realize. The three cases give at least $204$, $210$, and $206$ bad six-body marginals, respectively. We then bound the tolerance for which at most 720 six-body marginals can be good and construct a state attaining 720.

\subsubsection{States with a bad four-body marginal}\label{app:twelve-four}\label{app:finite-twelve}

\begin{lemma}\label{lem:12-twofour}
If a four-body marginal of $\rho$ has positive defect, then $b\ge204$.
\end{lemma}
\begin{proof}
The six-body marginals come in complementary pairs of equal defect. A bad marginal $\rho_R$ forces a pair $\{B,B^c\}$ to be bad whenever $R\subseteq B$ or $R\subseteq B^c$, by monotonicity. Say that a pair covers a five-body label $Q$ if $Q$ lies on either side. By Eq.~\eqref{eq:12cover}, bad pairs must cover every $Q$. We show that 101 pairs cannot suffice.

\emph{At least two bad four-body marginals.}
Fix two, on $R$ and $R'$, and put $s=|R\cap R'|$. The forced pairs leave $Q$ uncovered precisely when $i=|Q\cap R|$ and $j=|Q\cap R'|$ both belong to $\{1,2\}$.
Give such a $Q$ weight $w_h$, where $h$ counts how many of $i,j$ equal one. Let $f$ be the number of forced pairs, $W$ the total weight and $M$ the greatest weight covered by one further pair. Counting gives
\[
\begin{array}{c|c|r|r|r|r}
s &(w_0,w_1,w_2)& f&W&M&(101-f)M\\\hline
0&(0,1,2)&50&416&8&408\\
1&(0,1,1)&56&335&7&315\\
2&(0,1,1)&55&330&7&322\\
3&(1,3,3)&49&1218&22&1144
\end{array}
\]
For example, to compute $M$, delete each site in turn from each side of an unforced pair and add the weights of the resulting five-body labels. In every row $W>(101-f)M$.

\emph{A unique bad four-body marginal $\rho_R$.}
All defects of order at most three vanish, since a lower-order bad marginal would have at least two bad four-body extensions. All sums over $H$ below run over the five-body marginals $\rho_H$ of $\rho$. Put $\chi_Q(H)=(-1)^{|Q\cap H|}$. For a site $x$ and a triple $Q$, respectively, the full-state shadow in Eq.~\eqref{eq:pre-swap-shadow} gives
\[
 L_x=5\sum_H\chi_{\{x\}}(H)D_H-12\chi_{\{x\}}(R)D_R\ge0,
 \qquad
 L_Q=4\chi_Q(R)D_R-\sum_H\chi_Q(H)D_H\ge0 .
\]
Both expressions are the same sum of shadow inequalities. Write $Z_J=\sum_V(-1)^{|J\cap V|}\Tr\rho_V^2$. For $|U|=1,3$, respectively, take $8\sum_{|I|=4-|U|}Z_{U^c\triangle I}$, where $\triangle$ denotes symmetric difference. To expand, put $s=4-|U|$. For $k=|V|<6$, pairing $V$ with its complement gives the coefficient of $1+D_V$ as
\[
 16(-1/2)^k\chi_U(V)\sum_{a=0}^s(-1)^a\binom{k}{a}\binom{12-k}{s-a}.
\]
The six-body coefficients and the total constant term vanish. The coefficients at orders four and five are $(-12,5)$ for a singleton and $(4,-1)$ for a triple, in each case times $\chi_U(V)$.

Summing $L_x$ over $x\in R$ and $5L_Q$ over the four triples in $R$ gives
\[
 32D_R\le20\left(\sum_{|H\cap R|=1}D_H-\sum_{|H\cap R|=3}D_H\right).
\]
Fix a five-body marginal $\rho_F$ with $D_F>0$ and $F\cap R=\{r\}$. Put $A=(R\cup F)^c$, $B=R\setminus F$ and $C=F\setminus R$, of sizes $4,3,4$. A second positive five-body marginal can be chosen to avoid $r$, to meet each of $A,B,C$ and to have one or two sites in $B$. To see this, write $T(X)=\sum_{Q\subseteq X,\,|Q|=3}L_Q$ and add the following nonnegative terms
\[
 3L_r+4T(B)+6T(R)+\sum_{|Q|=3,\,|Q\cap B|=2} L_Q
 +\sum_{X\in\{A,C\}}\big[2T(X\cup\{r\})+T(R\cup X)\big].
\]
The coefficient of $D_R$ is $-24$. The only positive coefficients of $D_H$, listed by $(|H\cap A|,|H\cap B|,|H\cap C|)$, occur with $r\notin H$. They are $36$ for $(1,1,3),(3,1,1)$, $28$ for $(1,2,2),(2,2,1)$ and $52$ for $(2,1,2)$. These coefficients follow by summing $\chi_Q(H)$ over the specified triples. Since $D_R>0$, at least one of the five types occurs with $D_H>0$.

The pairs with a side containing $R$, $F$, or $H$ are bad. These three families have sizes $\binom82=28$, $7$, and $7$. They are disjoint because every two marked sets intersect and their union has more than six sites. Delete the five-sets already covered by these 42 pairs. Each remaining five-set $Q$ belongs to seven of the other $420$ pairs, one for each way to extend $Q$ to six sites. Assign a nonnegative integer weight to each remaining $Q$. Write $W$ for the total weight and let the load of a pair be the sum of the weights it covers. If $M$ is the largest load of the $420$ pairs, then any $k$ additional pairs covering all remaining $Q$ satisfy $kM\ge W$.

Use the following rule for all five intersection types $h=(|H\cap A|,|H\cap B|,|H\cap C|)$. Label the sites $0,\ldots,11$, put $R=\{0,1,2,3\}$, $F=\{0,4,5,6,7\}$, and encode $Q$ by $m(Q)=\sum_{i\in Q}2^i$. Take $H$ of the specified type avoiding $0$ with least $m(H)$. Permutations within $A,B,C$ carry this choice to every other labeling of the same type. Initialize the weights and loads to zero. For each remaining $Q$, sort the loads of its seven incident pairs in descending order. Choose the $Q$ whose resulting seven-tuple is lexicographically least, breaking ties by the least $m(Q)$. Increase its weight and the seven incident loads by one and repeat for $7435$ steps.

The total weight then equals $7435$. The following integer counts can be obtained by listing the $\binom{12}5$ five-sets and representing each of the $462$ pairs by its side containing $0$. An incidence occurs when $Q$ is contained in either side.
\begin{center}
\begin{tabular}{@{}lrrrrr@{}}\toprule
$h$ & $(1,1,3)$ & $(1,2,2)$ & $(2,1,2)$ & $(2,2,1)$ & $(3,1,1)$\\
Uncovered five-sets & 526 & 519 & 522 & 524 & 528\\
Maximum load $M$ & 125 & 126 & 125 & 126 & 125\\\bottomrule
\end{tabular}
\end{center}
Since $59\cdot126=7434<7435$, at least $60$ further bad pairs are needed. With the $42$ forced pairs this gives $b\ge2(42+60)=204$.
\end{proof}

\subsubsection{Four-uniform states with a bad five-body marginal}
\begin{proposition}\label{prop:12-fivebad}
A four-uniform pure twelve-qubit state with a bad five-body marginal has at least 210 bad six-body marginals.
\end{proposition}
\begin{proof}
Let $W=\sum_{|F|=5}D_F>0$, and for a fixed five-body marginal $\rho_F$ put $S_j(F)=\sum_{|H|=5,\,|F\cap H|=j}D_H$.
We use the overlap bound
\begin{equation}\label{eq:12-five-overlap}
 30D_F+6S_4(F)+S_3(F)\le\tfrac53W .
\end{equation}
Multiply the nonnegative shadow row~\eqref{eq:pre-swap-shadow} by $64$. Let $Z_{ij}$ be its mean over all $J$ with $|J\cap F|=i$ and $|J\setminus F|=j$. The difference between the right and left sides of Eq.~\eqref{eq:12-five-overlap} equals
\[
 \tfrac5{96}(18,140,63,41,91,63)\cdot
 (Z_{17},Z_{26},Z_{35},Z_{37},Z_{53},Z_{55})
 +\tfrac43S_1(F)+\tfrac{10}3S_4(F).
\]
Indeed, in $Z_{ij}/64$ the coefficient of the mean purity on $a$ sites of $F$ and $b$ sites of $F^c$ is the coefficient of $u^av^b$ in
$(1-u)^i(1+u)^{5-i}(1-v)^j(1+v)^{7-j}$. Equating complementary purities and setting defects through order four to zero gives the displayed identity.

Write $p_F=D_F/W$, and for each five-set $Q$ put $x_Q=\sum_{F\cap Q=\varnothing}p_F$. Each $F$ has 21 disjoint five-sets, so $\sum_Q x_Q=21$. A pair $F,H$ has $\binom{2+|F\cap H|}{5}$ common disjoint five-sets. Weighting Eq.~\eqref{eq:12-five-overlap} by $p_F/W$ and summing consequently gives
\[
 \sum_Q x_Q^2+9\sum_F p_F^2\le5/3.
\]
The vector consisting of $(x_Q)_Q$ and five copies of $(4p_F/3)_F$ has sum $83/3$ and squared norm at most $5/3$, since $5(4/3)^2\le9$. Cauchy--Schwarz gives at least $\lceil83^2/15\rceil=460$ positive entries.

Every $Q$ has $d(Q)\ge1$. If $x_Q>0$, choose a disjoint bad five-set $F$. Adding either of the two sites outside $F\cup Q$ to $F$ and taking complements gives two distinct bad blocks through $Q$. If $p_Q>0$, all seven extensions of $Q$ have positive defect. Hence
\[
 6b=\sum_Q d(Q)\ge 792+|\{Q:x_Q>0\}|+5|\{Q:p_Q>0\}|\ge1252 .
\]
Since $b$ is even, $b\ge210$.
\end{proof}

\subsubsection{Five-uniform states}\label{app:twelve-fiveuniform}

\begin{lemma}\label{lem:12-sixmass}
If $\rho$ is five-uniform, the defects of the six-body marginals $\rho_B$ of $\rho$ extending a fixed five-body marginal $\rho_F$ satisfy
\begin{equation}\label{eq:12weighted}
 \sum_{B\supset F,\ |B|=6}D_B=3\qquad (|F|=5).
\end{equation}
\end{lemma}
\begin{proof}
For an arbitrary pure twelve-qubit state, eliminate $W_7,\ldots,W_{12}$ from the coefficient $H_4$ of Eq.~\eqref{eq:global-shadow} using complementary purities~\eqref{eq:pre-Fcomp}. Equation~\eqref{eq:Sk} gives
\begin{equation}\label{eq:12-purity-row}
 S_6=396+704W_1+248W_2+\tfrac{504}{5}W_3+\tfrac{296}{15}W_4
       +\tfrac43 W_5+\tfrac{2^{12}}{960}H_4 .
\end{equation}
For a five-uniform state, $W_1=\cdots=W_5=0$ and $S_6=W_6$. The nonnegativity of $H_4$ gives $W_6\ge396$, while Eq.~\eqref{eq:pre-Fcomp} at $\ell=5$ gives $W_7=6(396-W_6)\ge0$. Hence $W_6=396$ and $W_7=0$. The seven-body marginal $\rho_{F^c}$ has defect three by Eq.~\eqref{eq:complement}. Only weight-six Pauli terms contribute to this defect, each in exactly one six-body marginal of $\rho_{F^c}$. Taking complements gives Eq.~\eqref{eq:12weighted}.
\end{proof}

Call a six-site block $B$ saturated if its marginal $\rho_B$ has $D_B=3$, and let $s,u$ count saturated and unsaturated bad blocks, respectively, so $b=s+u$. Saturated blocks share no five-face with any other bad block.

For a four-site label $R$ of $\rho$, let $s_R$ count the saturated six-body extensions of $\rho_R$. On the eight remaining sites draw an edge $ij$ of weight $D_{Rij}/3$ whenever this weight is positive. Equation~\eqref{eq:12weighted} gives weighted degree one at every vertex. An edge of weight one is isolated, so $s_R\in\{0,1,2,4\}$ because three such edges force the remaining two vertices to form a fourth. Call $R$ a four-edge core when $s_R=4$. Let $k_0$ count these cores and $a_1$ the labels with $s_R=1$.

\begin{lemma}\label{lem:12-corestructure}
Fix a four-edge core $R$ and write $Rq_1,\ldots,Rq_4$ for its four saturated extensions. For each pairing $p\mid R\setminus p$ of $R$, the defects $D_{p q_i q_j}$ ($i\ne j$) have a common value $d_p=d_{R\setminus p}$. At least two of the three pairings have $0<d_p<3$.
\end{lemma}
\begin{proof}
The outside pairs $q_1,\ldots,q_4$ partition the eight remaining sites. Each $B_i=R q_i$ is saturated. For every $t\notin B_i$, the seven-body marginal $\rho_{B_i\cup\{t\}}$ of $\rho$ factors as $\rho_{B_i}\otimes I_t/2$ by the Pauli expansion and Eq.~\eqref{eq:12weighted}. Thus $\rho_{B_i}=P_i/16$ for a rank-16 projection with $\Tr_{q_i}P_i=I_R$. For $U=R q_i q_j$, the rank-16 support of the eight-body marginal $\rho_U$ lies in the ranges of both extended projections. Writing $P_U$ for its support projection, $\Tr_U(P_iP_j)=\Tr_R I_R=16=\Tr P_U$. Since $P_U\preceq P_i,P_j$, $\|(P_i-P_U)(P_j-P_U)\|_{\mathrm{HS}}^2=\Tr(P_iP_j)-\Tr P_U=0$, giving $P_iP_j=P_U=P_jP_i$. Expanding the commuting projections in matrix units on $q_i,q_j$ shows that their matrix coefficients on $R$ generate mutually commuting finite-dimensional $*$-algebras. Splitting their centers and then their matrix factors decomposes each sector of $R$ into factors of dimensions $d_0,d_1,\ldots,d_4$, where $P_i$ acts only on factor $i$ and $q_i$, with rank $d_i$ there. Project the global vector onto one central sector and normalize. Within this sector its support in factor $i$ and $q_i$ is the $d_i$-dimensional range of $P_i$. Identifying each such range with $\mathbb C^{d_i}$ gives a pure state on the five factors. The marginal of this five-factor pure state on factors $0,i,j$ is maximally mixed because the sector restriction of $\rho_U=P_iP_j/16$ is proportional to the identity on these ranges. For $\{i,j,k,l\}=[4]$, Schmidt rank gives $d_0d_id_j\le d_kd_l$ and $d_0d_kd_l\le d_id_j$. Thus $d_0=1$ and all three complementary pair products agree, forcing $d_1=\cdots=d_4=d$. Since $d^4\le16$, $d=2$ would give an AME$(4,2)$ state, excluded by Section~\ref{sec:four-example}. Thus every central sector is one-dimensional. Since $\rho_R=I_R/16$, all 16 sector weights are $1/16$ and
\[
 \ket\psi=\tfrac14\sum_{z=1}^{16}\ket z_R\bigotimes_i\ket{\beta_{i,z}}_{q_i}.
\]
For $i\ne j$, the products $\ket{\beta_{i,z}\beta_{j,z}}$ form an orthonormal basis by four-uniformity.

For a pair $p\subset R$, let $\rho_{p,z}=\Tr_{R\setminus p}\ket z\!\bra z$ be the two-body marginal of the four-qubit basis state $\ket z_R$. Orthogonality after tracing two pairs shows $D_{p q_iq_j}=d_p=\frac1{16}\sum_z(4\Tr\rho_{p,z}^2-1)$, independent of $i,j$ and identical for $p$ and $R\setminus p$. Apply Eq.~\eqref{eq:12weighted} to $r q_i q_j$, $r\in R$. An extension by an outside site shares five sites with a union of three outside pairs, which is saturated as the complement of some $Rq_k$. The three remaining extensions, obtained by adding a site of $R$, represent its three pairings. Hence their $d_p$ sum to three. If $d_p=3$, every $\rho_{p,z}$ is pure and $\ket z_R$ factors across $p\mid R\setminus p$. The expansion of $\ket\psi$ then contains 16 products of six normalized pair states $\ket{\gamma_{t,z}}$, $1\le t\le6$. Tracing out any four pairs removes off-diagonal terms because two of the original outside pairs are among them. Four-uniformity now implies that for any $t\ne t'$ the 16 products $\ket{\gamma_{t,z}\gamma_{t',z}}$ form an orthonormal basis. For fixed $z$, the one-pair marginal gives $\sum_{z'\ne z}|\langle\gamma_{t,z}|\gamma_{t,z'}\rangle|^2=3$ for each $t$. Each pair therefore requires at least three other indices with nonzero overlap, and two distinct pairs cannot use the same index by product orthogonality. This requires at least $6\cdot3=18$ other indices, whereas there are only $15$. Thus at least two pairings have $0<d_p<3$.
\end{proof}

\begin{lemma}\label{lem:12-multiplicity}
The number $a_1$ of four-sets with a unique saturated extension satisfies $6a_1\ge16k_0$.
\end{lemma}
\begin{proof}
For every four-edge core $R$, Lemma~\ref{lem:12-corestructure} supplies at least four inner pairs with $0<d_p<3$. Choose any such $p\subset R$ with $0<d_p<3$ and any $q_i$. We claim that $S=p\cup q_i$ has exactly one saturated extension, $B=R\cup q_i$. A six-set extending $S$ by one inner and one outer site shares five sites with $B$ and has zero defect. Adding a whole outer pair gives defect $d_p<3$. If the two new outer sites $x,y$ belong to different pairs, apply Eq.~\eqref{eq:12weighted} to $S\cup\{x\}$. Its extension completing the pair of $x$ already has defect $d_p>0$, so adding $y$ gives defect at most $3-d_p<3$.

This produces at least $4\cdot4=16$ distinct four-sets $S$ per core. Conversely, $S$ determines its unique saturated extension $B$, and the core is $(B\setminus S)\cup p$ for one of the six pairs $p\subset S$. Thus $6a_1\ge16k_0$.
\end{proof}

\begin{lemma}\label{lem:12-face-projection}
If $d(F)=2$ and $B=F\cup\{i\}$ is a bad extension, set $P_B=64\rho_B-I$. Then
\[
 \Pi_B=(P_B^2-2P_B)/3=\Pi_F\otimes I_i,\qquad
 D_B=3\Tr\Pi_B/64
\]
for a projection $\Pi_F$ on $F$.
\end{lemma}
\begin{proof}
Let $C=F\cup\{j\}$ be the other bad extension and $T=F\cup\{i,j\}$. Their defects sum to three by Eq.~\eqref{eq:12weighted}. Their complements exhaust the defect of the extensions of $T^c$, so $B,C$ are the only bad six-subsets of $T$. Five-uniformity and $W_7=0$ give $128\rho_T=I+P_B+P_C$, extending operators by the identity on the missing site. Its complementary five-body marginal is maximally mixed, so $\rho_T$ has 32 equal nonzero eigenvalues. Hence $(P_B+P_C)^2=3I+2(P_B+P_C)$.

Write $P_B=\sum_\alpha X_\alpha\sigma_\alpha^{(i)}$, $P_C=\sum_\beta Y_\beta\sigma_\beta^{(j)}$. Comparing coefficients on both sites gives $\{X_\alpha,Y_\beta\}=0$. Comparing coefficients on $i$ alone gives $2X_\alpha=\mathrm i[X_\gamma,X_\delta]$ for cyclic $(\alpha,\gamma,\delta)$. Each $Y_\beta$ therefore also commutes with $X_\alpha$, so $X_\alpha Y_\beta=0$. The positive operators $\sum X_\alpha^2$ and $\sum Y_\beta^2$ have orthogonal ranges and sum to $3I$. The first equals three times its range projection $\Pi_F$. Expanding $P_B^2$ now gives $P_B^2-2P_B=3\Pi_F\otimes I_i$. Since $\Tr P_B=0$ and $D_B=\Tr P_B^2/64$, the trace identity follows.
\end{proof}

\begin{proposition}\label{prop:12-fiveuniform}
Every five-uniform pure twelve-qubit state has $s\le66$ saturated blocks and at least $206$ bad six-body marginals.
\end{proposition}
\begin{proof}
Let $a_0$ count the four-sets with $s_R=0$. By counting the 15 four-sets in each saturated block and using Lemma~\ref{lem:12-multiplicity},
\[
 990-15s=\sum_{|R|=4}(2-s_R)=2a_0+a_1-2k_0\ge0 .
\]
Thus $s\le66$. Each unsaturated bad block $B$ has a five-face $F$ with $d(F)\ge3$. Indeed, all its faces have degree at least two by Eq.~\eqref{eq:12weighted}. If all six had degree two, Lemma~\ref{lem:12-face-projection} would give the same projection $\Pi_B$ an identity factor on each site of $B$. It would then be a scalar projection. Since $D_B>0$, this would force $\Pi_B=I$ and $D_B=3$.

The $6s$ faces belonging to saturated blocks have degree one. The other $792-6s$ faces have degree at least two, and at least $u$ incidences with unsaturated blocks occur on faces of degree at least three. On these faces $d(F)-2\ge d(F)/3$. Summing gives
\[
 6u-2(792-6s)\ge u/3,\qquad 17b+19s\ge4752 .
\]
With $s\le66$ this yields $17b\ge3498$, hence $b\ge206$. At most $924-206=718$ six-body marginals of a five-uniform state can therefore be maximally mixed.
\end{proof}

\subsubsection{The finite-tolerance bound and its attaining state}\label{sec:twelve-radius}

The preceding three cases give $b\ge204$. The $924$ six-body marginals form $462$ complementary pairs of equal defect. At most $360$ pairs can therefore have zero defect. To control finite tolerance, choose one label from each of any $361$ distinct pairs and write $\mathcal C$ for the resulting family.

\begin{lemma}\label{lem:12-quartic}
For every such $\mathcal C$ and every pure twelve-qubit state,
\[
 \sum_{A\in\mathcal C}D_A>361\eta_{12}, \qquad \eta_{12}=2^{-2^{12943}} .
\]
\end{lemma}
\begin{proof}
The sum is invariant under local changes of basis. Map a product vector of maximum overlap with the state to $\ket{0}^{\otimes12}$ and make its amplitude positive. Varying the product vector in the two real tangent directions at each site shows that all twelve weight-one amplitudes vanish. Local diagonal phases can also make the twelve weight-eleven amplitudes real: their phase shifts are $(J-I)\theta$, where the all-ones matrix $J$ has size twelve and $J-I$ is invertible. Thus every state has a representative in a fixed real coordinate subspace of dimension $N=8192-1-24-12=8155$.

For an unnormalized vector $x$ in this subspace let $\rho_A(x)=\Tr_{A^c}\ket x\!\bra x$. The integer quartic
\[
 H(x)=\sum_{A\in\mathcal C}\bigl[64\Tr\rho_A(x)^2-\|x\|_2^4\bigr]
\]
equals $\|x\|_2^4\sum_{A\in\mathcal C}D_A(x/\|x\|_2)$ for $x\ne0$. It is positive there, since the $361$ pairs cannot all have zero defect. Put $h=\min_{\sum_i x_i^4=1}H(x)>0$. At a minimizer, Lagrange multipliers and Euler's identity give $\partial_i H(x)=t x_i^3$ with $t=4h$. Thus $4h$ is a root of the resultant~\cite{Qi2005}
\[
 P(t)=\operatorname{Res}_x\bigl(t x_i^3-\partial_i H(x):1\le i\le N\bigr).
\]
This integer polynomial is monic of degree $M=N3^{N-1}$: the resultant has degree $3^{N-1}$ in the coefficients of each cubic and equals one on the pure cubes $x_i^3$. Each root has a nonzero complex solution $x$. In its largest-modulus coordinate, $|t|$ is bounded by the sum of absolute coefficients of the corresponding derivative, which is at most
\[
 K=4\cdot361(16\cdot64^5+N^2)<2^{45}.
\]
Indeed each purity has four indices of size 64, each product of four complex amplitudes gives at most 16 real terms, and $\|x\|_2^4$ has coefficient sum $N^2$. Differentiation costs a factor of at most four. The product of the nonzero roots of $P$ is a nonzero integer up to sign, so $4h\ge K^{-(M-1)}$. For every Euclidean unit vector $z$,
\[
 H(z)\ge h\sum_i z_i^4\ge \frac{K^{-(M-1)}}{4N}>361\eta_{12},
\]
where $4N\cdot361<2^{24}$ and $45(M-1)+24<2^{12943}$.
\end{proof}

\begin{theorem}\label{cor:12-qex}
For $0\le\varepsilon\le\eta_{12}$,
\[
 \boxed{\Qrob{\varepsilon}(12,6)=720.}
\]
\end{theorem}
\begin{proof}
More than $720$ good six-body marginals would give at least $361$ good complementary pairs. Taking one side of each would give $\sum_{A\in\mathcal C}D_A\le361\varepsilon$, contradicting Lemma~\ref{lem:12-quartic}.

For attainment, use the following stabilizer representation of a graph in the length-twelve catalogue~\cite{DanielsenParker2006,DanielsenCodeDatabase}. Label the sites by $(i,z)\in\mathbb Z_3\times\mathbb F_4$. Label single-qubit Paulis by $0,1,\omega,\omega^2$ (identity and the three nonidentity Paulis), where $\omega^2=\omega+1$. Take the Pauli words
\[
 c_i(z)=u_i+a_i z+a_{i+1}z^2,\qquad u_i,a_i\in\mathbb F_4 .
\]
These form a six-dimensional $\mathbb F_4$-linear space: a quadratic vanishing at all four points is zero. They commute because the trace of $\sum_{i,z}c_i(z)c'_i(z)^2$ is their binary commutation form. In fact this sum vanishes: under $\sum_z f(z)g(z)^2$, the constant is orthogonal to $1,z,z^2$, while $z,z^2$ are orthogonal of norm one. Each term $a_i(a'_i)^2$ therefore occurs twice. The twelve independent binary generators define a pure stabilizer state.

By Eq.~\eqref{eq:pre-stab-defect}, a marginal is maximally mixed if it supports no nonzero word. If its six sites meet the blocks in sizes $1,2,3$, the polynomial on the one-site block has three roots and is zero. Either neighboring polynomial with a nonzero remaining coefficient has exactly one root, so the two-site block forces this coefficient to vanish. No nonzero constant fits. This gives $3!\binom41\binom42\binom43=576$ good marginals.

For sizes $(2,2,2)$, write $d_i$ for the difference of the two omitted labels in block $i$. Vanishing at both fixes $u_i$ and requires $a_i=d_i a_{i+1}$. A nonzero solution exists precisely when $d_0d_1d_2=1$. Each nonzero difference occurs for two unordered pairs, so $6^3-3^2\,2^3=144$ further marginals are maximally mixed. The $576+144=720$ zero defects attain the upper bound throughout the stated interval.
\end{proof}

\subsubsection{Five-uniformity and the minimum mean purity}\label{sec:twelve-mean}
Returning to zero tolerance, we sharpen the count for five-uniform states to compare it with the minimum mean purity.

\begin{lemma}\label{lem:12-saturation-refined}
Every five-uniform pure twelve-qubit state has at most $60$ saturated blocks.
\end{lemma}
\begin{proof}
For the sharper bound we first refine the reverse count of Lemma~\ref{lem:12-multiplicity}. Fix a four-set $S$ arising from a core $R$ there, and let $B$ be its unique saturated extension. If $R=X\cup p$ is valid, with $X=B\setminus S$ and $S=p\cup q_i$, the complementary choice gives $X\cup q_i$. Since $X=R\setminus p$ and $d_X=d_p\in(0,3)$, the same construction gives $X\cup q_i$ a unique saturated extension. It cannot be a core. Hence at most one member of each of the three complementary pairings of $S$ yields a core. Thus $3a_1\ge16k_0$.

Set $\delta_R=2-s_R$ for a four-site label $R$. Since a saturated block contains $15$ such labels, $\Delta:=\sum_R\delta_R=990-15s$. Since $a_1\ge5k_0$ by this refined count and $s_R\in\{0,1,2,4\}$,
\[
 \Delta\ge a_1-2k_0\ge0,\qquad
 16k_0-2a_1\le 2(a_1-2k_0)\le2\Delta .
\]
In particular $s\le66$. For a three-site label $T$ write $s_T$ for the number of saturated extensions and set $\delta_T=6-s_T$. Then $\sum_T\delta_T=1320-20s=4\Delta/3$.

Fix a saturated block $B$, and let $N_4$ count the other saturated blocks meeting $B$ in four sites. Complementarity pairs these with the blocks meeting $B$ in two sites. There is one disjoint block and none meeting $B$ in one or five sites. Counting triples and four-subsets in the intersections gives
\[
 \sum_{R\subset B,\,|R|=4}\delta_R=15-N_4,\qquad
 \sum_{T\subset B,\,|T|=3}\delta_T=72-s+2(15-N_4).
\]
Sum over saturated $B$. Each $T$ occurs $6-\delta_T$ times, and $\sum_R s_R\delta_R=a_1-8k_0$. It follows that
\[
 s(72-s)=8\Delta+16k_0-2a_1-\sum_T\delta_T^2
          \le10\Delta .
\]
Write $g=66-s\ge0$, so $\Delta=15g$ and $s\le150g/(6+g)$. If $s\ge62$, then $g\le4$ and the right side is at most $150\cdot4/10=60$, a contradiction. Saturated blocks occur in complementary pairs, so $s\le60$.
\end{proof}

\begin{proposition}\label{prop:12-fiveuniform-comparison}
Every five-uniform pure twelve-qubit state has at least $234$ six-body marginals with positive defect.
\end{proposition}
\begin{proof}
We show that each unsaturated bad block has at least three five-faces of degree at least three. By Lemma~\ref{lem:12-face-projection}, every degree-two face of $B$ gives an identity factor of the same projection $\Pi_B=(P_B^2-2P_B)/3$. If four faces of an unsaturated $B$ were covered twice, $\Pi_B=\Pi_R\otimes I$ would act only on the remaining two qubits $R$, with $r=\operatorname{rank}\Pi_R\le3$ (the defect equals $3\Tr\Pi_B/64$). The nonzero projection $E=(\Pi_B+P_B)/4$ is positive under partial transposition across $R$. Indeed partial spin flip on the two sites fixes every full-weight term of $P_B$, and full spin flip fixes $P_B$ and its spectral projection $\Pi_B$, hence fixes $\Pi_R$ there. The reduction map composed with transposition is completely positive (its Choi matrix is $I-\mathrm{Swap}$). The reduction criterion therefore gives $E\preceq I_R\otimes\Tr_R E=(r/4)I$, impossible for a nonzero projection. Each unsaturated block thus has at least three faces of degree at least three. Each face outside saturated blocks has degree at least two by Eq.~\eqref{eq:12weighted}. Each unsaturated block contributes at least three incidences to faces of degree at least three. Summing $d(F)-2\ge d(F)/3$ on these faces gives
\begin{equation}\label{eq:12n}
 6u-2(792-6s)\ge u,\qquad\text{hence}\qquad 5b+7s\ge1584 .
\end{equation}
Since $s\le60$ by Lemma~\ref{lem:12-saturation-refined}, we have $5b\ge1164$. The even integer $b$ is therefore at least $234$.
\end{proof}

The five-uniform dodecacode state~\cite{Scott2004,Zhang2025} attains $\Qex(12,5)=\binom{12}5=792$. By Eq.~\eqref{eq:12-purity-row}, $S_6\ge396$, and equality requires $W_1=\cdots=W_5=0$. Conversely, summing Eq.~\eqref{eq:12weighted} for a five-uniform state gives $6S_6=3\binom{12}5$. Thus precisely the five-uniform states attain the minimum six-body mean purity $5/224$, or equivalently the maximum mean linear entropy $219/224$~\cite{Lewis2026}. By Proposition~\ref{prop:12-fiveuniform-comparison}, each has at most $924-234=690$ maximally mixed six-body marginals. Consequently no state attaining the five-body count or the minimum six-body mean purity attains the six-body count $720$, excluding a twelve-qubit PEME state.

\section{Global shadow inequalities}\label{sec:global-shadow}
In the third application of the comparison in Fig.~\ref{fig:bound-comparison}, the two sides constrain sums over the full half-system layer. At nine qubits, for instance, summing the seven-site shadows gives $S_4\ge18$. At zero tolerance, a 113th good four-set would leave at most 13 bad four-sets; their disjointness relations and complementary purities give $S_4\le17$. Both bounds refer to the same 126 four-site labels. At eleven qubits the proof uses $S_5$ over all 462 five-site labels and also uses global shadow inequalities that retain the intersections with a marked set. In each case the excess-count hypothesis supplies the incompatible upper constraints. What determines the method is the scope of the inequality used in the contradiction: individual local shadows can be summed to obtain it, and local defect estimates can enter the subsequent counting.

\subsection{Nine qubits and the exclusion of a 113th good marginal}\label{sec:nine}\label{app:nine}
The graph in Fig.~\ref{fig:graph-states}(d) gives a lower bound of 112~\cite{DanielsenCodeDatabase,Cabello2011,CabelloGraphData,Zhang2025}. Its cut matrices identify which four-body marginals fail to be maximally mixed.

Encode row $i$ of an adjacency matrix by $\sum_{j\in N(i)}2^{j-1}$, where sites are labeled from one. The graph of Fig.~\ref{fig:graph-states}(d) has two $K_4$ layers on $1,\ldots,4$ and $5,\ldots,8$, matching edges $(i,i+4)$, and edges from $9$ to $1,2,7,8$. Its rows are
\begin{equation}
(286,301,75,135,225,210,436,376,195).
\label{eq:peme-nine-drawn}
\end{equation}
For comparison with the catalog of Cabello \textit{et al.}~\cite{Cabello2011,CabelloGraphData}, its 15-edge representative has rows $(82,97,200,324,385,386,15,308,184)$ after shifting the catalog labels by one. Locally complement at $(7,1,2,8,9,3)$ in order, toggling the edges among the neighbors of each selected vertex. Relabel the original vertices $1,\ldots,9$ as $(3,4,1,5,7,8,9,6,2)$ to obtain Eq.~\eqref{eq:peme-nine-drawn}.

For each $A\in\binom{[9]}4$, restrict the four rows indexed by $A$ to the columns in $A^c$ and take the rank over $\mathbb F_2$. The ranks $4,3,2$ occur $112,12,2$ times, respectively. All cuts with at most three sites on the smaller side have full row rank. By the cut-rank criterion of Section~\ref{sec:graph-cut}, exactly 112 four-body marginals are maximally mixed. We now show that a 113th good marginal would violate a shadow inequality even at tolerance $1/127$.

\subsubsection{A global lower bound from seven-body shadows}\label{app:nine-shadow}
Fix a pure nine-qubit state $\rho$ with sites labeled by $[9]$. For $T\subseteq[9]$ with $|T|=7$, let $\rho_T=\Tr_{T^c}\rho$ be its seven-body marginal on $T$. Write $R_7=\sum_{|T|=7}h_T$ for the sum of their spin-flip overlaps, as defined in Eq.~\eqref{eq:pre-shadow-purity}.
\begin{lemma}\label{lem:nine-global}
The sum $S_4$ of the defects of the 126 four-body marginals of $\rho$ satisfies
\begin{equation}\label{eq:9lower}
 S_4=18+64W_1+22W_2+6W_3+4R_7\ge18 .
\end{equation}
\end{lemma}
\begin{proof}
To form a single inequality for the full four-body layer, apply the local spin-flip identity~\eqref{eq:shadow0} to each of the 36 seven-body marginals of $\rho$ and sum. Each weight-$j$ string occurs $\binom{9-j}{7-j}$ times. Thus
\begin{equation}\label{eq:9R7}
 128R_7=36-28W_1+21W_2-15W_3+10W_4-6W_5+3W_6-W_7\ge0 .
\end{equation}
At $n=9$, Eq.~\eqref{eq:pre-Fcomp} gives $F_5=2F_4$, $F_6=8F_3$, and $F_7=32F_2$. Expanding these three identities using Eq.~\eqref{eq:global-F} gives
\begin{align*}
 W_5&=126+42W_1+7W_2-3W_3-3W_4,\\
 W_6&=84-7W_2+2W_4,\\
 W_7&=108-24W_1-10W_2+3W_3+2W_4 .
\end{align*}
Substituting yields $W_4=18+8W_1+W_2+4R_7$. Substitution into Eq.~\eqref{eq:Sk} with $k=4$, namely $S_4=56W_1+21W_2+6W_3+W_4$, gives Eq.~\eqref{eq:9lower}. All four terms added to 18 are nonnegative.
\end{proof}
By Eq.~\eqref{eq:meanpurity}, this sharp bound gives $\max\bar L_4=13/14$, equivalent to the minimum mean purity $1/14$ found by Zha \emph{et al.}~\cite{Zha2020}.

\subsubsection{Closed neighborhoods and individual defect bounds}\label{app:nine-neighborhoods}
Given a pure nine-qubit state $\rho$, construct an auxiliary graph as follows. For each four-element subset $A\subseteq[9]$, take one vertex labeled $A$ to represent the marginal $\rho_A$ on those four qubits. Join two vertices $A$ and $B$ exactly when $A\cap B=\varnothing$. The resulting graph is the Kneser graph $KG(9,4)$~\cite{Lovasz1978}. Its $\binom{9}{4}=126$ vertices and their adjacencies depend only on the qubit labels.

The five neighbors of $A$ correspond to the five ways to retain four qubits from $A^c$. Write $N(A)$ for this neighborhood and $N[A]=N(A)\cup\{A\}$ for the closed neighborhood. Among the $b=b_\varepsilon$ bad vertices, a vertex is \emph{trapped} if all its neighbors are bad. Let $u$ count the trapped vertices.

\begin{lemma}\label{lem:three}
The closed neighborhoods of any three distinct vertices cover at least 14 vertices.
\end{lemma}
\begin{proof}
For distinct $A,B$, the intersection of $N[A]$ and $N[B]$ has size two when $A,B$ are disjoint, one when $|A\cap B|=3$, and zero otherwise. Indeed, a common open neighbor is a four-subset of $(A\cup B)^c$. If at most one pair of the three centers is disjoint, the pairwise intersection sizes sum to at most four. Inclusion--exclusion gives a union of size at least $18-4=14$. If one center is disjoint from both others, those two four-sets lie in its five-site complement and intersect in three sites. The three intersection sizes are $2,2,1$, and the first center belongs to all three neighborhoods. The union therefore has size at least $18-5+1=14$. There cannot be three pairwise disjoint four-sets.
\end{proof}

\begin{lemma}\label{lem:nine-caps}
If $b\le13$, then $u\le2$. The marginals at good vertices, nontrapped bad vertices, and trapped vertices have respective defect bounds $c_E=\varepsilon$, $1+2\varepsilon$, and $3+4\varepsilon$. Consequently,
\begin{align}
 S_4&\le\sum_E c_E=(126-b)\varepsilon+(b-u)(1+2\varepsilon)+u(3+4\varepsilon)\nonumber\\
 &=b+2u+(126+b+2u)\varepsilon\le17+143\varepsilon .\label{eq:9upper}
\end{align}
\end{lemma}
\begin{proof}
The closed neighborhood of each trapped vertex consists entirely of bad vertices. Three such centers would give at least 14 bad vertices by Lemma~\ref{lem:three}. For adjacent vertices $E,B$, apply the disjoint-set bound~\eqref{eq:pre-disjoint} with $n=9$ and $|B|=4$. This gives $D_E\le1+2D_B$. A nontrapped bad vertex has a good neighbor, so one application gives $D_E\le1+2\varepsilon$. A trapped vertex has five bad neighbors, at most one of which can be the other trapped vertex. Going through a nontrapped neighbor gives $D_E\le1+2(1+2\varepsilon)=3+4\varepsilon$ (Fig.~\ref{fig:9-kneser-types}). Summing over the three classes and using $b\le13$, $u\le2$ proves Eq.~\eqref{eq:9upper}.
\end{proof}

\begin{figure}[htbp]
\centering
\begin{tikzpicture}[>=Latex,scale=0.83, every node/.style={font=\small}, vtx/.style={circle,draw,minimum size=6.5mm,inner sep=0pt}, good/.style={vtx,fill=green!18}, bad/.style={vtx,fill=red!16}, unknown/.style={vtx,fill=gray!12}, nt/.style={vtx,fill=orange!22}]
\begin{scope}[shift={(0,0)}]
\foreach \x/\y in {0/1.05,1.0/0.35,0.62/-0.85,-0.62/-0.85,-1.0/0.35} {\draw[gray] (0,0)--(\x,\y);\node[unknown] at (\x,\y) {?};}
\node[good] at (0,0) {$G$};
\node[draw=none,rectangle] at (0,-1.55) {good vertex};
\end{scope}
\begin{scope}[shift={(4.7,0)}]
\foreach \x/\y in {0/1.05,1.0/0.35,0.62/-0.85,-0.62/-0.85,-1.0/0.35} {\draw[gray] (0,0)--(\x,\y);}
\node[good] at (0,1.05) {$G$};
\node[unknown] at (1.0,0.35) {?};\node[unknown] at (0.62,-0.85) {?};\node[unknown] at (-0.62,-0.85) {?};\node[unknown] at (-1.0,0.35) {?};
\node[bad] at (0,0) {$B$};
\node[draw=none,rectangle] at (0,-1.55) {nontrapped bad vertex};
\end{scope}
\begin{scope}[shift={(9.6,0)}]
\coordinate (B) at (0,0);
\coordinate (n1) at (0,1.05);
\coordinate (n2) at (1.0,0.35);
\coordinate (n3) at (0.62,-0.85);
\coordinate (n4) at (-0.62,-0.85);
\coordinate (n5) at (-1.0,0.35);
\coordinate (G) at (2.05,0.35);
\foreach \P in {n1,n2,n3,n4,n5} {\draw[gray] (B)--(\P);}
\draw[very thick,blue] (B)--(n2)--(G);
\node[bad] at (B) {$B$};
\node[nt] at (n1) {$B'_2$};
\node[nt] at (n2) {$B'_1$};
\node[nt] at (n3) {$B'_3$};
\node[nt] at (n4) {$B'_4$};
\node[bad] at (n5) {$B^\dagger$};
\node[good] at (G) {$G$};
\node[draw=none,rectangle] at (0,-1.55) {trapped bad vertex $B$};
\end{scope}
\end{tikzpicture}
\caption{Neighborhoods in $KG(9,4)$. Each vertex represents a four-body marginal, and each edge joins two marginals on disjoint sets of qubits. From left to right, the panels show a good vertex (green), a bad vertex with a good neighbor, and a trapped bad vertex. Orange neighbors are known to be nontrapped. The blue path gives two propagation steps. The fifth red neighbor may also be nontrapped.}\label{fig:9-kneser-types}
\end{figure}

\subsubsection{Four matched neighborhoods and the exclusion of 113 good four-body marginals}\label{app:nine-matching}
For $u\le1$, substituting $b\le13$ in Eq.~\eqref{eq:9upper} gives $S_4\le15+141\varepsilon$. When $u=2$, the next two lemmas bound all five defects in each of four disjoint neighborhoods jointly.

\begin{lemma}\label{lem:nine-matching}
If $b\le13$ and $u=2$, then for either trapped vertex $A$ there are four good vertices $G_j=(A\setminus\{a_j\})\cup\{i_j\}$ with distinct $a_j\in A$ and distinct $i_j\in A^c$. The four open neighborhoods $N(G_j)$ are pairwise disjoint and each contains a bad vertex.
\end{lemma}
\begin{proof}
Fix a trapped vertex $A$ and put $C=A^c$. The five sets $B_i=C\setminus\{i\}$, $i\in C$, are bad. Form a $4\times5$ array with entries $G_{ai}=(A\setminus\{a\})\cup\{i\}$, $a\in A$, $i\in C$. At most seven entries are bad, since $A$ and the $B_i$ already account for six bad vertices outside the array.

First consider an entirely bad row $a$. With $S=A\setminus\{a\}$, the six four-sets extending $S$ (the five row entries and $A$) and the five $B_i$ are all bad. Every vertex $X\ne A$ has at most two neighbors among these eleven. Indeed, if $X$ meets $S$, it is disjoint from no extension of $S$ and from at most one $B_i$, since it contains a site of $C$. If $X$ misses $S$, it is disjoint from precisely two extensions of $S$ and from no $B_i$, since at least three of its sites lie in $C$. At most two further vertices can be bad. Thus $X$ cannot have five bad neighbors, contrary to the existence of a second trapped vertex. Every row has a good entry. If $r$ rows met good entries in at most $r-1$ columns, they would contain at least $r(6-r)$ bad entries. For $r=2,3,4$ this is $8,9,8$, respectively, whereas at most seven are bad. Hall's marriage theorem now gives the claimed matching.

Two matched entries share exactly two sites. Their five-site complements therefore intersect in three sites, so their collections of four-faces are disjoint.
Each matched $G_j$ is disjoint from the bad vertex $B_{i_j}$.
\end{proof}

\begin{lemma}\label{lem:nine-joint}
If $b\le13$ and $u=2$, then
\[
 S_4\le17+127\varepsilon+12W_1+8W_2+4W_3 .
\]
\end{lemma}

\begin{proof}
For a good vertex $G$, consider the complementary five-body marginal $\rho_T$ with $T=G^c$. The five neighbors of $G$ label the five four-body marginals of $\rho_T$. In the one-site deletion identity~\eqref{eq:pre-deletion}, set $s=5$. The coefficient of $a_j(T)$ is $5-j$, and Eq.~\eqref{eq:complement} gives $D_T=1+2D_G$. Thus
\begin{align}
 \sum_{E\in\binom T4}D_E
 &=D_T+3a_1(T)+2a_2(T)+a_3(T)-a_5(T)\nonumber\\
 &\le1+2\varepsilon+3W_1+2W_2+W_3 .\label{eq:nine-joint}
\end{align}
Each of the four neighborhoods selected by Lemma~\ref{lem:nine-matching} contains a bad vertex $B_{i_j}$. This vertex has a good neighbor $G_j$, so its individual bound is $1+2\varepsilon$. Each other vertex has individual bound at least $\varepsilon$. The five bounds therefore sum to at least $1+6\varepsilon$. In each neighborhood replace their sum by Eq.~\eqref{eq:nine-joint}. Because the four neighborhoods are disjoint, the change from the sum of the individual bounds is at most
\[
 4\big[(1+2\varepsilon+3W_1+2W_2+W_3)-(1+6\varepsilon)\big] .
\]
Adding this change to $17+143\varepsilon$ in Eq.~\eqref{eq:9upper} proves the claimed upper bound.
\end{proof}

\begin{theorem}\label{thm:nine-conditional}
If $b\le13$, then $\varepsilon>1/127$. Thus every pure nine-qubit state has at least 14 bad four-body marginals for $0\le\varepsilon\le1/127$.
\end{theorem}
\begin{proof}
Suppose $b\le13$. By Lemma~\ref{lem:nine-caps}, at most two bad vertices are trapped. If $u\le1$, Eqs.~\eqref{eq:9lower} and~\eqref{eq:9upper} give $18\le S_4\le15+141\varepsilon$, so $\varepsilon\ge1/47>1/127$. If $u=2$, comparing the four-neighborhood upper bound of Lemma~\ref{lem:nine-joint} with the global lower bound~\eqref{eq:9lower} gives
\begin{equation}\label{eq:nine-refined}
 18+52W_1+14W_2+2W_3+4R_7\le17+127\varepsilon .
\end{equation}
All terms on the left are nonnegative, so $\varepsilon\ge1/127$. It remains to exclude equality.

Suppose $\varepsilon=1/127$. Equation~\eqref{eq:nine-refined} forces $W_1=W_2=W_3=R_7=0$ and equality in the upper bound of Lemma~\ref{lem:nine-joint}. Fix a trapped vertex $A$, take the four matched good vertices $G_j=(A\setminus\{a_j\})\cup\{i_j\}$ of Lemma~\ref{lem:nine-matching}, and put $U=\bigcup_{j=1}^4N(G_j)$. The upper bound uses the individual estimate $D_X\le c_X$ for each $X$ outside $U$, so equality requires $D_X=c_X$ there. Since $A$ has no good neighbor, $A\notin U$ and $D_A=3+4\varepsilon$.

At this endpoint, Eq.~\eqref{eq:nine-joint} bounds the sum of the five neighboring defects of every good vertex by $1+2\varepsilon$. Of the 16 unmatched entries in the $4\times5$ array, at most 13 are bad. Choose a good one, $G=(A\setminus\{a_\ell\})\cup\{i\}$ with $i\ne i_\ell$. Its neighbor $B=A^c\setminus\{i\}$ is disjoint from $A$, so $D_A\le1+2D_B$ gives $D_B\ge1+2\varepsilon$. Now choose a four-subset $F$ of $G^c$ containing both $a_\ell$ and $i_\ell$. It meets $G_\ell$ at $i_\ell$ and every other matched vertex at $a_\ell$. Thus $F\notin U$ and $D_F=c_F\ge\varepsilon$. The vertices $B$ and $F$ are distinct neighbors of $G$ since only $F$ contains $a_\ell$. Their defects alone sum to at least $1+3\varepsilon>1+2\varepsilon$, contradicting the bound for $G$. Hence $\varepsilon>1/127$.
\end{proof}

\begin{corollary}\label{thm:9}
For $0\le\varepsilon\le1/127$, $\Qrob{\varepsilon}(9,4)=112$.
\end{corollary}
\begin{proof}
Theorem~\ref{thm:nine-conditional} forces at least 14 bad four-body marginals, leaving at most $126-14=112$ good ones.
The graph state in Eq.~\eqref{eq:peme-nine-drawn} attains 112 already at zero tolerance.
\end{proof}

\subsection{Eleven qubits and heavy extensions of four-site supports}\label{sec:eleven}\label{app:eleven}
A weighted global shadow first bounds the largest five-body defect under the hypothesis of 397 good marginals. The remaining argument counts the heavy extensions of a marked four-site Pauli support. When at least five extensions are heavy, their disjoint neighbors give an edge bound for four ordinary graphs, and the required triangles cannot be realized.

\subsubsection{Global and local defect bounds}
Fix a pure eleven-qubit state $\rho$ and consider its 462 five-body marginals at tolerance $\varepsilon\ge0$. Let $b=b_\varepsilon$ count the five-body marginals with defect above $\varepsilon$, and write $M=\max_{|A|=5}D_A$. To exclude 397 good marginals, suppose $b\le65$. We first bound $M$ using the global shadow and the counts of disjoint marginals, then count extensions of four-site Pauli supports.

\begin{lemma}\label{lem:eleven-shadow}
Every pure eleven-qubit state satisfies
\begin{equation}\label{eq:eleven-global}
 S_5\ge66+252W_1+\tfrac{444}{5}W_2+\tfrac{156}{5}W_3+\tfrac{22}{5}W_4.
\end{equation}
Let $w_R$ be the mass of the exact Pauli support $R$. If at most 65 five-body marginals have defect above $\varepsilon$, then every five-body marginal $\rho_A$ satisfies
\begin{equation}\label{eq:eleven-local}
 D_A\le1+20\varepsilon+2\!\sum_{R\in\binom A4}w_R .
\end{equation}
\end{lemma}
\begin{proof}
Let $L_J$ denote the global shadow in Eq.~\eqref{eq:pre-swap-shadow}. Pairing complementary purities and summing over $J$ gives
\begin{align*}
32\sum_{|J|=8} L_J&=-660-420W_1-48W_2-32W_3+26W_4+10W_5,\\
32\sum_{|J|=10} L_J&=132+52W_1+16W_2-2W_4-2W_5 .
\end{align*}
The coefficient of $W_r$ for $|J|=s$ is $\sum_{k=r}^5 2^{6-k}\binom{11-r}{k-r}\sum_j(-1)^j\binom{k}{j}\binom{11-k}{s-j}$. The first row and $S_5=210W_1+84W_2+28W_3+7W_4+W_5$ imply Eq.~\eqref{eq:eleven-global}.

If a two-body marginal had defect above $\varepsilon$, defect monotonicity~\eqref{eq:monotonicity} would force all $\binom{9}{3}=84$ five-body marginals containing that pair to have defect above $\varepsilon$, contradicting $b\le65$. Thus every two-body marginal is good. Tracing out one qubit of a pair gives the same conclusion for every one-body marginal. Hence $w_R\le D_R\le\varepsilon$ for $1\le |R|\le2$. For each five-body marginal $\rho_A$, the ten two-site support masses in the local scalar shadow inequality~\eqref{eq:pre-shadow-even} are each at most $\varepsilon$, giving Eq.~\eqref{eq:eleven-local}.
\end{proof}

\subsubsection{An upper bound on the largest five-body defect}
Each five-body marginal corresponds to a vertex of $KG(11,5)$. Two vertices are adjacent when the retained sites are disjoint. On an edge $CB$, Eq.~\eqref{eq:pre-disjoint} gives $D_C\le1+2D_B$. Hence each vertex with $D_C>1+2\varepsilon$ has six bad neighbors. If $D_C>3+4\varepsilon$, every vertex within distance two of $C$ is bad.

\begin{lemma}\label{lem:eleven-deep}
If $b\le65$ and $0\le\varepsilon\le1/200$, then $M\le3+60\varepsilon$.
\end{lemma}
\begin{proof}
Suppose $D_A>3+60\varepsilon$. Group the five-body marginals $\rho_C$ by $j=|A\cap C|$. Let $d_j$ and $b_j$ denote the total defect and bad count in class $j$, and let $U_2=\sum_{|C|=5,\,|A\cap C|=2}\sum_{R\in\binom C4}w_R$. We bound the same quantity $Y=2d_1+2d_2+d_3-U_2$ in two ways.

Write $\bar L_{r,s}$ for the mean of $L_J$ over $J$ containing $r$ sites in $A$ and $s$ in $A^c$, and set $\Lambda=\sum_{p=0}^3\binom3p\bar L_{2+p,6-p}\ge0$. With $w_{ij}$ the total support mass on $i$ sites in $A$ and $j$ in $A^c$, expansion gives
\begin{align*}
 Y={}&100(1+\Lambda)+400w_{01}+360w_{10}
       +140w_{02}+140w_{11}+110w_{20}\\
    &+45w_{03}+52w_{12}+45w_{21}+30w_{30}
       +3w_{13}+4w_{22}+5w_{31}\ \ge100 .
\end{align*}
The coefficient of the mean purity $\pi_{kl}$ on $k$ sites in $A$ and $l$ in $A^c$ in $\Lambda$ is $[x^ky^l]8(1-x)^2(1-y)^3(1-xy)^3$. Pair complementary purities, then insert $2^{k+l}\pi_{kl}=\sum_{i,j}\binom ki\binom lj w_{ij}/[\binom5i\binom6j]$ for $k+l\le5$, including $w_{00}=1$. For $k+l=5$, multiplying $2^5\pi_{kl}-1$ by $\binom5k\binom6l$ gives $d_k$.

The radius-two neighborhood of $A$ consists of classes $0,4,5$ with sizes $6,30,1$. Thus $b_1+b_2+b_3\le28$. Each neighbor of $A$ has defect greater than $1+30\varepsilon$. By Eq.~\eqref{eq:eleven-local}, one of its five four-site supports has mass greater than $\varepsilon$. A four-site support in $A^c$ occurs in only two neighbors. There are at least three distinct such supports, each with five bad extensions in class~1, giving $b_1\ge15$.

Every class-$1,2,3$ defect is at most $3+4\varepsilon$. Otherwise a second radius-two neighborhood would be bad. For $|A\cap B|=1,2,3$ the two neighborhoods have $4,0,4$ common vertices, respectively, so their unions have $70,74,70$ vertices. For intersection one, a common label must avoid one center and contain the four sites outside it from the other center. For intersection three, it must contain the three common sites and one site from each difference.

A class-1 vertex with defect above $1+2\varepsilon$ has four bad neighbors in class~3, and a class-3 vertex has two neighbors in class~1. At most $b_3/2$ class-1 marginals thus have defect above $1+2\varepsilon$. Writing $(N_1,N_2,N_3)=(75,200,150)$ for the class sizes and allowing each good marginal defect $\varepsilon$, we obtain
\begin{align*}
 d_1&\le b_1+b_3+(75+b_1+b_3)\varepsilon,\\
 d_j&\le 3b_j+(N_j+3b_j)\varepsilon\quad(j=2,3),\\
 d_2-2U_2&\le b_2+(200+19b_2)\varepsilon .
\end{align*}
The last line sums Eq.~\eqref{eq:eleven-local} on the bad class-2 marginals. Using $2d_2-U_2=\tfrac32 d_2+\tfrac12(d_2-2U_2)$ gives
\[
 Y\le2b_1+5b_2+5b_3+(700+2b_1+14b_2+5b_3)\varepsilon
 \le95+912\varepsilon<100 ,
\]
where $b_1\ge15$, $b_1+b_2+b_3\le28$, and $\varepsilon\le1/200$. This contradicts the weighted shadow bound.
\end{proof}

\subsubsection{Counting the heavy extensions}\label{sec:eleven-radius}
\begin{theorem}\label{cor:eleven-robust}
For $0\le\varepsilon\le2^{-12}$,
\[\boxed{\Qrob{\varepsilon}(11,5)=396.}\]
\end{theorem}
\begin{proof}
Write $p=2^{-12}$ and suppose at this tolerance that the number $b$ of bad five-sets is at most 65. Call a four-set $R$ marked if its exact Pauli mass $w_R$ exceeds $p$. Every five-set containing a marked four-set is bad. Put $c=1+30p$ and call a five-set $A$ heavy if $D_A>c$. Every five-set disjoint from a heavy set is bad by Eq.~\eqref{eq:pre-disjoint}. Write $x_A=(D_A-c)_+$. At most five four-site masses in Eq.~\eqref{eq:eleven-local} are discarded when only marked supports are retained. Each discarded mass is at most $p$, so their contribution after multiplication by two is at most $10p$. Thus
\[
 x_A\le 2\sum_{\substack{R\in\binom A4\\ w_R>p}}w_R .
\]
Let $t_R$ count the heavy sets among the seven five-site extensions of a marked $R$. Each good five-set contributes at most $p$ to $S_5$, and each bad five-set contributes at most $c+x_A$. Hence
\[
 66+\tfrac{22}{5}W_4\le S_5\le65+2347p+
     2\sum_{\substack{|R|=4\\w_R>p}}t_Rw_R .
\]
The first inequality is Eq.~\eqref{eq:eleven-global}. If every $t_R\le2$, the last sum, including its factor two, is at most $4W_4$, in contradiction with $2347p<1$. Choose a marked $R$ maximizing $t=t_R$ and write $V=R^c$.

\emph{Three or four heavy extensions.}
Choose a triple $E$ uniformly from the $t$ heavy extension sites of $R$ and partition the sites into $R,E,V\setminus E$, of sizes $4,3,4$. In each such partition, count the excesses $x_A$ against the marked four-sets in $A$ and count the heavy sets disjoint from each bad five-set. Section~\ref{sec:eleven-weights} gives a fixed nonnegative integer combination $F$ of the resulting inequalities and the global shadow rows. Averaging over $E$ gives
\[
 0\le\mathbb E_E F\le
 \begin{cases}
 -134190073/2048,&t=3,\\
 -561180329/8192,&t=4 .
 \end{cases}
\]
The same coefficients apply in both cases. This excludes $t=3,4$.

\emph{Five or more heavy extensions.}
For $i=0,\ldots,4$, let $b_i,h_i$ count the bad and heavy five-sets, respectively, meeting $R$ in $i$ sites. First we show that $w_K\le c$ for every four-set $K\subset V$. If $w_K>c$, all seven extensions of $K$ are heavy. This contradicts maximality when $t<7$. When $t=7$, let $N_L$ be the five-sets containing or avoiding $L$, for $L=R,K$. Every member of $N_L$ is bad. In particular, a five-set avoiding $L$ misses two sites outside $L$ and is disjoint from the heavy extension through either one. Since $R,K$ are disjoint, $|N_R\cup N_K|=28+28-6=50$. The other 412 five-sets therefore include at most 15 bad sets.

A five-set outside this union has at least five disjoint neighbors outside it. In fact, it meets each of $R,K$ in one to three sites. Of the six ways to delete one site from its complement, a deletion can leave the family only if the five-set meets one of $R,K$ in three sites and the deletion removes its unique remaining site. At most one deletion does so. Two different five-sets have at most one common disjoint neighbor. If the outside family contained eight heavy sets, they would have at least 40 incidences with at most 15 outside bad sets. For their incidence counts $d_B$, convexity of $\binom d2$, padding the list to length 15 by zeros, gives
\[
 \sum_B\binom{d_B}{2}\ge 10\binom32+5\binom22=35>\binom82 .
\]
This contradicts the common-neighbor bound. At most seven of the outside sets can be heavy. Let $Y$ be the sum of the defects of all 412 outside sets. Lemma~\ref{lem:eleven-deep} bounds any individual defect by $3+60p$, so
\[
 Y\le 15+(412+29\cdot15)p+7(2+30p)
 =29+1057p<32.
\]
The five global shadow rows in Section~\ref{sec:eleven-weights} give $Y\ge32$ for every pure state and every disjoint pair $R,K$. This proves the asserted mass bound.

Let $U=\sum_{K\in\binom V4,\,w_K>p}w_K$. Each marked $K\subset V$ has four bad extensions $K\cup\{r\}$, $r\in R$, and each such label determines $K$ uniquely. Thus $4U\le c b_1$. The first two rows of Table~\ref{tab:elw-high} give $b_3\le7$, with $h_1\ge5$ at $t=5$ and $h_1\ge4$ at $t=6,7$. The latter bounds hold for both degrees because $b_0=21$ whenever $t\ge6$: every five-set in $V$ omits at least one heavy extension site. Seven is a common upper bound on the number of heavy extensions of any four-set.

For each $r\in R$, take a separate graph on $V$. Join $i$ and $j$ if $(R\setminus\{r\})\cup\{i,j\}$ is bad. The four graphs have a total of $b_3$ edges. A heavy set $\{r\}\cup K$ forces a triangle in the graph for $r$, since every pair of sites in $V\setminus K$ labels a five-set disjoint from this heavy set. Different heavy sets give different triangles. A simple graph with at most seven edges and at least four triangles contains a $K_4$. Indeed, two triangles must share an edge, giving five edges on four vertices. Without the missing edge on these four vertices, a third triangle would require a new vertex and both remaining edges, leaving no fourth triangle. A $K_4$ accounts for four triangles, and at most one extra edge cannot create another. Applying this argument to the disjoint union of the four graphs excludes $t=5$.

For $t=6,7$, the same count forces $h_1=4$. The four heavy labels correspond to the triangles of one $K_4$ in one graph, so they have the same site $r\in R$ and distinct four-sets in $V$. Each marked $K\subset V$ therefore has at least three nonheavy bad extensions among the $b_1-4$ such labels. These extensions from different $K$ are distinct. It follows that $3U\le c(b_1-4)$. The last row of Table~\ref{tab:elw-high}, now applicable because $h_1=4$, gives $3U>c(b_1-4)$. This contradiction completes the upper bound. At $p=2^{-12}$ at least 66 five-sets are bad, so at most $462-66=396$ are good. This upper bound holds at every smaller tolerance by monotonicity.

For attainment, project one qubit of the five-uniform dodecacode graph state onto $\ket0$~\cite{Scott2004,Zhang2025}. The resulting graph state is four-uniform, since the measured qubit together with any four remaining qubits had a maximally mixed marginal. Hence $W_1=\cdots=W_4=0$. The two global shadow rows of Lemma~\ref{lem:eleven-shadow} give $S_5=W_5=66$. By the upper bound just proved, at least 66 five-body defects are positive. Each positive stabilizer defect is an integer by Eq.~\eqref{eq:pre-stab-defect}, so exactly 66 equal one and the remaining 396 vanish. This attains 396 throughout the stated interval. It also attains $\max\bar L_5=27/28$ by Eqs.~\eqref{eq:eleven-global} and~\eqref{eq:meanpurity}.
\end{proof}

\subsubsection{Weighted inequalities for the eleven-qubit count}\label{sec:eleven-weights}
We give the coefficients used to rule out $b\le65$ at $p=1/4096$. A marked four-set $R$ has exact Pauli mass $w_R>p$. A heavy five-set has defect greater than $1+30p$. Choose $R$ to maximize the number $t$ of its heavy five-site extensions.

\paragraph{Purity and incidence inequalities.}
Partition the sites into blocks of sizes $\boldsymbol n$. A type $r$ gives the number of sites of a subset in each block. Write $m_r=\binom{\boldsymbol n}{r}$, with vector binomials multiplied coordinatewise and invalid binomials set to zero. At weights three and four retain exact support masses greater than $p$, and at weight five retain all masses. At weights one and two every mass is at most $p$, since a defective pair alone forces $\binom93=84$ bad five-sets. Let $a_r$ and $e_r$ be the type means of the retained and discarded masses, respectively, using zero for a missing part. Put $a_0=1$, $u_r=m_ra_r$, and $P_v=2^{|v|}\overline{\operatorname{Tr}\rho_A^2}$, with the mean over labels $A$ of type $v$. The support expansion and complementary purities give
\[
P_v=\sum_{r\le v}\binom vr(a_r+e_r)\quad(|v|\le5),\qquad
P_v=2^{2|v|-11}P_{\boldsymbol n-v}\quad(|v|>5).
\]
For each five-site type $k$, let $b_k,h_k$ count its bad and heavy labels. Its total defect is $d_k=m_k(P_k-1)$, and its total excess is $x_k=\sum_{A\text{ of type }k}(D_A-1-30p)_+$. We use the following nonnegative expressions:
\[
\begin{aligned}
S_j&=32\sum_{v\le\boldsymbol n}2^{-|v|}P_v
       \prod_i[z^{v_i}](1-z)^{j_i}(1+z)^{n_i-j_i} &&(|j|\equiv0\pmod2),\\
U_j&=\sum_{v\le j}(-1)^{|j|-|v|}\binom jvP_v &&(|j|>5),\\
D_k^*&=(1+29p)b_k+x_k+m_kp-d_k, &
X_k^*&=(2+30p)h_k-x_k,\\
H_k&=b_k-h_k, &T&=65-\sum_k b_k .
\end{aligned}
\]
Here $S_j$ is a type average of the global shadow, and $U_j$ is the exact support mass obtained by inversion. Separating good, bad and heavy labels gives $D_k^*\ge0$. The individual bound $D_A\le3+60p$ gives $X_k^*\ge0$.
For a family $J$ of five-site types put $E_{rk}=\binom{\boldsymbol n-r}{k-r}$ and $N_{jk}=\binom{\boldsymbol n-j}k$. Suppose $s$ bounds the number of heavy extensions of any marked four-set. For $|r|=4$ define
\[
i_r(s,J)=\min\!\left(s,\sum_{k\in J}E_{rk}\right),\qquad
g_r(s,J)=\max_{\substack{0\le y\le\boldsymbol n-r\\|y|=s}}\sum_{k\in J}\binom{y}{\boldsymbol n-r-k}.
\]
For the distinguished four-set $R$, whenever the block counts $v_*$ of its extension sites are known, use $y=v_*$ in $g_r$; a specified set of $s$ sites containing all its heavy extension sites also suffices. The following two incidence inequalities hold:
\[
\begin{aligned}
I_J&=2\sum_{|r|=4}i_r(s,J)u_r-\sum_{k\in J}x_k\ge0,\\
G_J&=\sum_{k\in J}(b_k+2d_k)+
 2\sum_{|r|=3}u_r\min\!\left(37,\sum_{k\in J}N_{rk}\right)
 +\sum_{|r|=4}g_r(s,J)u_r\\
&\hspace{2em}-\sum_{|j|=5}\sum_{k\in J}N_{jk}(h_j+x_j)\ge0 .
\end{aligned}
\]
For $I_J$, the scalar shadow bounds the excess of a heavy set by twice the marked four-site mass it contains. A four-set has at most $s$ heavy extensions and at most $\sum_{k\in J}E_{rk}$ extensions in $J$. For $G_J$, fix a bad five-set $B$ of type in $J$ and collect the heavy $A$ disjoint from it. Each $A$ carries retained mass at least $D_A-30p=1+x_A$, where $x_A=D_A-1-30p$. Counting each distinct support in $B^c$ once costs at most $D_{B^c}=1+2D_B$. A five-site support never repeats. A marked triple lies in 28 bad five-sets, so at most $65-28=37$ bad $B$ can avoid it. In any one $B^c$ it occurs in at most three five-sets, giving at most two repeats. A four-set $M$ is counted twice only if both $M+i$ and $M+j$ are heavy and $B^c=M+i+j$. The pair of extension sites then determines $B$. Summing these repeats across $J$ with the same set of extension sites gives $g_r(s,J)$. Every five-site neighbor disjoint from a heavy label is bad, so summing over $B$ counts all incidences on the last line.
All remaining variables are nonnegative. Their individual upper bounds, always taken with the multiplicities of the partition in use, are
\[
\bar e_r=p,\quad \bar a_r=2^{|r|}-1\ (r\ne0),\quad
\bar b_k=\bar h_k=m_k,\quad \bar x_k=(2+30p)\min(65,m_k).
\]
Thus an affine form $c+\sum_z\alpha_z z$ is at most $c+\sum_z(\alpha_z)_+\bar z$ and at least $c+\sum_z\min(0,\alpha_z)\bar z$.

\paragraph{A common sum for three and four heavy extensions.}
Select three of the $t$ heavy extension sites uniformly to form $E$, and use the blocks $(R,E,(R\cup E)^c)$ of sizes $(4,3,4)$. Take $s=t$, and for $R$ use $v_*=(0,3,t-3)$. Before averaging, set $b_k=m_k$, $h_k=k\cdot v_*$ when $k_1=4$. Set $b_k=m_k$ when $k_1=0$ and $k_2<3$, since such a label avoids $R$ and misses a marked extension site, making it disjoint from a heavy label.
To average over $E$, give each label an intrinsic type $j$ according to its intersection with $R$, all $t$ heavy extension sites and the remaining sites. Write $\boldsymbol n'=(4,t,7-t)$, $m'_j=\binom{\boldsymbol n'}j$. The number of labels of both coarse type $r$ and intrinsic type $j$ for fixed compatible blocks is
\[
M_{rj}=\mathbf1_{j_1=r_1,\ |j|=|r|}\binom4{r_1}\binom3{r_2}
      \binom{t-3}{j_2-r_2}\binom{7-t}{j_3}.
\]
Therefore each type mean $z_r$ averages to $\sum_j M_{rj}z'_j/m_r$, and each type total to $\sum_{j:m'_j>0} M_{rj}z'_j/m'_j$. In intrinsic types with $j_1=4$, substitute $b'_j=m'_j$, $h'_j=m'_j\mathbf1_{j_2=1}$. For $j_1=0$, $j_2<t$, substitute $b'_j=m'_j$.
The nonzero coefficients of the weighted sum $F$ are listed in Table~\ref{tab:elw-low}. Each three-digit string denotes a type. In subscripts of $I,G$, commas join the members of a family $J$. We abbreviate $\mathcal J_a=\{k:|k|=5,\ k_1=a\}$ and $\mathcal J_{ab\cdots}=\mathcal J_a\cup\mathcal J_b\cup\cdots$. For example, $G_{032,230}$ has weight $15719$.
\begin{table}[htbp]\centering\small
\caption{Integer weights of the common sum at $t=3,4$. The middle panel gives two coefficients for each type $k$. Blank entries have weight zero.}\label{tab:elw-low}
\setlength{\tabcolsep}{4pt}\begin{tabular}{@{}lr@{\quad}crr@{\quad}lr@{}}\toprule
Row & Weight & $k$ & $D^*_k$ & $X^*_k$ & Row & Weight\\\midrule
$S_{134}$ & 594113 & $032$ & 731198 &  & $G_{032}$ & 28550\\
$S_{224}$ & 1417778 & $104$ & 908977 & 302254 & $G_{032,230}$ & 15719\\
$S_{233}$ & 1811948 & $113$ & 887232 & 295024 & $G_{032,302}$ & 34012\\
$S_{314}$ & 524050 & $122$ & 861120 & 295712 & $G_{\mathcal J_{123}}$ & 41048\\
$S_{323}$ & 2771295 & $131$ & 952210 & 228096 & $G_{104,302}$ & 43539\\
$S_{332}$ & 1484584 & $203$ & 952210 & 171203 & $G_{113,122}$ & 62703\\
$S_{334}$ & 56686 & $212$ & 952210 & 213133 & $G_{113,302}$ & 2735\\
$S_{404}$ & 95351 & $221$ & 925184 & 187895 & $G_{122,221}$ & 849\\
$S_{413}$ & 1057537 & $230$ & 936601 & 166075 & $G_{221,311}$ & 932\\
$S_{422}$ & 1639571 & $302$ & 794055 & 236224 & $G_{221,320}$ & 25436\\
$S_{431}$ & 367583 & $311$ & 872850 & 262425 & $G_{\mathcal J_3}$ & 78989\\
$S_{433}$ & 29659 & $320$ & 848519 & 254335 & $I_{203}$ & 49063\\
$U_{402}$ & 4189 & $401$ & 628840 & 393994 & $I_{\mathcal J_2}$ & 227598\\
$U_{420}$ & 312061 &  &  &  & $I_{\mathcal J_{23}}$ & 83652\\
$H_{122}$ & 28184 &  &  &  & $I_{203,230}$ & 77034\\
$T$ & 1000000 &  &  &  & $I_{221,230}$ & 48875\\
\bottomrule\end{tabular}\end{table}
Expanding $\mathbb E_E F=c+\sum_z\alpha_z z$ with these substitutions, and using the intrinsic multiplicities in $\bar z$, gives
\[
 -c-\sum_z(\alpha_z)_+\bar z =
 \begin{cases}134190073/2048,&t=3,\\561180329/8192,&t=4.\end{cases}
\]
Both upper bounds for $\mathbb E_E F$ are strictly negative, whereas each row of $F$ is nonnegative. This excludes both degrees.

\paragraph{The defect outside two disjoint four-sets.}
Let $R,K$ be disjoint four-sets and partition the sites into $R,K,(R\cup K)^c$, of sizes $(4,4,3)$. Write $Z_{abc}=(S_{abc}+S_{bac})/64$. The sum $Y$ of five-site defects on labels meeting both $R$ and $K$ in one, two or three sites is $\sum m_v(P_v-1)$ over the corresponding five-site types. Subtract from $Y$ the nonnegative expression
\[
16\big(3Z_{323}+Z_{330}+4Z_{413}+9Z_{422}+3Z_{431}\big).
\]
The constant term is $32$. The following list gives the type and coefficient of every remaining nonzero term. The variable for $abc$ is the total exact Pauli mass on types $abc$ and $bac$, with each support counted once.
\begin{center}\small\begin{tabular}{@{}cl@{}}\toprule
Weight & (Type, coefficient)\\\midrule

1 & $(001,196),\ (100,218)$\\
2 & $(002,60),\ (101,75),\ (110,353/4),\ (200,70)$\\
3 & $(003,20),\ (102,21),\ (111,119/4),\ (201,22),\ (210,63/2),\ (300,21)$\\
4 & $(103,3),\ (112,19/4),\ (211,13/2),\ (220,6),\ (301,1),\ (310,21/4)$\\
5 & $(113,1/4),\ (311,5/4),\ (401,4),\ (410,1)$\\
\bottomrule\end{tabular}\end{center}
These coefficients follow by expanding $P_v$ in the full exact masses, pairing the first two blocks and dividing each type-mean coefficient by its orbit multiplicity. Hence $Y\ge32$ for every pure eleven-qubit state.

\paragraph{Reading the five intersection layers at $t\ge5$.}
Partition the sites into $R$ and $R^c$. In layer $i=|A\cap R|$ let $b_i,h_i$ count the bad and heavy five-sets, respectively. Each heavy label in layer one forces a distinct triangle in the four graphs on $R^c$; together the graphs have exactly $b_3$ edges. Let $U=35a_{04}$ be the total marked four-site mass inside $R^c$ and put $c=1+30p$. Table~\ref{tab:elw-high} gives the bounds needed to exclude $b\le65$.
\begin{table}[H]\centering\small
\caption{How the weighted inequalities exclude the cases $t=5,6,7$. The last row uses the equality $h_1=4$ from the preceding row.}\label{tab:elw-high}
\begin{tabularx}{\linewidth}{@{}l l X@{}}\toprule
Case & Bound from the inequalities & Consequence for the count\\\midrule
$t=5$ & $h_1\ge5,\ b_3\le7$ & Seven edges support at most four triangles.\\
$t=6,7$ & $h_1\ge4,\ b_3\le7$ & Four triangles on at most seven edges form a single $K_4$, so $h_1=4$.\\
$t=6,7,\ h_1=4$ & $3U>c(b_1-4)$ & Counting the bad extensions outside the four heavy labels gives $3U\le c(b_1-4)$.\\\bottomrule
\end{tabularx}\end{table}
\paragraph{Verification of Table~\ref{tab:elw-high}.}
Write $k_i=(i,5-i)$, $x_i=x_{k_i}$ and $G_i=G_{\{k_i\}}$. For the distinguished support $R$ take $y=(0,s)$ in $g_{(4,0)}(s,\{k_i\})$. We use the following additional nonnegative expressions.
\[
\begin{aligned}
L_i&=2\sum_{|r|=4}E_{r k_i}u_r-x_i,\qquad C=c b_1-4U,\\
I_{23}&=2\sum_{|r|=4}\min(s,E_{r k_2}+E_{r k_3})u_r-x_2-x_3.
\end{aligned}
\]
The excess in a single layer gives $L_i\ge0$; counting both layers together gives $I_{23}\ge0$. Each marked four-set in $R^c$ has mass at most $c$ and forces four distinct bad labels in layer one. Each such label determines the four-set, so $C\ge0$. All seven extensions of $R$ are bad and exactly $t$ are heavy. A five-set in $R^c$ is bad if either omitted site is a heavy extension site, giving $b_0\ge21-\binom{7-t}2$.
For $t=5$ set $s=5$. At both $t=6$ and $t=7$, allowing all seven extension sites gives the same upper bounds on excesses and repeated incidences and the same counts $b_0=21$, $b_4=7$. Thus take $s=7$ for both. In the order of the three rows of Table~\ref{tab:elw-high}, the five quantities to be proved positive are
\[
(\ell_1,\ldots,\ell_5)=(h_1-4,\ 8-b_3\mid h_1-3,\ 8-b_3\mid U-c(b_1-4)/3).
\]
Only $\ell_5$ uses $h_1=4$. The first four signs give the stated bounds by integrality. To check all five signs, order the nonnegative expressions in four blocks,
\[
\begin{aligned}
\mathbf r={}&(S_{06},S_{17},S_{26},S_{35},S_{37},S_{44})\mid(D_1^*,D_2^*,D_3^*,D_4^*,X_1^*,X_2^*,X_3^*)\\
&\mid(L_2,L_3,G_0,G_1,G_2,G_3,I_{23})\mid(C,T).
\end{aligned}
\]
The vector $\mathbf y_j$ gives the multipliers for $\ell_j$ in this order. Each line clears its denominator; the vertical bars separate the four blocks of $\mathbf r$.
{\footnotesize\thinmuskip=1mu\relax
\[
\begin{aligned}
131\,\mathbf y_1&=\big(22,225,1428,1632,247,614\mid155,275,294,96,95,66,68\mid0,0,16,27,20,0,75\mid113,297\big)\\[1pt]
113\,\mathbf y_2&=\big(0,289,1725,1638,257,807\mid184,373,289,0,61,97,69\mid0,0,0,28,35,6,80\mid196,411\big)\\[1pt]
169\,\mathbf y_3&=\big(0,198,1383,1326,245,608\mid156,268,284,0,108,69,94\mid59,0,0,38,16,0,0\mid90,286\big)\\[1pt]
112\,\mathbf y_4&=\big(0,175,999,951,143,454\mid121,227,143,0,40,68,37\mid0,10,0,15,27,0,21\mid118,256\big)\\[1pt]
2692\,\mathbf y_5&=\big(0,1247,9180,8904,1722,3790\mid862,1830,1982,0,403,528,557\mid242,305,0,224,153,0,0\mid0,1996\big).
\end{aligned}
\]
}
Expand $\ell_j-\mathbf y_j\cdot\mathbf r=\alpha_0+\sum_z\alpha_z z$ at $a_0=1$ and $p=1/4096$. In the intervals $[\underline z,\bar z]$ given above for the two blocks, impose $b_4=7$ in every case. For $t=5$ use $b_0\ge20$ and $h_4=5$; for $t=6,7$ use $b_0=21$. In the final line also set $h_1=4$. Each $\ell_j$ is at least $\alpha_0+\sum_z\min(\alpha_z\underline z,\alpha_z\bar z)$. The smallest of these five lower bounds is $\frac{4723}{16539648}>0$, attained for $\ell_5$.

\section{Local defect bounds and incidence counts at general sizes}\label{sec:general}
For a fixed pure state let $k=\lfloor n/2\rfloor$ and let $b$ be the number of bad $k$-body marginals. Choose a uniform $\ell$-set $T$ and let $b_T$ count its bad $k$-subsets. Then
\begin{equation}\label{eq:general-incidence}
 \frac{b}{\binom nk}=\mathbb E_T\!\left[\frac{b_T}{\binom\ell k}\right],
\end{equation}
because each $k$-set belongs to $\binom{n-k}{\ell-k}$ choices of $T$. The local shadow argument of Section~\ref{ch:4m} leaves exactly this incidence count at $n=4m$. With $k=2m$ and $0\le\varepsilon<1/(k+1)$, every $(k+1)$-set contains a $k$-body marginal with defect above $\varepsilon$. Taking $\ell=k+1$ gives $b/\binom nk\ge1/(k+1)$.

For the other two size families take $\ell=k+2$. At odd sizes, nonnegative Pauli support weights give a two-site deletion inequality. Pure-state compatibility relates the internal $k$-sets to the one-site extensions of $T^c$, and the final count weighs both families of labels. At $n\equiv2\pmod4$, a scalar shadow inequality on $\rho_T$ strengthens a deletion-moment bound. One-site replacements within $T$ then force a minimum number of bad $k$-subsets in each $T$. Counting these incidences and rounding the bad count up to an even number completes the bound. Both arguments retain an explicit positive tolerance.

\subsection{Odd sizes and weighted counts of internal and external marginals}\label{ch:odd}\label{app:odd}

Fix a pure state $\rho$ on $n=2k+1$ qubits, $k\ge6$, and one of its $(k+2)$-body marginals $\rho_T$. Let $C=T^c$, $H=\binom{k+2}{2}$ and $N_k=\binom{2k+1}{k}$. The $H$ internal $k$-body marginals have labels $T\setminus\{i,j\}$ for distinct $i,j\in T$. The $k+2$ external ones have labels $C\cup\{i\}$, with $i\in T$. Write $x$ and $y$ for the respective numbers of bad marginals.

\begin{lemma}\label{lem:odd-deletion}
For every such $(k+2)$-body marginal $\rho_T$, the defects of its $k$-body marginals satisfy
\begin{equation}\label{eq:odd-local-budget}
 S_T:=\sum_{A\in\binom Tk}D_A\ge(k-5)+2(k-2)D_C.
\end{equation}
\end{lemma}
\begin{proof}
Apply the two-site deletion bound~\eqref{eq:pre-delete-two} to $\rho_T$ with $|T|=k+2$. In Eq.~\eqref{eq:complement} the dimension factors for $T$ and $T\setminus\{i\}$ are eight and two, respectively. Substituting gives
\begin{equation}\label{eq:odd-two-delete}
 S_T\ge\sum_{i\in T}\bigl(1+2D_{C\cup\{i\}}\bigr)-(7+8D_C).
\end{equation}
By monotonicity~\eqref{eq:monotonicity}, every $D_{C\cup\{i\}}\ge D_C$. There are $k+2$ terms, giving Eq.~\eqref{eq:odd-local-budget}.
\end{proof}

If almost all the external marginals are good, complementarity limits each internal defect, so more internal marginals must be bad. Counting the external failures as well gives one inequality valid in both regimes.

\begin{lemma}\label{lem:odd-weight}
For an integer $1\le r\le k-5$ and
\begin{equation}\label{eq:odd_eps}
 0\le\varepsilon<\frac{k-4-r}{H+r-1},
\end{equation}
the counts $x,y$ for every choice of $T$ satisfy
\begin{equation}\label{eq:weighted_odd}
 x+\frac{r-1}{2}y\ge r .
\end{equation}
\end{lemma}
\begin{proof}
Since $H\varepsilon<k-5$, Lemma~\ref{lem:odd-deletion} rules out $x=0$. For $y\ge2$ this suffices. Suppose $y\le1$. Each internal $A=T\setminus\{i,j\}$ is disjoint from the two external sets $C\cup\{i\}$ and $C\cup\{j\}$, so one of them, say $B$, is good. For these disjoint labels, Eq.~\eqref{eq:pre-disjoint} with $n=2k+1$ and $|B|=k$ gives
\begin{equation}\label{eq:oddprop}
 D_A\le D_{B^c}=1+2D_B .
\end{equation}
Each bad internal marginal therefore has defect at most $1+2\varepsilon$. If $x\le r-1$, then
\[
 S_T\le x(1+2\varepsilon)+(H-x)\varepsilon
 \le(r-1)+(H+r-1)\varepsilon<k-5 ,
\]
contradicting Lemma~\ref{lem:odd-deletion}. Thus $x\ge r$ when $y\le1$, proving Eq.~\eqref{eq:weighted_odd}.
\end{proof}

\begin{theorem}\label{thm:odd-count}
For the parameters of Lemma~\ref{lem:odd-weight},
\begin{equation}\label{eq:odd-weighted-bound}
 \Qrob{\varepsilon}(2k+1,k)\le N_k-\left\lceil\frac{2rN_k}{(k+2)(k+r)}\right\rceil .
\end{equation}
\end{theorem}
\begin{proof}
Over all $T$, each bad $k$-set occurs $\binom{k+1}{2}$ times internally and $k$ times externally. There are $kN_k/(k+2)$ choices of $T$. Summing Eq.~\eqref{eq:weighted_odd}, the total bad count $b$ obeys
\[
 \left[\binom{k+1}{2}+\frac{(r-1)k}{2}\right] b
 \ge \frac{rkN_k}{k+2}.
\]
Rounding $b$ up gives Eq.~\eqref{eq:odd-weighted-bound}.
\end{proof}

\subsection{Deletion moments and complementary pairs at sizes congruent to two modulo four}\label{ch:even2}\label{app:even2}

Fix a pure state $\rho$ on $n=2s$ qubits, with odd $s\ge5$, and any $t$-body marginal $\rho_T$, where $t=s+2$. The complementary marginal $\rho_C$ has labels $C=T^c$ and defect $q=D_C$. Write $H=\binom t2$, $a_w=a_w(T)$ and $\ell=t-w$. Equation~\eqref{eq:complement} at sizes $s+2$ and $s+1$ gives $D_T=15+16q$ and $D_{T\setminus\{i\}}=3+4D_{C\cup\{i\}}\ge3+4q$ by monotonicity. Deleting zero, one, or two sites from $\rho_T$ in Eq.~\eqref{eq:pre-deletion} therefore gives
\begin{align}
 M_0&=\sum_{w=1}^{t}a_w=15+16q,\qquad
 M_1=\sum_{w=1}^{t}(t-w)a_w\ge3t+4tq,\nonumber\\
 M_2&=\sum_{w=1}^{t}\binom{t-w}{2}a_w=\sum_{A\in\binom Ts}D_A.\label{eq:moment2}
\end{align}
\begin{lemma}\label{lem:even-moments}
For each integer $j\ge1$, let $P_j(t)=3jt-\tfrac{15}{2}j(j+1)$. For odd $j\ge1$, let $R_j(t)=3t(j+\tfrac12)-\tfrac{15}{2}(j+1)^2+\tfrac72$. Define $L_s$ as the largest of the following admissible values:
\[
 \begin{aligned}
 P_j(t)&\quad\text{if } t-2(j+1)\ge0,\\
 R_j(t)&\quad\text{if }4t(j+\tfrac12)-8(j+1)^2+4\ge0 .
 \end{aligned}
\]
Then $M_2\ge L_s$. In particular $L_s\ge(9s-35)/2$ for all odd $s\ge5$.
\end{lemma}
\begin{proof}
For $j\ge1$, the integer inequality $\binom\ell2\ge j\ell-j(j+1)/2$ follows by factoring the difference as $(\ell-j)(\ell-j-1)/2$. Hence
\[
 M_2\ge P_j(t)+4j[t-2(j+1)]q.
\]
For odd $j\ge1$, one also has
\[
 \binom\ell2\ge(j+\tfrac12)\ell-\tfrac12(j+1)^2
                   +\tfrac14+\tfrac14(-1)^{t-\ell}.
\]
The difference is $(\ell-j-1)^2/2-[1+(-1)^{t-\ell}]/4$. When $t-\ell$ is even, $\ell$ is odd and the square is at least one. Otherwise the second term vanishes. Multiply by $a_w$ and sum over $w=1,\ldots,t$. For this $t$-body marginal the scalar shadow~\eqref{eq:shadow0} gives $\sum_{w=1}^t(-1)^wa_w\ge-1$. Substituting Eq.~\eqref{eq:moment2} yields
\begin{align*}
 M_2&\ge R_j(t)+[4t(j+\tfrac12)-8(j+1)^2+4]q.
\end{align*}
In each admissible case the coefficient of $q=D_C\ge0$ is nonnegative, so taking the maximum gives the claim. Both $P_1(t)=3(s-3)$ and $R_1(t)=(9s-35)/2$ are admissible. For odd $s\ge7$, the latter is strictly larger. For $s=5,7,9,11,13$ the respective maxima are $6,14,23,33,45$.
\end{proof}

To convert $L_s$ into a count, compare two $s$-sets that differ by one site. A bound on the defect of a good neighbor also limits the defect of the original set.

\begin{lemma}\label{lem:even-replacements}
Every $(s+2)$-body marginal $\rho_T$ of $\rho$ has at least $r$ bad $s$-body marginals if
\[
 1\le r\le\min\{2s+1,\lceil L_s/3\rceil\},\qquad
 0\le\varepsilon<\frac{L_s-3(r-1)}{H+3(r-1)},\qquad r\in\mathbb Z .
\]
\end{lemma}
\begin{proof}
For a one-site replacement $A\to B$, $|A\cup B|=s+1$ and $(A\cup B)^c\subset B^c$. Equations~\eqref{eq:monotonicity} and~\eqref{eq:complement}, with $D_{B^c}=D_B$, give
\begin{equation}\label{eq:evenprop}
 D_A\le D_{A\cup B}=3+4D_{(A\cup B)^c}\le3+4D_B .
\end{equation}
Every internal $s$-set has $2s$ one-site replacements inside $T$ (delete one of its $s$ sites and insert either of the two missing sites). If only $x\le r-1\le2s$ internal sets are bad, any bad set has at most $x-1$ bad neighbors and hence a good neighbor. Equation~\eqref{eq:evenprop} then bounds the defect of the corresponding marginal by $3+4\varepsilon$. Summing the $x$ bad and $H-x$ good contributions gives
\[
 M_2\le3x+(H+3x)\varepsilon
 \le3(r-1)+[H+3(r-1)]\varepsilon<L_s ,
\]
contradicting Lemma~\ref{lem:even-moments}.
\end{proof}

\begin{theorem}\label{thm:even-count}
For the parameters of Lemma~\ref{lem:even-replacements}, with $N=\binom{2s}s$,
\begin{equation}\label{eq:even-weighted-bound}
 \Qrob{\varepsilon}(2s,s)\le N-2\left\lceil\frac{rN}{(s+1)(s+2)}\right\rceil .
\end{equation}
\end{theorem}
\begin{proof}
Each $s$-set has $\binom{s}{2}$ extensions $T$, and each $T$ has $H$ internal $s$-sets. By Lemma~\ref{lem:even-replacements} at least the fraction $r/H$ of the incidences involve bad sets. Thus their total number $b$ is at least $rN/H$. Since complementary $s$-body marginals have equal defects, $b$ is even. Rounding up to an even integer gives the stated bound.
\end{proof}

\subsection{Explicit tolerances and improvement over earlier counting bounds}

\subsubsection{Odd sizes}\label{sec:odd-strict}

At $r=1$, the all-good internal clique $\binom Tk$ is forbidden. The earlier zero-error obstruction also included the $k+2$ external sets $C\cup\{i\}$~\cite{Zhang2025}. Taking $r=k-5$ in Theorem~\ref{thm:odd-count} yields the excluded count $\lceil 2(k-5)N_k/[(k+2)(2k-5)]\rceil$ for $\varepsilon<[H+k-6]^{-1}$. The excluded count before rounding increases with $r$. Already at $r=1$ it exceeds the corresponding count in Eq.~(10) of Ref.~\cite{Zhang2025} by
\[
 \frac{2N_k}{(k+1)(k+2)}-\frac{2N_k}{(k+2)(k+3)}
 =\frac{4N_k}{(k+1)(k+2)(k+3)}\ge\frac{k+4}{6}>1 .
\]
Here $N_k\ge\binom{2k+1}{4}\ge\binom{k+4}{4}$ for $k\ge6$. The improvement survives integer rounding.

\subsubsection{Sizes congruent to two modulo four}\label{sec:even-strict}

Set $s=2m+1$ and $r=3m-4$, $m\ge2$. Since $L_s\ge9m-13$, this choice is admissible and its error radius is at least $1/(m^2+7m-6)$. The unrounded bad count $A=2rN/[(s+1)(s+2)]$ is larger than the bound $B=2(s-1)N/[s(s+1)^2]$ in Eq.~(9) of Ref.~\cite{Zhang2025}. Indeed,
\[
 \frac{A}{B}=\frac{(3m-4)(2m+1)(2m+2)}{2m(2m+3)}
 \ge \frac{(2m+1)(2m+2)}{2(2m+3)}>2 .
\]
The last inequality is equivalent to $4m^2-2m-10>0$. Moreover, $N\ge\binom{2s}{2}$ gives $B\ge2(s-1)(2s-1)/(s+1)^2\ge2$ for $s\ge5$. Therefore $2\lceil A/2\rceil>\lceil B\rceil$. At ten qubits, for instance, Eq.~\eqref{eq:even-weighted-bound} lowers the count from $240$ to $228$. The intersecting-core argument in Section~\ref{sec:ten} lowers it further to the attained value $200$. For both families the explicit parameter choices raise the guaranteed bad fraction from order $1/n^2$ in the earlier bounds to order $1/n$.

\section{The asymptotic density of maximally mixed marginals}\label{sec:asymptotic-density}\label{app:density}
At zero tolerance, every subset of a maximally mixed marginal is itself maximally mixed. How much of the half-system layer can these compatibility conditions exclude as $n$ grows? For $k=\lfloor n/2\rfloor$ and a pure state $\rho$, put
\[
 \mathcal H(\rho)=\{A\subseteq[n]: |A|=k,\ \rho_A=I/2^k\},\qquad
 g_n=\max_{\rho\ {\rm pure}}\frac{|\mathcal H(\rho)|}{\binom nk}
     =\frac{\Qex(n,k)}{\binom nk}.
\]
Fixing $k-r$ sites leaves an $r$-uniform hypergraph of possible completions. If a group of sites occurs pairwise in these completions, single-qubit Pauli errors on those sites must carry a common Schmidt support into mutually orthogonal subspaces. The resulting dimension bound limits the density of the completion hypergraph at each depth. At even sizes, projective geometry over $\mathbb F_4$ gives the constants in this recursion, and stabilizer states from quaternary codes reach the resulting density in the limit.

\subsection{A density bound from hypergraph links}
Let $\mathcal H\subseteq\binom Vk$ be any $k$-uniform hypergraph on $n$ vertices. For $S\subset V$ of size $k-r$, its link is
\[
 \mathcal L_S=\{J\subseteq V\setminus S:S\cup J\in\mathcal H\}.
\]
An edge of $\mathcal L_S$ records the $r$ vertices that complete $S$ to an edge of $\mathcal H$. The \emph{2-shadow} (or \emph{2-section}) $G_S$ is the graph on $V\setminus S$ that joins two vertices if they occur together in an edge of $\mathcal L_S$. Here the shadow is a family of vertex pairs. The Rains shadow inequalities used in the finite-size proofs concern marginal purities. A clique of size $t$ is a set of $t$ pairwise adjacent vertices. Each pair lies in a link edge, which may differ from pair to pair. Write $\omega(G_S)$ for the maximum clique size. For nonnegative weights $x_i$ on $V\setminus S$ summing to one, set
\[
 F_S(x)=r!\sum_{J\in\mathcal L_S}\prod_{i\in J}x_i,\qquad
 \Lambda(\mathcal L_S)=\max_x F_S(x).
\]
Each edge can be drawn in $r!$ orders, so $F_S(x)$ is the probability that $r$ independent draws with weights $x$ form an edge. In terms of the usual hypergraph Lagrangian $\lambda$~\cite{PolWalshWigal}, our maximum is $\Lambda(\mathcal L_S)=r!\lambda(\mathcal L_S)$, the blowup density in Ref.~\cite{Keevash2011}.

\begin{lemma}\label{lem:asym-finite-upper}
Fix $1\le r\le k$. Suppose that, for $2\le s\le r$ and every $T\subset V$ of size $k-s$, we have $\omega(G_T)\le M_s$, where $M_s\ge2$. For each $S$ of size $k-r$,
\[
 \Lambda(\mathcal L_S)\le U_r,\qquad
 U_r=\prod_{s=2}^{r}\left(1-\frac{1}{M_s}\right)^{s-1},\qquad U_1=1.
\]
Consequently $|\mathcal H|/\binom nk\le L^r U_r/(L)_r$, with $L=n-k+r$ and $(L)_r=L(L-1)\cdots(L-r+1)$.
\end{lemma}
\begin{proof}
The case $r=1$ follows from $F_S(x)\le1$. For $r\ge2$ we may assume the maximum is positive. Among its maximizers choose one with the fewest positive weights. Any two positive vertices occur together in an edge whose vertices all have positive weights. Otherwise, with the sum of the two weights fixed, $F_S$ is affine in either weight. Moving all their weight to one endpoint would preserve the maximum and reduce the support. Thus the positive support is a clique in $G_S$ and has size at most $M_r$. This is the support reduction used in the graph Lagrangian argument of Motzkin and Straus~\cite{MotzkinStraus1965}.

Choose a largest weight $x_i\ge1/M_r$. Derivatives with respect to positive coordinates agree at the maximum, by varying two weights at fixed sum. Euler's identity $\sum_j x_j\partial_j F_S=rF_S$ sets each of these derivatives to $rF_S(x)$. Terms containing $i$ are indexed by the edges of $\mathcal L_{S\cup\{i\}}$ after deleting $i$. Since $x_i<1$ at a positive maximum for $r\ge2$,
\[
 F_S(x)=\frac{1}{r}\partial_i F_S(x)
 =(1-x_i)^{r-1}F_{S\cup\{i\}}\!\left(\frac{x_{-i}}{1-x_i}\right)
 \le\left(1-\frac1{M_r}\right)^{r-1}U_{r-1},
\]
where $x_{-i}$ omits $i$. Induction gives the product bound. A related recursion uses contraction in a matroid to bound its basis density~\cite{PolWalshWigal}. Here the derivative identity uses only the definition of a hypergraph link.

For uniform weights $1/L$, the $r$ draws are distinct with probability $(L)_r/L^r$. Conditional on distinctness, their set is a uniform $r$-subset of $V\setminus S$. Dividing $F_S$ by that probability therefore gives the edge density in the link. If $S$ is chosen uniformly as well, its union with this uniform extension is a uniform $k$-subset of $V$. Averaging the link bound gives the stated density bound.
\end{proof}

For an example in which basis exchange fails, start with $\ket{\Phi^+}_{13}\ket{\Phi^+}_{24}$ and apply $V=(I+iF)/\sqrt2$ to qubits $3,4$, where $F$ swaps them. The marginals on $12$ and $34$ remain maximally mixed. For $a\in\{1,2\}$, $b\in\{3,4\}$, the Bell-pair identity gives $\langle Z_a Z_b\rangle=\Tr(Z_{a+2}V^\dagger Z_bV)/4$. In the expansion $2V^\dagger Z_bV=Z_b+FZ_bF+i[Z_b,F]$, the commutator has zero trace against $Z_{a+2}$ because the two $Z$ operators commute. Exactly one of $Z_b,FZ_bF$ equals $Z_{a+2}$, so all four cross correlations are $1/2$. The only edges of $\mathcal H(\rho)$ are therefore $\{1,2\}$ and $\{3,4\}$. Replacing $1$ in the first edge by either vertex of the second never gives an edge.

\subsection{Quantum error spaces and the clique bound}\label{app:asym-packing}
For $\mathcal H(\rho)$, write $n=2k+\delta$, $\delta\in\{0,1\}$. Fix $|S|=k-r$ with $2\le r\le k$. If $i,j$ are adjacent in $G_S$, there is an edge $S\cup J$ whose maximally mixed marginal contains both sites. Tracing out $J\setminus\{i,j\}$ gives $\rho_{S\cup\{i,j\}}=I/2^{|S|+2}$. Hence, for a clique of size $t\ge2$, every pair of its sites has this property and $\rho_S$ is maximally mixed. We prove that the possible clique sizes are bounded by
\begin{equation}\label{eq:asym-cap}
 M_r^{(0)}=\frac{4^r-1}{3},\qquad M_r^{(1)}=\frac{2\cdot4^r-5}{3}.
\end{equation}
For a clique of size $t\ge2$, the Schmidt decomposition across $S$ and its complement has the form
\[
 \ket\psi=K^{-1/2}\sum_{a=1}^K\ket a_S W\ket a,\qquad K=2^{k-r},\quad W^\dagger W=I_K.
\]
Let $\mathcal E$ consist of the identity and the three single-qubit Pauli operators at each of the $t$ sites. For distinct $E,F\in\mathcal E$, $E^\dagger F$ is a traceless Pauli operator up to phase, supported on at most two of these sites. Since their joint marginal with $S$ is maximally mixed, $\langle\psi| (\ket a\bra b)_S\otimes E^\dagger F|\psi\rangle=0$ for all $a,b$. Thus $W^\dagger E^\dagger FW=0$. The cross inner products $\langle EW a\mid FW b\rangle$ vanish for all Schmidt basis vectors $a,b$. The $1+3t$ images $E\operatorname{im} W$ are therefore mutually orthogonal $K$-dimensional subspaces, giving the Hamming dimension bound~\cite{Gottesman2002}
\begin{equation}\label{eq:asym-packing}
 (1+3t)K\le 2^{k+r+\delta}.
\end{equation}
For $\delta=0$, dividing by $K$ gives $1+3t\le4^r$, hence $t\le(4^r-1)/3$.

For $\delta=1$, we exclude the largest integer $t=2(4^r-1)/3$ allowed by Eq.~\eqref{eq:asym-packing}. At this value $b=1+3t=2\cdot4^r-1\equiv3\pmod4$ and the ambient dimension is $(b+1)K$. Let $P=WW^\dagger$. The unused subspace has projector
\[
 H=I-\sum_{E\in\mathcal E}EPE^\dagger .
\]
It has rank $K$ and is orthogonal to $P$. For a nonidentity Pauli string $R$ on the complement of $S$, let $j$ be the number of the $t$ selected sites in its support. At each of these $j$ sites exactly two of the three Pauli operators anticommute with $R$. It follows that
\[
 \Tr(HR)=-(b-4j)\Tr(PR).
\]
The identity coefficients of $H$ and $P$ agree, as do their squared Hilbert--Schmidt norms. Since $b-4j\equiv3\pmod4$, each multiplier has modulus at least one. Orthogonality of the Pauli basis therefore forces $|b-4j|=1$ whenever $\Tr(PR)\ne0$, and the congruence then forces $b-4j=-1$. All Pauli coefficients of $H$ and $P$ coincide. This contradicts $HP=0$ and $\operatorname{rank}P=K>0$, proving the odd bound in Eq.~\eqref{eq:asym-cap}.

Applying Lemma~\ref{lem:asym-finite-upper} to $\mathcal H(\rho)$ at all link depths $s\le r$ and then maximizing over $\rho$ gives, for $1\le r\le k$ and $L=k+r+\delta$,
\begin{equation}\label{eq:asym-finite-upper}
 g_{2k+\delta}\le \frac{L^r}{(L)_r}\,U_r^{(\delta)},\qquad
 U_r^{(\delta)}=\prod_{s=2}^r\left(1-\frac1{M_s^{(\delta)}}\right)^{s-1}.
\end{equation}
The factor $L^r/(L)_r$ enters when passing from independent draws to subsets of distinct sites. For fixed $r$ it tends to one as $k\to\infty$.

\subsection{Projective bases and the limiting quantum densities}
\subsubsection{The projective basis hypergraph}
Among hypergraphs satisfying the even clique caps, the weighted bound $\Lambda(\mathcal L_S)\le U_r^{(0)}$ is sharp at each depth $r$. Take as vertices the one-dimensional subspaces of $\mathbb F_4^r$, the points of $\mathrm{PG}(r-1,4)$, and as edges the unordered sets of $r$ points whose representatives are linearly independent. These are the bases of the projective-geometry matroid. There are $M_r^{(0)}=(4^r-1)/3$ points. To see why the bounds also hold in every link, fix $r-s$ independent points and quotient by their span. Two remaining points can belong to a common basis extension exactly when their images in the $s$-dimensional quotient are nonzero and span distinct directions. A clique therefore contains at most one representative of each of the $(4^s-1)/3$ directions. Choosing one per direction attains that size for $s\ge2$. A dependent set has empty link.

Under uniform weights, after $j$ independent points have been drawn, their span contains $(4^j-1)/3$ points. Thus the probability that $r$ draws form a basis is
\[
 \prod_{j=0}^{r-1}\frac{4^r-4^j}{4^r-1}
 =\frac{\prod_{s=1}^r(1-4^{-s})}{(1-4^{-r})^r}
 =U_r^{(0)} .
\]
For the last equality use $1-1/M_s^{(0)}=(1-4^{-(s-1)})/(1-4^{-s})$. For example, the plane $\mathrm{PG}(2,4)$ has $21$ points and each line has five. After the first point, the second must differ from it and the third must avoid their line, giving $(20/21)(16/21)=320/441$. The lemma and the uniform choice thus give the exact normalized Lagrangian of this basis hypergraph. Finite-field basis constructions occur in daisy-density bounds~\cite{EllisIvanLeader2024}, and the projective basis Lagrangian is evaluated in Ref.~\cite{PolWalshWigal}.

\subsubsection{Code states and the lower bound}
The projective example proves sharpness for the class of hypergraphs satisfying the even clique bounds. For the lower bound on $g_n$, we need actual states. For comparison, a specified $k$-body marginal of a random graph state is maximally mixed with probability $\prod_{i=0}^{k-1}(1-2^{i-(n-k)})$~\cite{Zhang2025}. Writing $c_q=\prod_{j=1}^{\infty}(1-q^{-j})$, this tends to $c_2$ for $n=2k$ and $2c_2$ for $n=2k+1$. Write $g_n^{\rm stab}$ for the maximum fraction of maximally mixed $\lfloor n/2\rfloor$-body marginals among pure stabilizer states on $n$ qubits.

\begin{lemma}\label{lem:asym-code}
For $m\ge1$, $p_m\le g_{2m}^{\rm stab}\le g_{2m}$, where
\begin{equation}\label{eq:asym-code}
 p_m=\prod_{j=1}^m\frac{1-(-1/2)^j}{1+2^{1-2j}}
 \xrightarrow[m\to\infty]{}c_4 .
\end{equation}
For $m\ge2$, $g_{2m-1}\ge g_{2m}$ and $g_{2m-1}^{\rm stab}\ge g_{2m}^{\rm stab}$.
\end{lemma}
\begin{proof}
Choose uniformly among the Hermitian self-dual linear codes $C\subset\mathbb F_4^{2m}$ on labeled coordinates, with $h(x,y)=\sum_i x_i y_i^2$. Identify the four symbols at each coordinate with the four single-qubit Pauli labels. Commutation is governed by the field trace of $h$~\cite{CRSS1998}. The $4^m=2^{2m}$ words of $C$ thus specify a maximal commuting family on $2m$ qubits. A common eigenstate is a pure stabilizer state. In its Pauli expansion, the partial trace to $A$ retains precisely the stabilizers supported on $A$, so $\rho_A$ is maximally mixed exactly when $C\cap\mathbb F_4^A=\{0\}$. Here $\mathbb F_4^A$ denotes the vectors supported on $A$.

Fix a balanced partition $A\sqcup B=[2m]$. The condition $C\cap\mathbb F_4^A=\{0\}$ makes the projection onto $B$ bijective. Thus $C=\{(Ux,x):x\in\mathbb F_4^B\}$ for a unique linear map $U$. Self-orthogonality is equivalent to $h(Ux,Uy)=h(x,y)$, so the favorable codes are counted by the order of the finite unitary group $U(m,2)$. An orthonormal basis can be chosen one vector at a time in successive nondegenerate Hermitian spaces. In dimension $j$, a vector has norm one exactly when its Hamming weight is odd, giving $[4^j-(-2)^j]/2$ choices. Hence
\[
 |U(m,2)|=\prod_{j=1}^m 2^{j-1}\bigl(2^j-(-1)^j\bigr).
\]
The number of labeled self-dual codes is $H_m=\prod_{j=1}^m(2^{2j-1}+1)$~\cite{HaradaMunemasa2011}. The probability for the fixed partition is $|U(m,2)|/H_m=p_m$. The expected number of maximally mixed marginals is $\binom{2m}{m}p_m$, so some code has at least this many. To take the limit, cancel the odd-index factors to obtain
\[
 p_m=
 \frac{\displaystyle\prod_{j=1}^{\lfloor m/2\rfloor}(1-4^{-j})}
 {\displaystyle\prod_{j=\lceil m/2\rceil+1}^{m}(1+2^{1-2j})}.
\]
The denominator tends to one.

For the odd-size assertion, fix any qubit $v$ of a pure $2m$-qubit state and condition on a nonzero-probability outcome of a rank-one measurement on $v$. Maximally mixed balanced marginals occur in complementary pairs, each with exactly one member containing $v$. If $\rho_{A\cup\{v\}}=I/2^m$, the postmeasurement state on $A$ is $I/2^{m-1}$. At least half of the original good labels therefore survive, and $\binom{2m-1}{m-1}=\binom{2m}{m}/2$. Applying this to a state attaining $g_{2m}$ proves the first projection inequality. For a state attaining $g_{2m}^{\rm stab}$, choose a single-qubit Pauli measurement. Each nonzero-probability outcome leaves a pure stabilizer state on the remaining qubits, proving the second inequality.
\end{proof}

\subsubsection{The limiting density}
\begin{theorem}\label{thm:asym-density}
At zero tolerance the even sequence has the limit
\begin{equation}\label{eq:asym-even-limit}
 \lim_{k\to\infty}g_{2k}=c_4=\prod_{j=1}^{\infty}(1-4^{-j}),
\end{equation}
and the odd sequence satisfies
\begin{equation}\label{eq:asym-odd-limits}
 c_4\le\liminf_{k\to\infty}g_{2k+1}
 \le\limsup_{k\to\infty}g_{2k+1}
 \le\prod_{s=2}^{\infty}\left(1-\frac{3}{2\cdot4^s-5}\right)^{s-1}.
\end{equation}
\end{theorem}
\begin{proof}
For each fixed $r$, first let $k\to\infty$ in Eq.~\eqref{eq:asym-finite-upper}. The limit superior is at most $U_r^{(\delta)}$. For $\delta=0$, the projective basis formula above gives $U_r^{(0)}\to c_4$ as $r\to\infty$, since $(1-4^{-r})^r\to1$. Lemma~\ref{lem:asym-code} supplies the matching lower bound. For $\delta=1$, substitute $M_s^{(1)}$ from Eq.~\eqref{eq:asym-cap} into $U_r^{(1)}$ and let $r\to\infty$. The odd lower bound follows by single-qubit projection in Lemma~\ref{lem:asym-code}.
\end{proof}

Since $p_k$ and $g_{2k}$ both tend to $c_4>0$, Lemma~\ref{lem:asym-code} gives
\[
 \lim_{k\to\infty}g_{2k}^{\rm stab}=c_4,\qquad
 \lim_{k\to\infty}\frac{g_{2k}^{\rm stab}}{g_{2k}}=1.
\]
At odd sizes the Pauli measurement gives $g_{2k+1}^{\rm stab}\ge g_{2k+2}^{\rm stab}$, hence $\liminf_k g_{2k+1}^{\rm stab}\ge c_4\simeq0.688538$. The upper bound in Eq.~\eqref{eq:asym-odd-limits}, approximately $0.824106$, applies to all pure states.

Taken over both parities, $\liminf_{n\to\infty}g_n=c_4$. Each maximally mixed marginal of the construction remains good at every $\varepsilon\ge0$, so its lower bound applies to the robust count as well. The upper bound here uses exact orthogonality at $\varepsilon=0$.

\subsection{Finite-size comparison of the link and incidence bounds}\label{sec:bound-crossover}

Theorem~\ref{thm:asym-density} takes $k\to\infty$ at each fixed depth $r$. At finite $n$ each candidate depth still carries the factor $L^r/(L)_r$. We retain this factor in both comparisons below.

\subsubsection{Four, six and eight qubits}

For $n=2k$, the depth-$1$ bound is $g_{2k}\le1$. At depth $2$ the clique cap is $M_2^{(0)}=5$, and Eq.~\eqref{eq:asym-finite-upper} gives $g_{2k}\le (4/5)(k+2)/(k+1)$. Substitution gives a right-hand side above one at all available depths $r\ge3$ for the three sizes in Table~\ref{tab:finite-link-comparison}.
\begin{center}
\begin{minipage}{\linewidth}\centering\captionsetup{hypcap=false}
\captionof{table}{The finite link bound at zero tolerance, optimized over the available depths in Eq.~\eqref{eq:asym-finite-upper}. Complementary pairing converts the density bound into an even count.}\label{tab:finite-link-comparison}

\begin{tabular}{ccccc}
\toprule
Qubits & Depth-$2$ bound & Best from Eq.~\eqref{eq:asym-finite-upper} & Count bound using pairs & Exact count\\
\midrule
$4$ & $16/15$ & $1$ & $6$ & $4$\\
$6$ & $1$ & $1$ & $20$ & $20$\\
$8$ & $24/25$ & $24/25$ & $66$ & $56$\\
\bottomrule
\end{tabular}\end{minipage}
\end{center}
At four qubits, the purity constraint in Section~\ref{sec:four-example} lowers the count from six to four. At six qubits an absolutely maximally entangled state attains all $\binom63=20$ marginals~\cite{Goyeneche2015}. At eight qubits the depth-two sampling factor $6/5$ raises the link density bound from $4/5$ to $24/25$, allowing up to 66 good labels after complementary pairing. The exact value 56 in Section~\ref{ch:4m} uses the additional lower bound of one on the sum of the four-body defects in each five-set. These sizes show how finite counts can depend on purity information that the clique caps do not record.

\subsubsection{Incidence bounds by system size}

We compare the finite link bound in Eq.~\eqref{eq:asym-finite-upper} with the incidence bounds of Theorems~\ref{thm:4m}, \ref{thm:odd-count}, and~\ref{thm:even-count} for $n\equiv0\pmod4$, odd $n$, and $n\equiv2\pmod4$, respectively. Each crossover below compares the two bounds within one of these size families. Write $N=\binom nk$ and let $C_n$ denote the local upper bound on the good fraction before rounding. Maximizing the guaranteed number of bad marginals over the parameters of each theorem gives
\[
 C_n=
 \begin{cases}
 1-\dfrac{1}{k+1}, & n=2k,\ k\text{ even},\\[5pt]
 1-\dfrac{2(k-5)}{(k+2)(2k-5)}, & n=2k+1,\ k\ge6,\\[5pt]
 1-\dfrac{2h_k}{(k+1)(k+2)}, & n=2k,\ k\ge5\text{ odd},
 \end{cases}
 \qquad
 h_k=\min\!\left\{2k+1,\left\lceil L_k/3\right\rceil\right\}.
\]
Here $L_k$ is the deletion-moment bound of Lemma~\ref{lem:even-moments}. The last line uses its largest admissible counting parameter. The uniform choice in Section~\ref{sec:even-strict} gives $h_k\ge(3k-11)/2$. These formulas therefore give $C_n=1-\Theta(1/n)$ in each family.

Write $B_r$ for the right-hand side of Eq.~\eqref{eq:asym-finite-upper} and $B_n^*=\min_{1\le r\le k}B_r$. For either density bound $x$, the integer count is $\lfloor Nx\rfloor$ at odd sizes and $2\lfloor Nx/2\rfloor$ at even sizes. To find the minimizing depth, consecutive terms satisfy
\[
 \frac{B_{r+1}}{B_r}=
 \left[\left(1+\frac{1}{k+r+\delta}\right)
             \left(1-\frac1{M_{r+1}^{(\delta)}}\right)\right]^r .
\]
Thus increasing the depth helps exactly when $M_{r+1}^{(\delta)}<k+r+\delta+1$. The caps in Eq.~\eqref{eq:asym-cap} increase by more than one per step. Once the inequality fails, further depths cannot improve the bound. Table~\ref{tab:bound-crossover} gives the resulting finite comparisons. In the $4m+2$ family, $h_k=15,21,28,35,42$ for $n=26,30,34,38,42$, respectively. The local formula is stronger at $30$ and $34$ as well as at $38$. If one fixes the simpler parameter $(3k-11)/2$ from Section~\ref{sec:even-strict}, its density bound exceeds $B_2$ by $(k^2-16k+49)/[5(k+1)(k+2)]$, which is positive for odd $k\ge13$ ($n\ge26$). The comparison in the table uses the full choice of parameters.

\begin{table}[H]
\centering
\caption{Selected zero-tolerance density upper bounds before integer rounding. The smaller entry gives the stronger bound. The strict comparisons shown also hold for counts after rounding.}
\label{tab:bound-crossover}
\begin{tabular}{lrrrr}
\toprule
Size family & $n$ & Minimizing link depth & Incidence $C_n$ & Link $B_n^*$\\
\midrule
$n\equiv0\pmod4$ & 16 & 2 & $8/9$ & $8/9$\\
                 & 20 & 2 & $10/11$ & $48/55$\\
\addlinespace
$n$ odd & 25 & 2 & $18/19$ & $20/21$\\
        & 27 & 2 & $299/315$ & $128/135$\\
\addlinespace
$n\equiv2\pmod4$ & 26 & 2 & $6/7$ & $6/7$\\
                 & 38 & 3 & $5/6$ & $7744/9261$\\
                 & 42 & 3 & $211/253$ & $10240/12397$\\
\bottomrule
\end{tabular}
\end{table}

The improvement persists for all larger sizes in each family. For even $k\ge10$ and $n=2k$, subtraction of the depth-two bound gives
\[
 C_n-B_2=\frac{k-8}{5(k+1)}.
\]
Since $N\ge\binom{2k}{2}$, the difference in counts before rounding exceeds two.
For $n=2k+1$, the depth-two bound is $B_2=8(k+3)/[9(k+2)]$. For $k\ge13$,
\[
 C_n-B_2=\frac{2k^2-35k+120}{9(k+2)(2k-5)}\ge\frac1{945}.
\]
After clearing the positive denominators, the last inequality reduces to $(k-13)(104k-485)\ge0$. Here $N\ge\binom{27}{3}>945$, so the difference in counts exceeds one.

Finally, for odd $k\ge21$ and $n=2k$, the cap $h_k\le2k+1$ gives $C_n\ge k(k-1)/[(k+1)(k+2)]$. At depth three $U_3^{(0)}=320/441$, so
\[
 C_n-B_3\ge\frac{121k^2-2361k-2880}{441(k+1)(k+2)} .
\]
The numerator is $900$ at $k=21$ and increases thereafter. Since $N\ge\binom{2k}{2}\ge(k+1)(k+2)$, the difference in counts is at least $900/441>2$. The preceding sizes in each family, $n=16,25,38$, show why the respective starting points $20,27,42$ cannot be lowered.

\subsection{Balanced marginals at fixed tolerance for even numbers of qubits}\label{sec:asym-robust}

With positive tolerance the Pauli coefficients of a good marginal can be nonzero, so the exact orthogonality in Eq.~\eqref{eq:asym-packing} need no longer hold. For even sizes a distribution on Pauli supports, weighted by their contribution to the defect, still forces a positive probability of a bad marginal after deleting sites.

\begin{theorem}\label{thm:positive-even}
Fix $0<\varepsilon<1$. A good $k$-body marginal then has purity at most $(1+\varepsilon)2^{-k}$, strictly below twice the maximally mixed value. At this fixed tolerance, a positive fraction of the balanced marginals of every sufficiently large pure $2k$-qubit state must be bad. Equivalently,
\[
 \limsup_{k\to\infty}\frac{\Qrob{\varepsilon}(2k,k)}{\binom{2k}{k}}<1 .
\]
\end{theorem}

\begin{proof}
Fix an integer $h\ge2$ and take $k\ge h$. In any pure $2k$-qubit state, choose a $(k+h)$-body marginal $\rho_T$. Write $t=k+h$, $C=T^c$, $m=4^h-1$ and $\alpha=(4^{h-1}-1)/m$. Complementarity gives $D_T=4^h(1+D_C)-1\ge m$. Sample a nonempty Pauli support $W\subseteq T$ with probability $w_W/D_T$ using Eq.~\eqref{eq:pre-support-mass}. For each $i\in T$, complementarity and $D_{C\cup\{i\}}\ge D_C$ give
\[
 \Pr(i\notin W)=\frac{D_{T\setminus\{i\}}}{D_T}
 =\frac{4^{h-1}(1+D_{C\cup\{i\}})-1}{4^h(1+D_C)-1}\ge\alpha .
\]
The fraction $Y=|T\setminus W|/t$ of avoided sites therefore satisfies $\mathbb E Y\ge\alpha$. Given $W$, $h$ independent uniform draws from $T$ all avoid $W$ with probability $Y^h$. Jensen's inequality gives $\mathbb E Y^h\ge\alpha^h$. Conditioning the draws to be distinct changes this probability by at most $\binom h2/t$, the union bound for a repeated site. For the resulting uniform $h$-subset $J$, the probability over $W$ of avoiding $J$ is $D_{T\setminus J}/D_T$. This ratio is at most $\varepsilon/m$ if $\rho_{T\setminus J}$ is good, and always at most one. Consequently the fraction of deletions yielding a bad marginal satisfies
\[
 \Pr_J(D_{T\setminus J}>\varepsilon)\ge\alpha^h-\frac{\varepsilon}{m}-\frac{\binom h2}{t}.
\]
Averaging over uniform $T$ gives the same lower bound for the bad fraction among all $k$-body marginals, since $T\setminus J$ is uniform. As $m\alpha^h=(1-4^{-h})(1-3/(4^h-1))^h\to1$, choose $h$ with $m\alpha^h>\varepsilon$. With this $h$ fixed, the lower bound has a strictly positive limit as $k\to\infty$, uniformly over pure states.
\end{proof}

\section{Simultaneous extrema and stabilizer states}\label{sec:structure}\label{app:peme}\label{ch:peme}
A state can attain the largest possible number of maximally mixed half-system marginals and still miss the optimum on smaller subsystems. Write $m_j^{(\varepsilon)}(\psi)=\#\{A\in\binom{[n]}j:D_A\le\varepsilon\}$ and $m_j=m_j^{(0)}$. To ask whether separate optima can coexist, associate to a pure state its entire count vector $(m_1,\ldots,m_{\lfloor n/2\rfloor})$. A perfect extremal maximally entangled (PEME) state realizes $m_j=\Qex(n,j)$ in every coordinate~\cite{Zhang2025}. We first determine which finite-size constructions realize this vector. We then consider how allowing positive defect changes the optimum within the stabilizer family.

\subsection{When the same state attains every marginal count}\label{app:peme-graphs}
For a graph state, Eq.~\eqref{eq:witness-rank} tests every coordinate of the count vector using binary cut ranks. If a half-system optimizer is $(\lfloor n/2\rfloor-1)$-uniform, all its smaller coordinates automatically attain $\binom nj$. This criterion applies to the finite graph states in Table~\ref{tab:peme-small}. In particular, the nine-qubit graph of Eq.~\eqref{eq:peme-nine-drawn} has full row rank on every one-, two-, and three-site cut, while exactly 112 of its four-site cuts have full row rank.

\begin{table}[htbp]
\centering
\caption{PEME graph states for $2\leq n\leq9$ and $n=11$. A wheel consists of a cycle and a hub adjacent to every cycle vertex. The graph $T_4$ consists of two copies of $K_4$ joined by a perfect matching.}
\label{tab:peme-small}
\begin{tabular}{@{}cll@{}}
\toprule
$n$ & Graph & $(m_1,\ldots,m_{\lfloor n/2\rfloor})$\\
\midrule
2 & Single edge & $(2)$\\
3 & Three-vertex path & $(3)$\\
4 & Two disjoint edges & $(4,4)$\\
5 & Five-cycle & $(5,10)$\\
6 & Wheel on six vertices & $(6,15,20)$\\
7 & Seven-vertex wheel~\cite{Huber2017} & $(7,21,32)$\\
8 & $T_4\cong K_4\mathbin{\square}K_2$~\cite{Zhang2025} & $(8,28,56,56)$\\
9 & 9-qubit PEME graph (Fig.~\ref{fig:graph-states}(d)) & $(9,36,84,112)$\\
11 & Vertex-deleted circulant (Section~\ref{sec:eleven}) & $(11,55,165,330,396)$\\
\bottomrule
\end{tabular}
\end{table}

At $n=2,3,5,6$ the half-system counts in the table also attain the full binomial count. The four- and eight-qubit examples attain the bounds in Sections~\ref{sec:four-example} and~\ref{ch:4m}. At seven qubits the wheel attains the bound obtained using the six-body shadow and the global mean purity in Section~\ref{sec:seven}~\cite{Higuchi2000,Huber2017,SudevanDas2022,Zhang2025}. The nine- and eleven-qubit cases use the exact upper bounds of Sections~\ref{sec:nine} and~\ref{sec:eleven}. For the graph-state orbit used at nine qubits, the catalog identification and the explicit labeling are given with Eq.~\eqref{eq:peme-nine-drawn}~\cite{Cabello2011,CabelloGraphData}. Together these constructions establish simultaneous attainment from two to nine qubits and at eleven qubits.

\subsection{Incompatible four- and five-body counts at ten qubits}\label{app:peme-tradeoff}
Section~\ref{sec:ten} gives $\Qex(10,5)=200$. To test whether a state attaining 200 can also attain $\Qex(10,4)$, it suffices to compare its four-body count with that of one explicit state. The minimum mean five-body purity over pure ten-qubit states is $1/21$~\cite[Table~5]{Lewis2026}. For every such state $\ket\psi$, Eq.~\eqref{eq:meanpurity} therefore gives $S_5=252(32\bar\pi_5-1)\ge252(32/21-1)=132$. Suppose at least 204 of its four-body marginals and 200 of its five-body marginals are good. To bound the same sum from above, group the five-body marginals $\rho_A$ and $\rho_{A^c}$ by their complementary index sets $[A]=\{A,A^c\}$ and write $D_{[A]}=D_A=D_{A^c}$.

\begin{lemma}\label{lem:peme-large-pairs}
If at most six four-body marginals of $\ket\psi$ are bad, every five-body marginal $\rho_A$ satisfies $D_A\le3+4\varepsilon$, and at most 18 complementary five-set pairs satisfy $D_{[A]}>1+10\varepsilon$.
\end{lemma}
\begin{proof}
The two halves of $[A]$ contain ten distinct four-faces, so at least one face $C$ indexes a good four-body marginal $\rho_C$. Choose the representative $A$ such that $C\subset A^c$. Applying monotonicity to $A\subset C^c$ (Eq.~\eqref{eq:monotonicity}) and complementarity to the $4|6$ cut $C|C^c$ (Eq.~\eqref{eq:complement}) gives
$D_A\le D_{C^c}=3+4D_C\le3+4\varepsilon$.

Let $\omega$ be either five-body marginal, $\rho_A$ or $\rho_{A^c}$. If its five four-body marginals are all good, applying the local shadow to $\omega$ (Eq.~\eqref{eq:five-shadow}) gives $D_{[A]}\le1+2\cdot5\varepsilon$. For every pair above this bound, each half therefore contains the index set of a bad four-body marginal. Join two such index sets by an edge when they are disjoint and add isolated vertices to obtain a graph on six vertices. The graph is triangle-free since three disjoint four-sets would need twelve sites. The triangle case of Tur\'an's theorem~\cite{Turan1941} therefore bounds its number of edges by $e\le\lfloor6^2/4\rfloor=9$. Each edge can lie across at most two complementary five-set pairs, obtained by allocating the two unused sites one to each half. Thus at most $2e\le18$ pairs exceed the bound.
\end{proof}

\begin{theorem}\label{thm:peme-tradeoff}
For any pure state of ten qubits and $0\le\varepsilon<1/63$,
\[
m_5^{(\varepsilon)}(\psi)\ge200
\quad\Longrightarrow\quad
m_4^{(\varepsilon)}(\psi)\le203 .
\]
\end{theorem}
\begin{proof}
Suppose $m_5^{(\varepsilon)}\ge200$ and $m_4^{(\varepsilon)}\ge204$. At most six of the $\binom{10}{4}=210$ four-body marginals and at most 52 of the $\binom{10}{5}=252$ five-body marginals are bad. The latter form $t\le26$ complementary pairs. Let $h$ count the pairs above $1+10\varepsilon$. Lemma~\ref{lem:peme-large-pairs} gives $h\le18$ and bounds every five-body defect by $3+4\varepsilon$. Summing the defects of the good five-body marginals and then those of the two classes of bad pairs gives
\begin{align}
S_5
&\le (252-2t)\varepsilon+2(t-h)(1+10\varepsilon)+2h(3+4\varepsilon)\notag\\
&=252\varepsilon+(2+18\varepsilon)t+(4-12\varepsilon)h
\le124+504\varepsilon<132 .
\end{align}
The last comparison uses $\varepsilon<1/63$, so the coefficients of $t$ and $h$ are positive. This contradicts the lower bound $S_5\ge132$.
\end{proof}

\begin{corollary}\label{cor:peme-ten}
Ten is the smallest number of qubits for which no PEME state exists.
\end{corollary}
\begin{proof}
In the row encoding of Eq.~\eqref{eq:peme-nine-drawn}, the ten-qubit graph with adjacency rows
\begin{equation}\label{eq:peme-ten-rows}
(226,921,368,610,646,973,45,51,38,58)
\end{equation}
has binary cut ranks $4,3,2$ with multiplicities $204,4,2$ over the four-vertex subsets. Its graph state therefore has 204 maximally mixed four-body marginals, giving $\Qex(10,4)\ge204$. A state attaining $\Qex(10,5)=200$ has $m_4\le203$ by Theorem~\ref{thm:peme-tradeoff} at $\varepsilon=0$. Such a state cannot attain the four-body optimum. Table~\ref{tab:peme-small} supplies a PEME state for each smaller size $2\le n\le9$.
\end{proof}
There is a second incompatibility at twelve qubits. The five-uniform dodecacode state attains $\Qex(12,5)=792$~\cite{Scott2004}. Any state attaining 792 is five-uniform and hence has at most 690 maximally mixed six-body marginals by Proposition~\ref{prop:12-fiveuniform-comparison}. Since $\Qex(12,6)=720$, no state attains both counts. At eleven qubits the four-uniform optimizer in Table~\ref{tab:peme-small} attains all the smaller counts as well.

\subsection{The first increase in the stabilizer optimum}\label{sec:witnesses}\label{sec:stab-threshold}
Table~\ref{tab:main} also raises a question at a single fixed subsystem size. Each listed optimum at zero defect is attained by a stabilizer state. How far can the tolerance increase before even this restricted family admits a higher count? Let $s_A$ be the stabilizer dimension in Eq.~\eqref{eq:pre-stab-defect}. Fix $1\le k\le\lfloor n/2\rfloor$ and let $M_r=\max_S\#\{A:|A|=k,\ s_A\le r\}$. For integers $r\ge0$, the maximum over pure stabilizer states equals $M_r$ on $2^r-1\le\varepsilon<2^{r+1}-1$. To see that the count increases at each step before saturation, we use the following perturbation.

\begin{lemma}\label{lem:stab-exchange}
Suppose $S_A$ and $S_{A^c}$ each contain a nonidentity stabilizer. There is a two-qubit Clifford rotation producing a stabilizer group $S'$ with $s'_A=s_A-1$ and $|s'_T-s_T|\le1$ for every $T$.
\end{lemma}
\begin{proof}
Take nonidentity $g\in S_A$ and $h\in S_{A^c}$. At one site of each support choose a Pauli anticommuting with the corresponding local factor. Their product $P$ anticommutes with both $g$ and $h$. The rotation $U=\exp(i\pi P/4)$ fixes $H=\{u\in S:[u,P]=0\}$, a codimension-one subgroup of both $S$ and $S'=USU^\dagger$. For any $T$, both $s_T$ and $s'_T$ are therefore either $\dim H_T$ or $\dim H_T+1$. On $A$, the presence of $g$ gives $\dim H_A=s_A-1$. Every element of $S'\setminus H$ has Pauli label $Pu$ for some $u\in S$. It anticommutes with $h$, so its support intersects $A^c$. Thus $S'_A=H_A$.
\end{proof}

\begin{theorem}\label{th:stab-step}
If $M_r<\binom nk$, then $M_{r+1}>M_r$. In particular, the first increase occurs at $\varepsilon=1$ whenever there is no $k$-uniform stabilizer state. A $k$-uniform state has every $k$-body marginal maximally mixed and attains $\binom nk$ already at zero tolerance.
\end{theorem}
\begin{proof}
For any chosen set $A$ of $k$ qubit labels, placing $k$ Bell pairs across $A|A^c$ and leaving the other qubits in a product state makes $\rho_A$ maximally mixed. Thus $M_r\ge1$. Choose a stabilizer attaining $M_r$. The graph of $k$-sets connected by single-site exchanges is connected, so some set with dimension at most $r$ has a neighbor $A$ with $s_A>r$. On adding one site the local subgroup dimension increases by at most two, since restriction to its two binary Pauli coordinates has kernel the old subgroup. It follows that $s_A$ is $r+1$ or $r+2$. In the first case the same state gains $A$ at the next step. In the second, complementarity gives $s_{A^c}=s_A+n-2k\ge s_A\ge2$. Lemma~\ref{lem:stab-exchange} then lowers $s_A$ to $r+1$, while every previously counted set has $s'_T\le r+1$. In either case at least $M_r+1$ sets are counted at the next step.
\end{proof}

Theorem~\ref{th:stab-step} locates the first increase in the stabilizer optimum at $\varepsilon=1$ for each size in Table~\ref{tab:stabw}. For each graph state in the table, Eq.~\eqref{eq:witness-rank} gives the distribution of defects of its $k$-body marginals. The circulant $C_n(J)$ joins vertices of $\mathbb Z_n$ at differences $\pm j$, $j\in J$. Here $-v$ denotes deletion of one vertex. At tolerance one, all $k$-body marginals of the listed graph states are good for $n=4,7,8,9,10,11$. For the twelve-qubit graph state, $918$ of its $924$ six-body marginals are good.
\begin{table}[htbp]\centering\caption{Defect multiplicities of explicit stabilizer states at tolerance one. Here $q$ is the optimum at zero error, and $N_j$ counts the $k$-body marginals of the indicated graph state with defect $j$. All other multiplicities vanish.}\label{tab:stabw}
\begin{tabular}{ccclrrr}\toprule
$n$ & $k$ & $q$ & Graph & $N_0$ & $N_1$ & $N_3$\\\midrule
4 & 2 & 4 & $P_4$ & 4 & 2 & 0\\
7 & 3 & 32 & $C_6\vee K_1$ & 32 & 3 & 0\\
8 & 4 & 56 & $C_8(\{1,4\})$ & 28 & 42 & 0\\
9 & 4 & 112 & $C_9(\{1,2\})$ & 90 & 36 & 0\\
10 & 5 & 200 & $C_{10}(\{1,5\})$ & 72 & 180 & 0\\
11 & 5 & 396 & $C_{12}(\{1,3,6\})-v$ & 396 & 66 & 0\\
12 & 6 & 720 & $C_{12}(\{1,3,6\})$ & 540 & 378 & 6\\
\bottomrule\end{tabular}\end{table}

\section{Tensor isometries, secret sharing, and open questions}\label{sec:applications}\label{app:applications}
The counts above admit two operational interpretations. Choosing a subsystem as the input legs of an amplitude tensor turns maximal mixing into an isometry condition. Choosing one leg as a reference qubit turns the same balanced cut into a statement about the information available to a group of shares. In both interpretations the normalization of the defect determines precisely what a nonzero tolerance controls.

\subsection{Isometric tensor flattenings}\label{sec:app-flattenings}

Fix a choice of $k$ input legs of an $n$-qubit amplitude tensor. The Schmidt spectrum across the resulting cut determines how closely the linear map from these inputs to the remaining legs preserves inner products.

\begin{lemma}\label{lem:flattening-gram}
Let $\rho_A$ be a $k$-body marginal of a normalized $n$-qubit pure state $\ket\psi$. Arrange the amplitudes of $\ket\psi$ in a matrix $T_A$ whose columns and rows are indexed by basis strings on $A$ and $A^c$, respectively. For $d_A=2^{|A|}$ and $F_A=\sqrt{d_A}\,T_A$,
\begin{equation}
 d_A^{-1}\|F_A^\dagger F_A-I_A\|_F^2=D_A .
\end{equation}
\end{lemma}
\begin{proof}
Contraction over the row index gives $T_A^\dagger T_A=\rho_A^{\mathsf T}$. Expanding the squared Frobenius norm and using $\Tr\rho_A=1$ gives $d_A\Tr\rho_A^2-1$.
\end{proof}

By Lemma~\ref{lem:flattening-gram}, a $k$-body marginal $\rho_A$ of $\ket\psi$ is good exactly when the normalized squared Gram error for the input legs $A$ is at most $\varepsilon$. In particular, $F_A$ is an isometry exactly when $\rho_A$ is maximally mixed. Perfect tensors have this property on every cut with $|A|\le |A^c|$~\cite{Goyeneche2015,Pastawski2015}. On $n=2k$ qubits each unordered unitary cut accounts for two complementary input choices, so its sharp maximum count is $\Qex(2k,k)/2$. The bounds and attaining states in Sections~\ref{ch:4m}, \ref{sec:ten}, and~\ref{sec:twelve} give $28,100,360$ such cuts at $n=8,10,12$. For $n=2k+1$, each $k$-leg input has a $(k+1)$-leg output, so every counted input gives a distinct cut. The sharp maxima are $112$ four-leg inputs at nine qubits and $396$ five-leg inputs at eleven qubits (Sections~\ref{sec:nine} and~\ref{sec:eleven}).

The fraction of unordered balanced cuts admitting a unitary flattening therefore has optimal value $g_{2k}$, with limit $c_4$ by Theorem~\ref{thm:asym-density}. One can impose a different constraint on the allowed inputs by arranging the legs on a circle and taking contiguous intervals. For the witnesses in Fig.~\ref{fig:graph-states}, take the following cyclic orders of the vertex labels.
\[
\begin{array}{c|l}
n&\textrm{cyclic order}\\\hline
4&(1,3,2,4)\\
7&(1,2,3,4,5,6,7)\\
8&(1,2,3,4,6,5,8,7)\\
9&(1,2,3,4,6,7,5,8,9).
\end{array}
\]
Applying Eq.~\eqref{eq:witness-rank} to these graphs gives full cut rank for every contiguous
$\lfloor n/2\rfloor$-site interval, including those crossing the start
of the order. The marginal on each such interval is therefore maximally mixed. Tracing out sites gives the same conclusion for every shorter interval.
For even $n$, these orders realize planar maximal entanglement in the sense of Ref.~\cite{Doroudiani2020}. For odd $n$, the largest contiguous input has $(n-1)/2$ sites.

\subsection{Depolarization and recovery in pure encodings}\label{sec:app-coalitions}
Consider an isometry $V$ from one input qubit to $2k-1$ shares, with normalized Choi state $\ket\Psi=(I_R\otimes V)\ket{\Phi^+}$. For a coalition $E$ of $k-1$ shares, the $k$-body marginal $\rho_{RE}$ of $\ket\Psi$ on the reference qubit $R$ and those shares is the normalized Choi state of the reduced channel to $E$. Choi inversion therefore identifies $\rho_{RE}=I_{RE}/2^k$ with the channel $X\mapsto\Tr(X) I_E/2^{k-1}$. General exact secrecy requires $\rho_{RE}=\rho_R\otimes\rho_E$. This decoupling implies exact recovery from the remaining shares for a pure encoding~\cite{Cleve1999,Helwig2012}.

\begin{corollary}\label{cor:depolarized-count}
The maximum number of fully depolarized $(k-1)$-share coalitions in a pure one-qubit encoding into $2k-1$ shares is $\Qex(2k,k)/2$.
\end{corollary}
\begin{proof}
For the pure Choi state $\ket\Psi$, Eq.~\eqref{eq:complement} with $n=2k$ pairs every maximally mixed $k$-body marginal $\rho_A$ with a maximally mixed $\rho_{A^c}$. Exactly one of their two index sets contains $R$. Removing $R$ from that set gives one fully depolarized coalition. Hence the count for any encoding is $m_k(\Psi)/2$. Conversely, take a state attaining $\Qex(2k,k)$. This state has at least one pair of maximally mixed complementary $k$-body marginals, since a product of $k$ Bell pairs across a balanced cut makes both marginals maximally mixed. Every choice of reference qubit $R$ lies on one side of that cut, so its one-body marginal is $\rho_R=I_R/2$. By the Schmidt decomposition the attaining state is a normalized Choi state for this choice of $R$, proving sharpness.
\end{proof}
Since $\binom{2k-1}{k-1}=\binom{2k}{k}/2$, the maximum fraction of fully depolarized coalitions of this size equals $g_{2k}$ and tends to $c_4$ by Theorem~\ref{thm:asym-density}. For $7,9,11$ shares the sharp counts are respectively 28 of 35 triples, 100 of 126 four-share coalitions, and 360 of 462 five-share coalitions. For seven shares the $T_4$ and Steane-code Choi states both attain 28, with different access structures~\cite{LiuShao2025}.

At positive tolerance the defect directly bounds the correlation of the shares with the reference. The identity below separates this correlation from any lack of uniformity in the state of the shares themselves.

\begin{lemma}\label{lem:app-decoupling}
Write $d_X=\dim\mathcal H_X$. For any encoding Choi state with $\rho_R=I_R/d_R$, the reduced state on $RE$ satisfies
\begin{equation}
 D_{RE}-D_E=d_Rd_E\|\rho_{RE}-\rho_R\otimes\rho_E\|_2^2,\qquad
 \tfrac12\|\rho_{RE}-\rho_R\otimes\rho_E\|_1\le\tfrac12\sqrt{D_{RE}-D_E}.
\end{equation}
In particular $D_{RE}\le\varepsilon$ gives trace distance at most $\sqrt{\varepsilon}/2$ from the product state.
\end{lemma}
\begin{proof}
For $X=\rho_{RE}-\rho_R\otimes\rho_E$, we have $\Tr_R X=0$. Since $\rho_R=I_R/d_R$, this gives $\Tr[X(\rho_R\otimes\rho_E)]=0$. Thus $\|X\|_2^2=\Tr\rho_{RE}^2-d_R^{-1}\Tr\rho_E^2$. The trace-norm bound follows from $\|X\|_1\le\sqrt{d_Rd_E}\|X\|_2$ and the final estimate from $D_E\ge0$.
\end{proof}

\subsection{Open questions}\label{sec:open-questions}
For a fixed half-system size $k=\lfloor n/2\rfloor$, the sharp robustness radius $\varepsilon_*(n,k)$ is the first tolerance at which an additional marginal can be counted. At four qubits the radius is exactly $1/3$. For $8\le n\le12$, Table~\ref{tab:main} supplies guaranteed stability intervals. Theorem~\ref{th:stab-step} also gives $\varepsilon_*(n,k)\le1$ in these cases, because a stabilizer attains the unrestricted zero-tolerance count and the stabilizer count increases at one. A state with a smaller $(\Qex(n,k)+1)$st defect would improve the upper bound on this radius. A matching lower bound on that order statistic for all pure states would determine it exactly.

We conjecture that a pure stabilizer state attains $\Qex(n,\lfloor n/2\rfloor)$ for every finite $n$. Table~\ref{tab:main} gives exact examples. At even sizes the stabilizer and unrestricted densities have the same limit $c_4$, which still allows a finite gap in the counts. At odd sizes it remains open whether $g_{2k+1}$ converges and whether the ratio of the stabilizer optimum to the unrestricted optimum tends to one.

\pdfbookmark[1]{Acknowledgments}{paper-acknowledgments}
\section*{Acknowledgments}
The research of Wanchen Zhang was supported by the Innovation Program for Quantum Science and Technology under Grant No. 2025ZD0300102. The research of Xiande Zhang was supported by the National Key Research and Development Program of China under Grant No. 2023YFA1010200, the National Natural Science Foundation of China under Grants No. 12171452 and No. 12231014, and the Quantum Science and Technology--National Science and Technology Major Project under Grant No. 2021ZD0302902. OpenAI ChatGPT/Codex (GPT-6) assisted in reorganizing the manuscript, preparing text and vector schematics, refining the seven-qubit estimates, integrating the eleven- and twelve-qubit derivations, developing the stabilizer threshold argument, asymptotic construction, and counting bounds, and implementing the numerical checks.
\pdfbookmark[1]{References}{paper-references}
\bibliographystyle{unsrtnat}
\bibliography{references}

@article{Scott2004,author={Scott, A. J.},title={Multipartite entanglement, quantum-error-correcting codes, and entangling power of quantum evolutions},journal={Phys. Rev. A},volume={69},pages={052330},year={2004},doi={10.1103/PhysRevA.69.052330}}

@article{Cleve1999,author={Cleve, Richard and Gottesman, Daniel and Lo, Hoi-Kwong},title={How to Share a Quantum Secret},journal={Phys. Rev. Lett.},volume={83},pages={648--651},year={1999},doi={10.1103/PhysRevLett.83.648}}

@article{Helwig2012,author={Helwig, Wolfram and Cui, Wei and Latorre, Jos{\'e} Ignacio and Riera, Arnau and Lo, Hoi-Kwong},title={Absolute maximal entanglement and quantum secret sharing},journal={Phys. Rev. A},volume={86},pages={052335},year={2012},doi={10.1103/PhysRevA.86.052335}}

@article{Goyeneche2015,author={Goyeneche, D. and Alsina, D. and Latorre, J. I. and Riera, A. and {\.{Z}}yczkowski, K.},title={Absolutely maximally entangled states, combinatorial designs, and multiunitary matrices},journal={Phys. Rev. A},volume={92},pages={032316},year={2015},doi={10.1103/PhysRevA.92.032316}}

@article{Pastawski2015,author={Pastawski, Fernando and Yoshida, Beni and Harlow, Daniel and Preskill, John},title={Holographic quantum error-correcting codes: Toy models for the bulk/boundary correspondence},journal={J. High Energy Phys.},volume={2015},number={06},pages={149},year={2015},doi={10.1007/JHEP06(2015)149}}

@article{Zhang2025,author={Zhang, Wanchen and Ning, Yu and Shi, Fei and Zhang, Xiande},title={Extremal maximal entanglement},journal={Phys. Rev. A},volume={111},pages={052410},year={2025},doi={10.1103/PhysRevA.111.052410}}

@article{Higuchi2000,author={Higuchi, A. and Sudbery, A.},title={How entangled can two couples get?},journal={Phys. Lett. A},volume={273},pages={213--217},year={2000},doi={10.1016/S0375-9601(00)00480-1}}

@article{Huber2017,author={Huber, Felix and G{\"u}hne, Otfried and Siewert, Jens},title={Absolutely Maximally Entangled States of Seven Qubits Do Not Exist},journal={Phys. Rev. Lett.},volume={118},pages={200502},year={2017},doi={10.1103/PhysRevLett.118.200502}}

@article{Facchi2008,author={Facchi, P. and Florio, G. and Parisi, G. and Pascazio, S.},title={Maximally multipartite entangled states},journal={Phys. Rev. A},volume={77},pages={060304(R)},year={2008},doi={10.1103/PhysRevA.77.060304}}

@article{Rains1999,author={Rains, Eric M.},title={Quantum shadow enumerators},journal={IEEE Trans. Inf. Theory},volume={45},pages={2361--2366},year={1999},doi={10.1109/18.796376}}

@article{ShorLaflamme1997,author={Shor, Peter W. and Laflamme, Raymond},title={Quantum Analog of the {MacWilliams} Identities for Classical Coding Theory},journal={Phys. Rev. Lett.},volume={78},pages={1600--1602},year={1997},doi={10.1103/PhysRevLett.78.1600}}

@article{Huber2018,author={Huber, Felix and Eltschka, Christopher and Siewert, Jens and G{\"u}hne, Otfried},title={Bounds on absolutely maximally entangled states from shadow inequalities, and the quantum {MacWilliams} identity},journal={J. Phys. A: Math. Theor.},volume={51},pages={175301},year={2018},doi={10.1088/1751-8121/aaade5}}

@article{EltschkaHuber2018,author={Eltschka, Christopher and Huber, Felix and G{\"u}hne, Otfried and Siewert, Jens},title={Exponentially many entanglement and correlation constraints for multipartite quantum states},journal={Phys. Rev. A},volume={98},pages={052317},year={2018},doi={10.1103/PhysRevA.98.052317}}

@misc{Lewis2026,author={Lewis, Antony},title={The most entangled states of four and seven qubits},year={2026},eprint={2611.05757},archivePrefix={arXiv},primaryClass={quant-ph},url={https://arxiv.org/abs/2611.05757}}

@article{Zha2020,author={Zha, Xinwei and Ahmed, Irfan and Zhang, Da and Zhang, Yanpeng},title={A criterion to identify maximally entangled nine-qubit state},journal={Laser Phys. Lett.},volume={17},pages={035201},year={2020},doi={10.1088/1612-202X/ab6a16}}

@article{Lovasz1978,author={Lov{\'a}sz, L{\'a}szl{\'o}},title={Kneser's conjecture, chromatic number, and homotopy},journal={J. Combin. Theory Ser. A},volume={25},pages={319--324},year={1978},doi={10.1016/0097-3165(78)90022-5}}

@article{Eltschka2018,author={Eltschka, Christopher and Siewert, Jens},title={Distribution of entanglement and correlations in all finite dimensions},journal={Quantum},volume={2},pages={64},year={2018},doi={10.22331/q-2018-05-22-64}}

@article{Turan1941,author={Tur{\'a}n, P.},title={Eine Extremalaufgabe aus der Graphentheorie},journal={Mat. Fiz. Lapok},volume={48},pages={436--452},year={1941}}

@article{DanielsenParker2006,author={Danielsen, Lars Eirik and Parker, Matthew G.},title={On the classification of all self-dual additive codes over {GF(4)} of length up to 12},journal={J. Combin. Theory Ser. A},volume={113},pages={1351--1367},year={2006},doi={10.1016/j.jcta.2005.12.004}}

@article{Hein2004,author={Hein, M. and Eisert, J. and Briegel, H. J.},title={Multiparty entanglement in graph states},journal={Phys. Rev. A},volume={69},pages={062311},year={2004},doi={10.1103/PhysRevA.69.062311}}

@article{Shi2025,author={Shi, Fei and Ning, Yu and Zhao, Qi and Zhang, Xiande},title={Bounds on {$k$}-Uniform Quantum States},journal={IEEE Trans. Inf. Theory},volume={71},pages={413--425},year={2025},doi={10.1109/TIT.2024.3481042}}

@misc{LiuShao2025,author={Liu, Xuhong and Shao, Shuai},title={Classification of Small-size Quantum Secret Sharing Schemes using Uniform States},year={2026},eprint={2510.07753v2},archivePrefix={arXiv},primaryClass={quant-ph},url={https://arxiv.org/abs/2510.07753v2}}

@article{Doroudiani2020,author={Doroudiani, Mehregan and Karimipour, Vahid},title={Planar maximally entangled states},journal={Phys. Rev. A},volume={102},pages={012427},year={2020},doi={10.1103/PhysRevA.102.012427}}

@article{Gour2010,author={Gour, Gilad and Wallach, Nolan R.},title={All maximally entangled four-qubit states},journal={J. Math. Phys.},volume={51},pages={112201},year={2010},doi={10.1063/1.3511477}}

@misc{DanielsenCodeDatabase,
  author = {Danielsen, Lars Eirik},
  title = {Database of Self-Dual Quantum Codes},
  year = {2005},
  url = {https://www.codetables.de/larsed/vncorbits/},
  note = {Connected graph representatives in graph6 format, including \texttt{vncorbits9.g6} and \texttt{vncorbits12.g6.gz}}
}

@article{Cabello2011,
  author = {Cabello, Ad{\'a}n and Danielsen, Lars Eirik and L{\'o}pez-Tarrida, Antonio J. and Portillo, Jos{\'e} R.},
  title = {Optimal preparation of graph states},
  journal = {Phys. Rev. A},
  volume = {83},
  pages = {042314},
  year = {2011},
  doi = {10.1103/PhysRevA.83.042314}
}

@misc{CabelloGraphData,
  author = {Cabello, Ad{\'a}n and Danielsen, Lars Eirik and L{\'o}pez-Tarrida, Antonio J. and Portillo, Jos{\'e} R.},
  title = {Database of Entanglement in Graph States},
  year = {2011},
  url = {https://www.codetables.de/larsed/entanglement/},
  note = {Nine-qubit data, file \texttt{entanglement9}}
}

@misc{SudevanDas2022,
  author = {Sudevan, Sowrabh and Das, Sourin},
  title = {$n$-qubit states with maximum entanglement across all bipartitions: A graph state approach},
  year = {2022},
  eprint = {2201.05622},
  archivePrefix = {arXiv},
  primaryClass = {quant-ph}
,url={https://arxiv.org/abs/2201.05622}}

@misc{Fattal2004,author={Fattal, David and Cubitt, Toby S. and Yamamoto, Yoshihisa and Bravyi, Sergey and Chuang, Isaac L.},title={Entanglement in the stabilizer formalism},eprint={quant-ph/0406168},archivePrefix={arXiv},year={2004},url={https://arxiv.org/abs/quant-ph/0406168}}

@article{Hanani1960,author={Hanani, Haim},title={On Quadruple Systems},journal={Can. J. Math.},volume={12},pages={145--157},year={1960},doi={10.4153/CJM-1960-013-3}}

@article{Qi2005,
 author={Qi, Liqun},
 title={Eigenvalues of a real supersymmetric tensor},
 journal={Journal of Symbolic Computation},
 volume={40}, pages={1302--1324}, year={2005},
 doi={10.1016/j.jsc.2005.05.007}
}

@article{CRSS1998,author={Calderbank, A. R. and Rains, E. M. and Shor, P. W. and Sloane, N. J. A.},title={Quantum error correction via codes over {GF}(4)},journal={IEEE Trans. Inf. Theory},volume={44},pages={1369--1387},year={1998},doi={10.1109/18.681315}}

@article{HaradaMunemasa2011,author={Harada, Masaaki and Munemasa, Akihiro},title={Classification of quaternary {Hermitian} self-dual codes of length 20},journal={IEEE Trans. Inf. Theory},volume={57},pages={3758--3762},year={2011},doi={10.1109/TIT.2011.2134330}}

@incollection{Gottesman2002,author={Gottesman, Daniel},title={An introduction to quantum error correction},booktitle={Quantum Computation: A Grand Mathematical Challenge for the Twenty-First Century and the Millennium},editor={Lomonaco, Jr., Samuel J.},series={Proceedings of Symposia in Applied Mathematics},volume={58},publisher={American Mathematical Society},address={Providence, RI},pages={221--235},year={2002},eprint={quant-ph/0004072},archivePrefix={arXiv},url={https://arxiv.org/abs/quant-ph/0004072}}

@misc{PolWalshWigal,author={van der Pol, Jorn and Walsh, Zach and Wigal, Michael C.},title={Tur{\'a}n densities for matroid basis hypergraphs},year={2025},eprint={2502.03673},archivePrefix={arXiv},primaryClass={math.CO},url={https://arxiv.org/abs/2502.03673}}

@article{EllisIvanLeader2024,author={Ellis, David and Ivan, Maria-Romina and Leader, Imre},title={Tur{\'{a}}n densities for daisies and hypercubes},journal={Bull. Lond. Math. Soc.},volume={56},pages={3838--3853},year={2024},doi={10.1112/blms.13171}}

@article{MotzkinStraus1965,author={Motzkin, T. S. and Straus, E. G.},title={Maxima for graphs and a new proof of a theorem of {Tur{\'{a}}n}},journal={Canad. J. Math.},volume={17},pages={533--540},year={1965},doi={10.4153/CJM-1965-053-6}}

@incollection{Keevash2011,author={Keevash, Peter},title={Hypergraph {Tur\'an} problems},booktitle={Surveys in Combinatorics 2011},editor={Chapman, Robin},series={London Mathematical Society Lecture Note Series},volume={392},pages={83--140},publisher={Cambridge University Press},address={Cambridge},year={2011},doi={10.1017/CBO9781139004114.004}}
\end{document}